\documentclass[a4paper,onecolumn,11pt,accepted=2024-11-14]{quantumarticle}
\pdfoutput=1
\usepackage[utf8]{inputenc}
\usepackage[english]{babel}
\usepackage[T1]{fontenc}
\usepackage[numbers,sort&compress]{natbib}
\usepackage{setspace,paralist}
\usepackage{amsmath, amssymb,amsfonts,xspace,graphicx,relsize,bm,mathtools,xcolor,stmaryrd}
\usepackage{subcaption}
\usepackage{mathrsfs}
\usepackage[pagebackref]{hyperref}
\usepackage{enumerate}
\usepackage{bm}
\usepackage{bbm}
\usepackage{enumitem}
\usepackage[normalem]{ulem}
\usepackage{multicol}
\usepackage{multirow}
\usepackage{environ}         % provides \BODY
\usepackage{etoolbox}        % provides \ifdimcomp
\usepackage{graphicx}        % provides \resizebox
\usepackage{makecell}

\newcommand\scalemath[2]{\scalebox{#1}{\mbox{\ensuremath{\displaystyle #2}}}}

\usepackage{amsthm}
\usepackage[capitalize,nameinlink,noabbrev]{cleveref}
\usepackage{lipsum} % for filler text
\newtheorem{theorem}{Theorem}
\newtheorem{result}[theorem]{Result}
\newtheorem{remark}[theorem]{Remark}

\newcommand{\GT}{\mathsf{GT}}
\newcommand{\QRAM}{\mathsf{QRAM}}
\newcommand{\QRAG}{\mathsf{QRAG}}

\newcommand{\ket}[1]{|#1\rangle}

\DeclareMathOperator{\poly}{poly}
\DeclareMathOperator{\supp}{supp}

\newtheorem{definition}[theorem]{Definition}
\newtheorem{fact}[theorem]{Fact}

\newtheorem{lemma}[theorem]{Lemma}

\def\01{\{0,1\}}

\usepackage{authblk}
\newcolumntype{C}[1]{>{\centering\let\newline\\\arraybackslash\hspace{0pt}}m{#1}}

\usepackage{algorithm,algpseudocode}

\definecolor{Ale}{rgb}{0.9,0.17,0.31}

\DeclareRobustCommand{\DE}[2]{#2}

\usepackage{geometry}
\usepackage[para]{footmisc}

\begin{document}

\title{Constant-depth circuits for Boolean functions and quantum memory devices using multi-qubit gates}

\author[1]{Jonathan Allcock}
%\affiliation{Centre for Quantum Technologies, National University of Singapore, Singapore}\thanks{\tiny }
\email{jonallcock@tencent.com}

\author[2]{Jinge Bao}
%\affiliation{Centre for Quantum Technologies, National University of Singapore, Singapore}\thanks{\tiny }
\email{jbao@u.nus.edu}

\author[3,2]{Joao F. Doriguello}
%\affiliation{Centre for Quantum Technologies, National University of Singapore, Singapore}\thanks{\tiny }
\email{doriguello@renyi.hu}
\homepage{www.joaodoriguello.com}
\orcid{0000-0002-8265-7334}

\author[2]{Alessandro Luongo}
%\affiliation{Centre for Quantum Technologies, National University of Singapore, Singapore}\thanks{\tiny }
\email{ale@nus.edu.sg}

\author[2,4]{Miklos Santha}
%\affiliation{Centre for Quantum Technologies, National University of Singapore, Singapore}\thanks{\tiny }
\email{cqtms@nus.edu.sg}

\affil[1]{Tencent Quantum Laboratory, Hong Kong, China}
\affil[2]{Centre for Quantum Technologies, National University of Singapore, Singapore}
\affil[3]{HUN-REN Alfréd Rényi Institute of Mathematics, Budapest, Hungary}
\affil[4]{CNRS, IRIF, Universit\'e Paris Cit\'e}

\date{}

% \author[1]{Jonathan Allcock\thanks{\scriptsize jonallcock@tencent.com}}
% \author[2]{Jinge Bao\thanks{\scriptsize jbao@u.nus.edu}}
% \author[2]{Jo{\~a}o F. Doriguello\thanks{\scriptsize joaofd@nus.edu.sg}}
% \author[2]{\\Alessandro Luongo\thanks{\scriptsize ale@nus.edu.sg}}
% \author[2,3]{Miklos Santha\thanks{\scriptsize cqtms@nus.edu.sg}}
% \affil[1]{Tencent Quantum Laboratory, Hong Kong, China}
% \affil[2]{Centre for Quantum Technologies, National University of Singapore, Singapore}
% \affil[3]{CNRS, IRIF, Universit\'e Paris Cit\'e}

% \date{\today}

\maketitle

\begin{abstract}
We explore the power of the unbounded Fan-Out gate and the Global Tunable gates generated by Ising-type Hamiltonians in constructing constant-depth quantum circuits, with particular attention to quantum memory devices. We propose two types of constant-depth constructions for implementing Uniformly Controlled Gates. These gates include the Fan-In gates defined by $|x\rangle|b\rangle\mapsto |x\rangle|b\oplus f(x)\rangle$ for $x\in\{0,1\}^n$ and $b\in\{0,1\}$, where $f$ is a Boolean function. The first of our constructions is based on computing the one-hot encoding of the control register $|x\rangle$, while the second is based on Boolean analysis and exploits different representations of $f$ such as its Fourier expansion. Via these constructions, we obtain constant-depth circuits for the quantum counterparts of read-only and read-write memory devices --- Quantum Random Access Memory ($\QRAM$) and Quantum Random Access Gate ($\QRAG$) --- of memory size $n$. The implementation based on one-hot encoding requires either $O(n\log^{(d)}{n}\log^{(d+1)}{n})$ ancillae and $O(n\log^{(d)}{n})$ Fan-Out gates or $O(n\log^{(d)}{n})$ ancillae and $16d-10$ Global Tunable gates, where $d$ is any positive integer and $\log^{(d)}{n} = \log\cdots \log{n}$ is the $d$-times iterated logarithm. On the other hand, the implementation based on Boolean analysis requires $8d-6$ Global Tunable gates at the expense of $O(n^{1/(1-2^{-d})})$ ancillae. 
\end{abstract}

\section{Introduction}

In this work, we study the power of constant-depth quantum circuits with a focus on circuits designed for quantum memory access and the execution of Boolean functions. Our investigation has two aims: firstly, to fill the theoretical gap in our understanding of quantum memory circuits from a computational complexity perspective and, secondly, to assess the practicality of physically implementing these circuits. We believe that the properties and limitations of these circuits can highlight their feasibility and potential for practical implementations. To obtain constant-depth circuits, we leverage multi-qubit ``magic'' gates like the Fan-Out gate (a generalization of the $\mathsf{CNOT}$ that can target multiple output qubits) and the multi-qubit entangling Global Tunable gate (that arises from the time evolution of Ising-type Hamiltonians). 
This analysis explores the potential for quantum memory to be accessed using specialized hardware (designed, for instance, to implement such multi-qubit gates), which may differ from the hardware in general-purpose quantum computers. 

\subsection{Quantum memory}

Quantum memory, besides being an important component of quantum computers from a theoretical perspective, is also fundamental to many quantum algorithms such as Grover's search~\cite{grover1997quantum}, solving the dihedral hidden subgroup problem~\cite{kuperberg2005subexponential}, collision finding~\cite{brassard1997quantum}, phase estimation for quantum chemistry~\cite{babbush2018encoding}, pattern recognition and machine learning algorithms~\cite{trugenberger2002phase,schutzhold2003pattern,schaller2006quantum,kerenidis2017quantum,kerenidis2020quantum},  cryptanalysis~\cite{chailloux2021lattice}, state preparation~\cite{grover2002creating}, among others.

Traditionally, there are two ways — via a Quantum Random Access Memory ($\QRAM$) or a Quantum Random Access Gate ($\QRAG$) — in which memory (classical or quantum) may be accessed quantumly. A $\QRAM$ can be seen as a ``read-only'' gate, while a $\QRAG$ can be interpreted as a ``read-write'' gate since qubits are swapped from memory into the main part of the quantum computer, acted on, and then swapped back.

\paragraph{Quantum random access memory.}  
A $\QRAM$~\cite{giovannetti2008architectures,giovannetti2008quantum} is the quantum analogue of a classical Random Access Memory (RAM) device that stores classical or quantum data and allows queries to be performed in superposition. More specifically, a  $\QRAM$ is a device comprising a memory register $\mathtt{M}$ that stores either classical or quantum information, an address register $\mathtt{A}$ that points to the memory cell to be addressed, and a target register $\mathtt{T}$ into which the content of the addressed memory cell is copied. If necessary, it also includes an auxiliary register supporting the overall operation, which is reset to its initial state at the end of the computation. A call to a $\QRAM$ (of size $n$) implements\footnote{Define $[n]:= \{0,\dots,n-1\}$.}
\begin{equation*}
   \ket{i}_{\mathtt{A}}\ket{b}_{\mathtt{T}}\ket{x_0, \dots, x_{n-1}}_{\mathtt{M}} \mapsto \ket{i}_{\mathtt{A}}\ket{b\oplus x_i}_{\mathtt{T}}\ket{x_0, \dots, x_{n-1}}_{\mathtt{M}}, \quad \forall x_0,\dots,x_{n-1},b \in \{0,1\}, i\in[n].
\end{equation*}
The bits $x_0,\dots,x_{n-1}$ represent the data to be accessed in superposition, which are separate from the qubits in the \emph{work register} of a fully programmable quantum computer. 

\paragraph{Quantum random access gate.}
Another device for random access to a quantum memory is the so-called $\QRAG$, which performs a swap gate between the target register and some portion of the memory register specified by the address register:
\begin{equation*}
    \ket{i}_{\mathtt{A}}\ket{b}_{\mathtt{T}}\ket{x_0, \dots, x_{n-1}}_{\mathtt{M}} \mapsto \ket{i}_{\mathtt{A}}\ket{x_i}_{\mathtt{T}}\ket{x_0, \dots, x_{i-1},b,x_{i+1}, \dots, x_{n-1}}_{\mathtt{M}}, \quad\forall x_0,\dots,x_{n-1},b \in \{0,1\}, i\in[n]. 
\end{equation*}
While $\QRAG$ does not enjoy the same level of publicity as $\QRAM$s, its importance lies in its necessity for quantum algorithms for element distinctness and collision finding~\cite{ambainis2007quantum}, as well as other quantum algorithms based on random walks on graphs~\cite{aaronson2019quantum,buhrman2021limits}.

\subsection{Multi-qubit ``magic'' gates}\label{sec:multimagic}

\paragraph{Uniformly Controlled Gate and Fan-In gate.} The $f$-Uniformly Controlled Gate ($f$-$\mathsf{UCG}$ or simply $\mathsf{UCG}$) is the unitary $\sum_{x\in\{0,1\}^n} |x\rangle\langle x|\otimes f(x)$, where $f:\{0,1\}^n\to\mathcal{U}(\mathbb{C}^{2\times 2})$ is a mapping from $n$-bit strings onto the set $\mathcal{U}(\mathbb{C}^{2\times 2})$ of single-qubit unitaries. $\mathsf{UCG}$s are also known as $\mathsf{SELECT}$ operators~\cite{Low2024tradingtgatesdirty}. Well-known examples of $f$-$\mathsf{UCG}$s can be found in quantum state preparation algorithms~\cite{grover2002creating,kerenidis2017quantum}, quantum Monte Carlo algorithms~\cite{montanaro2015quantum} in finance, and HHL-like algorithms~\cite{harrow2009quantum} in quantum machine learning. The $f$-$\mathsf{UCG}$ is a generalization of many 
multi-qubit gates including the $f$-Fan-In gate ($f$-$\mathsf{FIN}$) defined by the mapping $|x\rangle|b\rangle \mapsto |x\rangle|b\oplus f(x)\rangle$ for a Boolean function $f:\{0,1\}^n\to\{0,1\}$, where $x\in\{0,1\}^n$ and $b\in\{0,1\}$. Note that an $f$-$\mathsf{FIN}$ is simply an $f'$-$\mathsf{UCG}$ with $f'(x) = \mathsf{X}^{f(x)}$. Special cases of $f$-$\mathsf{FIN}$s include $\mathsf{OR}$, $\mathsf{AND}$, $\mathsf{PARITY}$, $\mathsf{MAJORITY}$, and even $\QRAM$, since it can be implemented with $f:\{0,1\}^n\times[n]\to\{0,1\}$, $f(x,i) = x_i$. 

General constructions of $f$-$\mathsf{UCG}$s and $f$-$\mathsf{FIN}$s using single and two-qubit gates can be framed as a unitary synthesis problem. There are several results in this direction for constructing a general $n$-qubit unitary~\cite{barenco1995elementary,knill1995approximation,vartiainen2004efficient,mottonen2005decompositions,shende2004minimal,rosenthal2021query}. Sun et al.~\cite{STY-asymptotically}, Yuan and Zhang~\cite{yuan2023optimal}, and Low et al.~\cite{Low2024tradingtgatesdirty} proposed circuits specifically for $f$-$\mathsf{UCG}$s using one and two-qubit gates, and therefore not in constant depth. Regarding constructions for controlled gates of the form $|x\rangle\langle x|\otimes \mathsf{U} + \sum_{y\in\{0,1\}^n\setminus\{x\}}|y\rangle\langle y|\otimes \mathbb{I}_m$, where $\mathsf{U}$ is an $m$-qubit gate, see~\cite{barenco1995elementary} (using one and two-qubit gates) and~\cite{green2001counting,hoyer2005quantum,martinez2016compiling,gokhale2021quantum} (using multi-qubit entangling gates defined below). While general sequential implementations for $f$-$\mathsf{FIN}$s are folklore, there have been proposals for specific Boolean functions~\cite{barnum2000quantum} or based on different models of computation like measurement-based quantum computation~\cite{daniel2022quantum}. See \Cref{table:known_results} for a summary of known results.

\paragraph{The Fan-Out gate.} The Fan-Out ($\mathsf{FO}$) gate on $n+1$ qubits implements the quantum operation $|b\rangle|x_0,\dots,x_{n-1}\rangle \mapsto |b\rangle|x_0\oplus b,\dots,x_{n-1}\oplus b\rangle$ for all $x_0,\dots,x_{n-1},b\in\{0,1\}$. In other words, it is a sequence of $\mathsf{CNOT}$ gates sharing a single control qubit. For this reason, unlike classical Fan-Out gates, the ability to implement quantum Fan-Out gates as a primitive is not usually taken for granted. Indeed, the Fan-Out gate is powerful in the sense that several interesting results follow from its use, especially connected to constant-depth complexity classes (more on this below). Moore~\cite{moore1999quantum} and Green et al.~\cite{green2001counting} proved that Fan-Out is equivalent to the $\mathsf{PARITY}$ gate. H\o{}yer and \v{S}palek~\cite{hoyer2005quantum} proved that $\mathsf{EXACT}[t]$ gates (which output $1$ if the input's Hamming weight is $t$ and $0$ otherwise) can be approximated with a polynomially small error by Fan-Out and single-qubit gates in constant depth. These in turn can simulate $\mathsf{AND},\mathsf{OR}$, and $\mathsf{THRESHOLD}[t]$ gates. Later, Takahashi and Tani~\cite{takahashi2016collapse} managed to prove that $\mathsf{EXACT}[t]$ can be simulated \emph{exactly} by Fan-Out and single-qubit gates in constant depth. See \Cref{table:known_results} for a summary of known results.

Unbounded Fan-Out gates that can act on any number of qubits are used in quantum complexity theory (and in this work) to compile certain circuits in constant depth.
Even though unbounded Fan-Out gates are just a theoretical construction, bounded Fan-Out gates are within the reach of next-generation quantum hardware~\cite{fenner2003implementing,zeng2005measuring,rasmussen2020single,yu2020scalability,kim2022high,guo2022implementing,fenner2022implementing} and can serve as building blocks in larger Fan-Out gates, since an $n$-arity Fan-Out gate can be simulated by $k$-arity Fan-Out gates in $O(\log_k{n})$-depth, offering interesting trade-offs for hardware implementations. Another approach to implementing Fan-Out gates is the work of Pham and Svore~\cite{pham20132d}, who proposed a constant-depth circuit for Fan-Out gates using $O(n)$ ancillae based on measurement-based quantum computation and classical feedback.

\paragraph{The Global Tunable gate.} Another powerful and physically implementable gate is the Global Tunable ($\mathsf{GT}$) gate. In its simplest form, it implements a product of two-qubit controlled-$\mathsf{Z}$ gates:
\begin{align*}
   \prod_{i\neq j\in S}\mathsf{C}_i\text{-}\mathsf{Z}_{\to j}
\end{align*}
for some subset $S$ of the physical qubits, where $\mathsf{C}_i$-$\mathsf{Z}_{\to j}$ denotes a $\mathsf{Z}$ gate applied to qubit $j$ controlled on qubit $i$ being in the $|1\rangle$ state (for the general definition see \Cref{sec:GT}).
The first proposal for this kind of gate dates back to Mølmer and Sørensen~\cite{molmer1999multiparticle}, and several experimental implementations have been reported~\cite{nigg2014quantum,landsman2019two,figgatt2019parallel,grzesiak2020efficient,gu2021fast}.

A few studies have explored the use of $\mathsf{GT}$ gates in constructing $n$-qubit Clifford gates~\cite{maslov2018use,van2021constructing,grzesiak2022efficient,bassler2022synthesis,maslov2022depth}. The state-of-the-art construction~\cite{bravyi2022constant} requires $4$ $\GT$ gates and $n$ ancilla or $26$ $\GT$ gates and no ancilla, plus $O(n)$ single-qubit gates. Similarly to the Fan-Out~\cite{hoyer2005quantum,takahashi2016collapse}, the $\GT$ gate has been used to implement the unbounded $\mathsf{OR}$ gate.
Constructions for $4$-$\mathsf{AND}$ gates using $7$ $\mathsf{GT}$ gates and no ancillae were reported in~\cite{ivanov2015efficient,martinez2016compiling,maslov2018use}. Regarding general $n$-arity $\mathsf{AND}$, several constructions~\cite{maslov2018use,groenland2020signal,grzesiak2022efficient} have been proposed, and improved to the state-of-the-art implementation of~\cite{bravyi2022constant} using $O(\log^\ast{n})$ $\GT$ gates with $O(\log{n})$ ancillae or using $4$ $\GT$ gates with $O(n)$ ancillae, where $\log^\ast{n}$ is the star-log function. See \Cref{table:known_results} for a summary of known results.

\subsection{Our results}

In this work, we propose new constant-depth quantum circuits, based on Fan-Out and $\mathsf{GT}$ gates, for $f$-$\mathsf{UCG}$s, which include $f$-$\mathsf{FIN}$s and certain quantum memory devices as special cases (see \Cref{fig:ucg_qmd_venn}). We use two different techniques: the first based on the one-hot encoding of the input (also known as indicator function~\cite{Low2024tradingtgatesdirty}, see definition below), and the second based on Boolean analysis of the function $f$. In \Cref{ssec:memory}, we formalize our model of quantum computers with quantum access to memory. A Quantum Memory Device ($\mathsf{QMD}$) of size $n$ (assume $n$ to be a power of $2$) comprises a $\log{n}$-qubit address register $\mathtt{A}$, a single-qubit target register $\mathtt{T}$, a $\operatorname{poly}(n)$-qubit auxiliary register $\mathtt{Aux}$, and an $n$-qubit memory $\mathtt{M}$ consisting of $n$ single-qubit registers $\mathtt{M}_0,\dots,\mathtt{M}_{n-1}$. A call to the $\mathsf{QMD}$ implements
\begin{align*}
    |i\rangle_{\mathtt{A}}|b\rangle_{\mathtt{T}}|x_i\rangle_{\mathtt{M}_i}|0\rangle^{\otimes\poly{n}}_{\mathtt{Aux}} \mapsto  |i\rangle_{\mathtt{A}}\mathsf{V}(i)\big(|b\rangle_{\mathtt{T}}|x_i\rangle_{\mathtt{M}_i}\big)|0\rangle^{\otimes\poly{n}}_{\mathtt{Aux}}, \quad\text{where}~\mathsf{V}:[n]\to \mathcal{V} ~\text{and}~ \mathcal{V}\subset \mathcal{U}(\mathbb{C}^{4\times 4})
\end{align*}
is a $O(1)$-size subset of two-qubit gates. Our model is general enough to include $\QRAM$ and $\QRAG$ as subcases (by letting $\mathsf{V}(i)$ equal $\mathsf{CNOT}$ or $\mathsf{SWAP}$ gates). It also includes $\mathsf{QMD}$s that we named $f$-$\QRAM$ for which $\mathsf{V}(i)=\mathbb{I}_1\otimes |0\rangle\langle 0|_{\mathtt{M}_i} + f(i)\otimes |1\rangle\langle 1|_{\mathtt{M}_i}$, where $f:\{0,1\}^{\log{n}}\to\mathcal{U}(\mathbb{C}^{2\times 2})$. As far as we know, a general model for a $\mathsf{QMD}$ has not been formally defined before. This model allows us to compare the power of different gates with quantum access to memory. In this direction, we show that a $\QRAG$ can simulate a $\QRAM$, but not vice-versa, and we discuss the similarities and differences between $\mathsf{QMD}$ and $f$-$\mathsf{UCG}$. In particular, even though $f$-$\mathsf{UCG}$s do not contain general $\mathsf{QMD}$s, since $\mathsf{V}(i)$ can act non-trivially on two qubits, an $f$-$\mathsf{QRAM}$ of memory size $n$ (i.e., $f:\{0,1\}^{\log{n}}\to\mathcal{U}(\mathbb{C}^{2\times 2})$) can be seen as an $f'$-$\mathsf{UCG}$ for some function $f'$ on $\{0,1\}^{n+\log{n}}$ (see \Cref{fig:ucg_qmd_venn}), since $\mathtt{A}$ and $\mathtt{M}$ work as control registers.

\begin{figure}[htb!]
    \centering
    \includegraphics[width=0.57\textwidth]{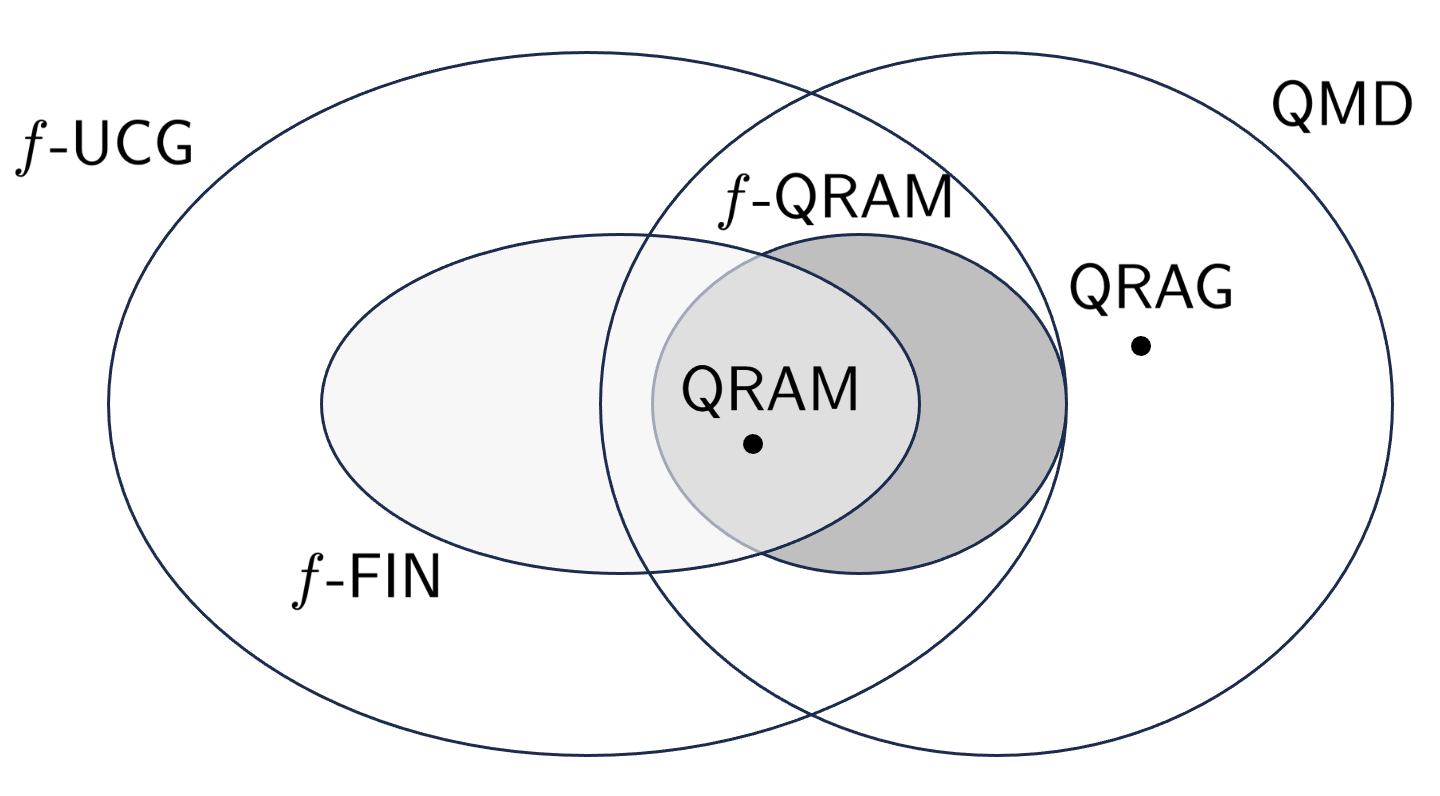}
    \caption{We give constant-depth circuits for $f$-$\mathsf{UCG}$s, which contain $f$-$\mathsf{FIN}$s and a subset of quantum memory devices ($\mathsf{QMD}$) including $\QRAM$ (and its generalization we call $f$-$\QRAM$) as special cases. %A refined analysis gives improved constructions for $f$-$\mathsf{FIN}$s compared with general $f$-$\mathsf{UCG}$s. 
    Although $\QRAG$ is not an $f$-$\mathsf{FIN}$, our (one-hot-encoding-based) construction for $\mathsf{QRAM}$ can be adapted to it.}
    \label{fig:ucg_qmd_venn}
\end{figure}

In \Cref{sec:multi-qubit_gates}, we discuss the Fan-Out and $\mathsf{GT}$ gates in more detail. In \Cref{sec:constructions_onehot}, we develop our quantum circuits based on one-hot encoding for any $f$-$\mathsf{UCG}$, $\sum_{x\in\{0,1\}^n} |x\rangle\langle x|\otimes f(x)$. The main idea is to use Fan-Out or $\mathsf{GT}$ gates to compute, in parallel, the one-hot encoding $e(x)\in\{0,1\}^{2^n}$ of the control register $|x\rangle$, where $e(x)_j = 1$ if and only if $j=x$, and to apply the single-qubit gate $f(j)$ controlled on the qubit $|e(x)_j\rangle$, for all $j\in\{0,1\}^{n}$. By the definition of the one-hot encoding, the correct gate $f(x)$ is selected. To perform all controlled single-qubit gates $f(j)$ in parallel, we use the well-known $\mathsf{Z}$-decomposition of single-qubit gates stating the existence of functions $\alpha,\beta,\gamma,\delta:\{0,1\}^n\to[-1,1]$ such that
\begin{align*}
   f(j) = e^{i\pi\alpha(j)}\mathsf{Z}(\beta(j))\mathsf{H}\mathsf{Z}(\gamma(j))\mathsf{H}\mathsf{Z}(\delta(j)), \quad \text{for all}~j\in\{0,1\}^n,
\end{align*}
where $\mathsf{H} := \frac{1}{\sqrt{2}}\bigl(\begin{smallmatrix}
1&1 \\ 1&-1 \end{smallmatrix} \bigr)$ and $\mathsf{Z}(\theta) := \bigl(\begin{smallmatrix} 1&0 \\ 0&e^{i\pi\theta} \end{smallmatrix} \bigr)$ for $\theta\in[-1,1]$. By a result of Green et al.~\cite{green2001counting} (see also~\cite{moore2001parallel,hoyer2005quantum}), $m$ commuting gates can be performed in parallel with the aid of $m-1$ ancillae and $2$ Fan-Out gates, or simply 1 $\mathsf{GT}$ gate and no ancillae. All the $\mathsf{Z}(\delta(j))$ gates can thus be performed in parallel (and similarly for $\mathsf{Z}(\gamma(j))$, $\mathsf{Z}(\beta(j))$, $e^{i\pi \alpha(j)}$).

Naively, one can compute the one-hot encoding of the whole input $x$. However, if $f$ is a junta, i.e., it depends on only a few coordinates, one only needs to compute the one-hot encoding of the coordinates on which it depends. More generally, this ``compression'' idea can be extended to a concept we introduce and call $(J,r)$-junta, where $J\subseteq[n]$ and $r\in\mathbb{N}$. We say $f:\{0,1\}^n \to \mathcal{U}(\mathbb{C}^{2\times 2})$ is a $(J,r)$-junta if, by fixing the coordinates in $\overline{J} := [n]\setminus J$ to any value, the resulting restriction of $f$ to $J$ is an $r$-junta, i.e., it depends on at most $r$ of its input coordinates. A fine example of a $(J,r)$-junta is $\QRAM$, since by fixing the coordinates of input $i$, the resulting restriction is a $1$-junta (as it depends only on $x_i$). It is possible to take advantage of this property and simplify our circuit construction: we partition the input $x$ into sub-strings $x_{\overline{J}}$ and $x_{J}$ and compute the one-hot encoding of $x_{\overline{J}}$ separately from the one-hot encoding of the coordinates in $J$ that the restriction of $f$ depends on. Both one-hot encodings are then used to select the correct $f(x)$ gate as described above. The resources required for our constructions are as follows (see also \Cref{table:ucg}).
\begin{result}[Informal version of \Cref{thr:ucg_construction}]\label{res:res1}
    Let $f:\{0,1\}^n\to\mathcal{U}(\mathbb{C}^{2\times 2})$ be a $(J,r)$-junta, $|\overline{{J}}| = t$. The $f$-$\mathsf{UCG}$ can be implemented in constant depth using either $O(2^{t+r}(t+r)\log(t+r))$ ancillae and $O(2^{t+r}(t+r))$ Fan-Out gates or $O(2^{t+r}(t+r))$ ancillae and $9$ $\mathsf{GT}$ gates. As a corollary, any $f'$-$\mathsf{QRAM}$ of size $n$, $f':\{0,1\}^{\log{n}}\to\mathcal{U}(\mathbb{C}^{2\times 2})$, can be implemented in constant depth using either $O(n\log{n}\log\log{n})$ ancillae and $O(n\log{n})$ Fan-Out gates or $O(n\log{n})$ ancillae and $9$ $\mathsf{GT}$ gates.
\end{result}

We then tailor \Cref{res:res1} to $f$-$\mathsf{FIN}$s specifically, given their simpler structure compared to $f$-$\mathsf{UCG}$s. The number of ancillae and Fan-Out gates are asymptotically the same, and the number of $\mathsf{GT}$ gates is reduced to $6$ (see \Cref{table:fan-in}). In particular, we apply the $f$-$\mathsf{FIN}$ results to $\QRAM$s and also show how to implement a $\QRAG$ in constant depth\footnote{A $O(1)$-depth and $O(n\log{n}\log\log{n})$-size circuit for $\QRAG$ using Fan-Out gates had previously appeared in~\cite[Lemma~4.3]{rosenthal2021query}. We note that the author missed the $\log\log{n}$-factor.}, even though it is not an $f$-$\mathsf{FIN}$ (see \Cref{table:qram-qrag}). In the following, $\log^{(d)}{n} = \log\cdots\log{n}$ is the $d$-times iterated logarithm.

\begin{result}[Informal version of \Cref{thr:qram_recursive_procedure}]
     Let $d\in\mathbb{N}$ be a constant. A $\mathsf{QRAM}$ of size $n$ can be implemented in $O(d)$-depth using either $O(n\log^{(d)}{n}\log^{(d+1)}{n})$ ancillae and $O(n\log^{(d)}{n})$ Fan-Out gates or $O(n\log^{(d)}{n})$ ancillae and $16d-10$ $\mathsf{GT}$ gates. A $\mathsf{QRAG}$ of size $n$ can be implemented in $O(d)$-depth using either $O(n\log^{(d)}{n}\log^{(d+1)}{n})$ ancillae and $O(n\log^{(d)}{n})$ Fan-Out gates or $O(n\log^{(d)}{n})$ ancillae and $21d-12$ $\mathsf{GT}$ gates.
\end{result}

In \Cref{sec:fourier}, we extend ideas from~\cite{hoyer2005quantum,takahashi2016collapse} to implement $f$-$\mathsf{UCG}$s in constant depth using tools from the analysis of Boolean functions. We give three slightly different constructions based on different representations of a real-valued Boolean function $g:\{0,1\}^n\to\mathbb{R}$. The first representation is the Fourier expansion (over the reals)
\begin{align*}
    g(x) = \sum_{S\subseteq[n]} \widehat{g}(S) \chi_S(x), \quad\text{where}~\chi_S(x) := (-1)^{\sum_{i\in S}x_i} ~~\text{and}~~ \widehat{g}(S) = \frac{1}{2^n}\sum_{x\in\{0,1\}^n} g(x)\chi_S(x)
\end{align*}
are the Fourier coefficients of $g$. The $\mathsf{PARITY}$ function $\chi_S$ over $\{-1,1\}$ is called characteristic function. The second representation is based on the existence of a function $p:\{0,1\}^n\to\mathbb{R}$ with a (potentially) sparse Fourier expansion that approximates $g$ up to an additive error $\epsilon>0$, i.e., $\max_{x\in\{0,1\}^n}|p(x) - g(x)| \leq \epsilon$. Finally, the third representation is the Fourier expansion of $g$ using $\mathsf{AND}$ functions instead of $\mathsf{PARITY}$ functions, which is sometimes called a real-polynomial representation over $\{0,1\}$,
\begin{align*}
    g(x) = \sum_{S\subseteq[n]}\widetilde{g}(S) x^S, \quad\text{where}~ x^S := \prod_{i\in S}x_i ~~\text{and}~~ \widetilde{g}(S) = \sum_{T\subseteq S}(-1)^{|S| - |T|}g(T).
\end{align*}
In the case of Boolean functions $g:\{0,1\}^n\to\{0,1\}$, the above representation over the reals can be ``compressed'' into a representation over the $\mathbb{F}_2$ field, also known as algebraic normal form, as
\begin{align*}
    g(x) = \bigoplus_{S\subseteq[n]}\widetilde{g}_{\mathbb{F}_2}(S) x^S, \quad\text{where}~ \widetilde{g}_{\mathbb{F}_2}(S)\in\{0,1\} ~\text{is given by}~ \widetilde{g}_{\mathbb{F}_2}(S) = \widetilde{g}(S)~(\text{mod}~2).
\end{align*}
Such relation is true since $\widetilde{g}(S)\in\mathbb{Z}$ for $g:\{0,1\}^n\to\{0,1\}$. The utility of each of the above representations depends on the Boolean properties of $g$, e.g., its Fourier support $\operatorname{supp}(g) := \{S\subseteq[n]: \widehat{g}(S) \neq 0\}$, (real) $\{0,1\}$-support $\operatorname{supp}_{\{0,1\}}(g) := \{S\subseteq[n]: \widetilde{g}(S) \neq 0\}$, and, for Boolean functions $g:\{0,1\}^n\to\{0,1\}$, its $\mathbb{F}_2$-support $\operatorname{supp}_{\mathbb{F}_2}(g) := \{S\subseteq[n]:\widetilde{g}_{\mathbb{F}_2}(S)\neq 0\}$. Other relevant properties of $g$ are its Fourier support $\operatorname{supp}^{>k}(g) := \{S\subseteq[n]:|S|>k,\widehat{g}(S)\neq 0\}$ at degree greater than $k$ (similarly for $\operatorname{supp}^{=k}(g)$, $\operatorname{supp}^{>k}_{\{0,1\}}(g)$, and $\operatorname{supp}^{>k}_{\mathbb{F}_2}(g)$), its real degree $\operatorname{deg}(g) := \max\{|S|:S\in\operatorname{supp}(g)\}$, its Fourier $1$-norm $\hat{\|}g\hat{\|}_1 := \sum_{S\subseteq[n]}|\widehat{g}(S)|$, and $\hat{\|}g^{>k}\hat{\|}_1 := \sum_{S\subseteq[n]:|S|>k}|\widehat{g}(S)|$.

We can generalize the above properties to operator-valued functions $f:\{0,1\}^n\to\mathcal{U}(\mathbb{C}^{2\times 2})$ in an indirect way by applying Boolean analysis to the functions $\alpha,\beta,\gamma,\delta:\{0,1\}^n\to[-1,1]$ arising from $f$'s $\mathsf{Z}$-decomposition and defining, for instance, $\operatorname{supp}(f) := \operatorname{supp}(\alpha)\cup \operatorname{supp}(\beta)\cup \operatorname{supp}(\gamma)\cup \operatorname{supp}(\delta)$ and $\operatorname{deg}(f) := \max\{\operatorname{deg}(\alpha),\operatorname{deg}(\beta),\operatorname{deg}(\gamma),\operatorname{deg}(\delta)\}$. Similar definitions apply to $\operatorname{supp}^{>k}(f)$, $\operatorname{supp}^{=k}(f)$, $\operatorname{supp}_{\{0,1\}}(f)$, and $\operatorname{supp}^{> k}_{\{0,1\}}(f)$. Note that other extensions of Boolean analysis exist in the literature and had been applied to problems in quantum computation~\cite{nayak2000quantum,fehr2008randomness,ben2008hypercontractive,montanaro2008quantum,rouze2022quantum}. However, this extension based on $\mathsf{Z}$-decomposition may be of independent interest.

The idea behind our constructions for $f$-$\mathsf{UCG}$s in \Cref{sec:fourier} is to reconstruct the functions $\alpha,\beta,\gamma,\delta$ using one of the aforementioned representations. Consider the Fourier expansion of $\alpha,\beta,\gamma,\delta$ as an example. First we compute the terms $\chi_S(x)$ in parallel using Fan-Out or $\mathsf{GT}$ gates, since $\chi_S(x)$ are $\mathsf{PARITY}$ functions. Since $\prod_{S\in\operatorname{supp}(\delta)}\mathsf{Z}(\widehat{\delta}(S)\chi_S(x)) = \mathsf{Z}(\sum_{S\in\operatorname{supp}(\delta)}\widehat{\delta}(S)\chi_S(x)) = \mathsf{Z}(\delta(x))$, it is possible to apply $\mathsf{Z}(\delta(x))$ onto a target qubit by simply applying onto this target qubit a sequence of phases $\mathsf{Z}(\widehat{\delta}(S))$ controlled on $\chi_S(x)$, for $S\in\operatorname{supp}(\delta)$. This sequence of controlled phases $\prod_{S\in\operatorname{supp}(\delta)}\mathsf{Z}(\widehat{\delta}(S)\chi_S(x))$ can be performed in constant depth in the case of $\mathsf{GT}$ gates by definition. In the case of Fan-Outs, it can be done by using techniques from H\o{}yer and \v{S}palek~\cite{hoyer2005quantum}. More precisely, first compute a cat state $(|0\rangle^{\otimes m} + |1\rangle^{\otimes m})/\sqrt{2}$ from the target qubit using one Fan-Out, where $m:=|\operatorname{supp}(\delta)|$, followed by applying the controlled phases $\mathsf{Z}(\widehat{\delta}(S))$ onto \emph{different} qubits of the cat state. This yields $(|0\rangle^{\otimes m} + (-1)^{\sum_S \widehat{\delta}(S)\chi_S(x)}|1\rangle^{\otimes m})/\sqrt{2} = \mathsf{Z}(\delta(x))(|0\rangle^{\otimes m} + |1\rangle^{\otimes m})/\sqrt{2}$. Finally, uncompute the cat state with another Fan-Out. The same idea applies to $\alpha,\beta,\gamma$ and the other two representations (for the real $\{0,1\}$-representation we compute $x^S$ instead of $\chi_S(x)$). The resources required for our constructions are stated below (see \Cref{table:ucg}). In the following, we say that a quantum circuit implements an $f$-$\mathsf{UCG}$ with spectral norm error at most $\epsilon$ if it implements an $f'$-$\mathsf{UCG}$ such that the spectral norm $\|f'(x) - f(x)\|$ is at most $\epsilon$ for all $x\in\{0,1\}^n$.
\begin{result}[Informal version of \Cref{thr:ucg_boolean_construction1}, \Cref{thr:ucg_boolean_construction2}, \Cref{thr:ucg_boolean_construction3}]
    Let $f:\{0,1\}^n\to\mathcal{U}(\mathbb{C}^{2\times 2})$ with $\mathsf{Z}$-decomposition $\alpha,\beta,\gamma,\delta:\{0,1\}^n\to[-1,1]$. We propose constant-depth quantum circuits that implement $f$-$\mathsf{UCG}$
    \begin{itemize}
      \item exactly using 
      \begin{itemize} 
        \item either $O\big(\sum_{S\in\operatorname{supp}(f)}|S|\big)$ ancillae and $O\big(|\operatorname{supp}^{>1}(f)|+\big|\bigcup_{S\in\operatorname{supp}^{>1}(f)}S\big|\big)$ Fan-Outs, 
        \item or $O(|\operatorname{supp}^{>1}(f)|)$ ancillae and $5$ $\mathsf{GT}$ gates;
      \end{itemize}
      \item with spectral norm error at most $\epsilon>0$ using
      \begin{itemize} 
        \item either $O(s\operatorname{deg}(f) + |\operatorname{supp}^{=1}(f)|)$ ancillae and $O(s + |\bigcup_{S\in\operatorname{supp}^{>1}(f)}S|)$ Fan-Outs, 
        \item or $O(s)$ ancillae and $5$ $\mathsf{GT}$ gates,
      \end{itemize}
      where $s:=(n/\epsilon^2)\sum_{\nu\in\{\alpha,\beta,\gamma,\delta\}}\hat{\|}\nu^{>1}\hat{\|}_1^2$;
      \item exactly using
      \begin{itemize}
        \item either $O\big(\sum_{S\in\operatorname{supp}_{\{0,1\}}(f)} |S|\log(1+|S|)\big)$ ancillae and $O\big(\sum_{S\in\operatorname{supp}^{>1}_{\{0,1\}}(f)}|S|\big)$ Fan-Outs,
        \item or $O\big(\sum_{S \in \operatorname{supp}^{>1}_{\{0,1\}}(f)}|S|\big)$ ancillae and $9$ $\mathsf{GT}$ gates.
      \end{itemize}
  \end{itemize}
\end{result}

Similarly to the one-hot-encoding-based constructions, we then simplify our Boolean-based constructions to $f$-$\mathsf{FIN}$s, which mainly reduces the number of $\mathsf{GT}$ gates (see \Cref{table:fan-in}), and apply them to $\QRAM$s, thus showing that it is possible to use fewer $\mathsf{GT}$ gates at the price of more ancillary qubits (see \Cref{table:qram-qrag}). We say that a quantum circuit implements an $f$-$\mathsf{FIN}$ with spectral norm error at most $\epsilon$ if it implements an $f'$-$\mathsf{UCG}$ such that $\max_{x\in\{0,1\}^n}\|f'(x) - \mathsf{X}^{f(x)}\| \leq \epsilon$.
\begin{result}[Informal version of \Cref{thr:qram_recursive_procedure_boolean}]
    Let $d\in\mathbb{N}$ be a constant. A $\QRAM$ of size $n$ can be implemented in $O(d)$-depth using either $O\big(n^{1/(1-2^{-d})}\log{n}\big)$ ancillae and $O\big(n^{1/(1-2^{-d})}\big)$ Fan-Out gates or $O\big(n^{1/(1-2^{-d})}\big)$ ancillae and $8d-6$ $\mathsf{GT}$ gates.
\end{result}

Depending on the properties of $f$, one construction can be more desirable compared to the others when it comes to implementing an $f$-$\mathsf{UCG}$ (similarly for $f$-$\mathsf{FIN}$s). A $(J,r)$-junta for small $|\overline{J}|$ and $r$ might call for a one-hot-encoding-based construction, while a function with sparse Fourier expansion could be more easily implementable using a Boolean-based circuit. The four different constructions presented above are thus incomparable. Nonetheless, in the worst case, the Boolean-based implementation using the Fourier expansion (\Cref{thr:ucg_boolean_construction1} and \Cref{thr:fin_boolean_construction1}) requires fewer resources: either $O(2^n)$ Fan-Out gates and $O(2^n n)$ ancillae, or $5$ $\mathsf{GT}$ gates and $O(2^n)$ ancillae. %Compare this with the worst case for the one-hot-encoding-based construction, which requires either $O(2^n n)$ Fan-Out gates and $O(2^n n\log{n})$ ancillae, or $9$ $\mathsf{GT}$ gates and $O(2^n n)$ ancillae.

%Finally, for convenience we summarise some of the results in the literature in \Cref{table:known_results}.

\begin{result}
    Any $f$-$\mathsf{UCG}$ with $f:\{0,1\}^n\to\mathcal{U}(\mathbb{C}^{2\times 2})$ can be implemented in constant depth using either $O(2^n n)$ ancillae and $O(2^n)$ Fan-Out gates, or $O(2^n)$ ancillae and $5$ $\mathsf{GT}$ gates.
\end{result}

\begin{table}[H]
\centering
\resizebox{\linewidth}{!}{
\begin{tabular}{|c|cc|cc|}
\hline
\multirow{2}{*}{Result} & \multicolumn{2}{c|}{$O(1)$-depth Fan-Out construction} & \multicolumn{2}{c|}{$O(1)$-depth $\mathsf{GT}$ construction} \\ 
\cline{2-5} 
& \multicolumn{1}{c|}{\#Fan-Out} & \#Ancillae & \multicolumn{1}{c|}{\#$\mathsf{GT}$} & \#Ancillae \\ \hline 
\makecell[c]{$f$-$\mathsf{UCG}$ $(\ast)$ \\ {\Cref{thr:ucg_construction}}} & \multicolumn{1}{c|}{$O(n + 2^{t+r} (t+r))$} & $O(2^{t+r} (t+r) \log(t+r))$ & \multicolumn{1}{c|}{\makecell[c]{9}} & $O(2^{t+r} (t+r))$ \\ \hline
\makecell[c]{$f\text{-}\mathsf{UCG}$\\ {\Cref{thr:ucg_boolean_construction1}}} & \multicolumn{1}{c|}{$O\Big(|\operatorname{supp}^{>1}(f)| + \big|\bigcup_{S\in\operatorname{supp}^{>1}(f)}S\big|\Big)$} & $O\Big(\sum_{S\in\operatorname{supp}(f)}|S|\Big)$ & \multicolumn{1}{c|}{5} & $O(|\operatorname{supp}^{>1}(f)|)$ \\ \hline
\makecell[c]{$f\text{-}\mathsf{UCG}$ $(\ddagger)$\\ {\Cref{thr:ucg_boolean_construction2}}} & \multicolumn{1}{c|}{$O\Big(s+\big|\bigcup_{S\in\operatorname{supp}^{>1}(f)}S\big|\Big)$} & $O(s\operatorname{deg}(f) + |\operatorname{supp}^{= 1}(f)|)$ & \multicolumn{1}{c|}{5} & $O(s)$ \\ \hline
\makecell[c]{$f\text{-}\mathsf{UCG}$\\ {\Cref{thr:ucg_boolean_construction3}}} & \multicolumn{1}{c|}{\makecell[c]{$O\Big( \sum_{S\in\operatorname{supp}^{>1}_{\{0,1\}}(f)} |S|\Big)$}} &  \makecell[c]{$O\Big(\sum_{S\in\operatorname{supp}_{\{0,1\}}(f)} |S|\log(1+|S|)\Big)$} & \multicolumn{1}{c|}{9} & \makecell[c]{$O\Big(\sum_{S \in \operatorname{supp}^{>1}_{\{0,1\}}(f)}|S|\Big)$} \\  \hline
\end{tabular}
}
\caption{Main results for $f$-$\mathsf{UCG}$, where $f:\{0,1\}^n\to\mathcal{U}(\mathbb{C}^{2\times 2})$ has the $\mathsf{Z}$-decomposition $\alpha,\beta,\gamma,\delta:\{0,1\}^n\to[-1,1]$. In $(\ast)$, $f$ is a $(J,r)$-junta with $|\overline{{J}}|=t$. In $(\ddagger)$, the gate is implemented with spectral norm error at most $\epsilon$ and $s:=(n/\epsilon^2)\sum_{\nu\in\{\alpha,\beta,\gamma,\delta\}}\hat{\|}\nu^{>1}\hat{\|}_1^2$.
}
\label{table:ucg}
\end{table}

\begin{table}[ht]
\centering
\resizebox{\linewidth}{!}{
\begin{tabular}{|c|cc|cc|}
\hline
\multirow{2}{*}{Result} & \multicolumn{2}{c|}{$O(1)$-depth Fan-Out construction} & \multicolumn{2}{c|}{$O(1)$-depth $\mathsf{GT}$ construction} \\ 
\cline{2-5} 
& \multicolumn{1}{c|}{\#Fan-Out} & \#Ancillae & \multicolumn{1}{c|}{\#$\mathsf{GT}$} & \#Ancillae \\ \hline 
\makecell[c]{$f\text{-}\mathsf{FIN}$ $(\ast)$\\ {\Cref{thr:fin_onehot}}} & \multicolumn{1}{c|}{$O(n + 2^{t+r} (t+r))$} & $O(2^{t+r} (t+r) \log(t+r))$ & \multicolumn{1}{c|}{\makecell[c]{6}} & $O(2^{t+r} (t+r))$ \\ \hline
\makecell[c]{$f\text{-}\mathsf{FIN}$\\ {\Cref{thr:fin_boolean_construction1}}} & \multicolumn{1}{c|}{$O\Big(|\operatorname{supp}^{>1}(f)| + \big|\bigcup_{S\in\operatorname{supp}^{>1}(f)}S\big|\Big)$} & $O\Big(\sum_{S\in\operatorname{supp}(f)}|S|\Big)$ & \multicolumn{1}{c|}{2} & $O(|\operatorname{supp}^{>0}(f)|)$ \\ \hline
\makecell[c]{$f\text{-}\mathsf{FIN}$ $(\ddagger)$\\ {\Cref{thr:fin_boolean_construction2}}} & \multicolumn{1}{c|}{$O\Big(s+\big|\bigcup_{S\in\operatorname{supp}^{>1}(f)}S\big|\Big)$} & $O(s\operatorname{deg}(f) + |\operatorname{supp}^{= 1}(f)|)$ & \multicolumn{1}{c|}{2} & $O(s + |\operatorname{supp}^{=1}(f)|)$ \\ \hline
\makecell[c]{$f\text{-}\mathsf{FIN}$\\ {\Cref{thr:fin_boolean_construction3}}} & \multicolumn{1}{c|}{\makecell[c]{$O\Big( \sum_{S\in\operatorname{supp}^{>1}_{\mathbb{F}_2}(f)} |S| \Big)$}} &  \makecell[c]{$O\Big(\sum_{S\in\operatorname{supp}_{\mathbb{F}_2}(f)} |S|\log(1+|S|)\Big)$} & \multicolumn{1}{c|}{6} & $O\Big(\sum_{S \in \operatorname{supp}_{\mathbb{F}_2}(f)}|S| \Big)$ \\ \hline
\end{tabular}
}
\caption{Main results for $f$-$\mathsf{FIN}$, where $f:\{0,1\}^n\to\{0,1\}$. In $(\ast)$, $f$ is a $(J,r)$-junta with $|\overline{{J}}|=t$. In $(\ddagger)$, the gate is implemented with spectral norm error at most $\epsilon$ and $s:=n\hat{\|}f^{>1}\hat{\|}_1^2/\epsilon^2$.
}
\label{table:fan-in}
\end{table}

\begin{table}[ht]
\centering
\resizebox{0.7\linewidth}{!}{
\begin{tabular}{|c|cc|cc|}
\hline
\multirow{2}{*}{Result} & \multicolumn{2}{c|}{$O(d)$-depth Fan-Out construction} & \multicolumn{2}{c|}{$O(d)$-depth $\mathsf{GT}$ construction} \\ 
\cline{2-5} 
& \multicolumn{1}{c|}{\#Fan-Out} & \#Ancillae & \multicolumn{1}{c|}{\#$\mathsf{GT}$} & \#Ancillae \\ \hline 
\makecell[c]{$\mathsf{QRAG}$ \\ {\Cref{thr:qram_recursive_procedure}}} & \multicolumn{1}{c|}{$O(n\log^{(d)}{n})$} & $O(n\log^{(d)}{n}\log^{(d+1)}{n})$ & \multicolumn{1}{c|}{$21d-12$} & $O(n\log^{(d)}{n})$ \\ \hline
\makecell[c]{$\mathsf{QRAM}$ \\ {\Cref{thr:qram_recursive_procedure}}} & \multicolumn{1}{c|}{$O(n\log^{(d)}{n})$} & $O(n\log^{(d)}{n}\log^{(d+1)}{n})$ & \multicolumn{1}{c|}{$16d-10$} & $O(n\log^{(d)}{n})$ \\ \hline
\makecell[c]{$\mathsf{QRAM}$ \\{\Cref{thr:qram_recursive_procedure_boolean}}} & \multicolumn{1}{c|}{$O(n^{1/(1-2^{-d})})$} & $O(n^{1/(1-2^{-d})}\log{n})$ & \multicolumn{1}{c|}{$8d-6$} & $O(n^{1/(1-2^{-d})})$ \\ \hline
\end{tabular}
}
\caption{Main results for $\QRAM$/$\QRAG$ with memory size $n$. $d\in\mathbb{N}$ and $\log^{(d)}{n}$ is the $d$-times iterated logarithm.}
\label{table:qram-qrag}
\end{table}

\begin{table}[ht!]
\centering
\def\arraystretch{1.17}
\resizebox{\linewidth}{!}{
\begin{tabular}{|c|c|c|c|c|c|c|}
\hline
Work & $n$-qubit Circuit & \#Ancillae  & \#Fan-Out & \#$\mathsf{GT}$ & Size & Depth \\ \hline \hline

\cite{barenco1995elementary}&Arbitrary Unitary& $0$ & - & - & $O(n^3 2^{2n})$ & $O(n^3 2^{2n})$ \\ \hline

\cite{knill1995approximation}&Arbitrary Unitary& $0$ & - & - & $O(n 2^{2n})$ & $O(n 2^{2n})$ \\ \hline

\makecell{\cite{vartiainen2004efficient} \\ \cite{mottonen2005decompositions}}&Arbitrary Unitary& $0$ & - & - & $O(2^{2n})$ & $O(2^{2n})$ \\ \hline

\multirow{3}{*}{\cite{hoyer2005quantum}}&Approx.\ $\mathsf{Exact}[t]$& $O(n\log^{(d)}\!{n})$ & $O(n)$ & - & - & $O(d)$ \\ 

&$\mathsf{Exact}[t]$& $O(n)$ & $O(n)$ & - & - & $O(\log^\ast{n})$ \\

&Approx.\ $\mathsf{Threshold}[t]$& $O(n\log{n})$ & $O(n)$ & - & - & $O(1)$ \\ \hline

\multirow{3}{*}{\cite{takahashi2016collapse}}&$\mathsf{Exact}[t]$& $O(n\log^{(d)}\!{n})$ & $O(n)$ & - &  & $O(d)$ \\

&$\mathsf{Threshold}[t \leq \log{n}]$& $O(n\log{n})$ & $O(n)$ & - & - & $O(1)$ \\

&$\mathsf{Threshold}[t \geq \log{n}]$& $O(\sqrt{t\log{n}})$ & $O(n\sqrt{t\log{n}})$ & - & - & $O(1)$ \\ \hline

\multirow{2}{*}{\cite{maslov2018use}}&Arbitrary Clifford& $0$ & - & $12n-18$ & - & $O(n)$ \\

&$\mathsf{Exact}[t]$& $n/2$ & - & $3n-6$ & - & $O(n)$ \\ \hline

\cite{groenland2020signal}&$\mathsf{Exact}[t]$& $0$ & - & $2n$ & - & $O(n)$ \\ \hline

\cite{van2021constructing}&Arbitrary Clifford& $0$ & - & $6n-8$ & - & $O(n)$ \\ \hline

\multirow{2}{*}{\cite{grzesiak2022efficient}}&Arbitrary Clifford& $n/2$ & - & $6\log{n} + O(1)$ & - & $O(\log{n})$ \\ 

&$\mathsf{Exact}[t]$& $O(1)$ & - & $3n/2$ & - & $O(n)$ \\ \hline

\cite{bassler2022synthesis}&Arbitrary Clifford& $0$ & - & $2n$ & - & $O(n)$ \\ \hline

\cite{maslov2022depth}&Arbitrary Clifford& $0$ & - & $2\log{n} + O(1)$ & - & $O(\log{n})$ \\ \hline

\multirow{4}{*}{\cite{bravyi2022constant}}&Arbitrary Clifford& $0$ & - & $26$ & - & $O(1)$ \\ 

&Arbitrary Clifford& $O(n)$ & - & $4$ & - & $O(1)$\\ 

&$\mathsf{Exact}[t]$& $O(n)$ & - & $4$ & - & $O(1)$\\ 

&$\mathsf{Exact}[t]$& $O(\log{n})$ & - & $O(\log^\ast{n})$ & - & $O(1)$\\ \hline

\multirow{2}{*}{\cite{rosenthal2021query}} & $\mathsf{QRAG}$ & $O(n\log{n}\log\log{n})$ & $O(n\log{n})$ & - & - & $O(1)$ \\ 

& Arbitrary Unitary& $\widetilde{O}(2^{2n})$ & - & - & $\widetilde{O}(2^{5n/2})$ & $\widetilde{O}(2^{n/2})$ \\ \hline

\cite{Zhang2022quantum}& $f\text{-}\mathsf{UCG}$ & $O(2^{n})$ & - & - & $O(n2^n)$ & $O(n)$ \\ \hline

\multirow{2}{*}{\cite{STY-asymptotically}}& Arbitrary Unitary& $O(m)$ & - & - & $O(2^{2n})$ & $O(n2^n + 2^{2n}/(n+m))$ \\ 

& $f\text{-}\mathsf{UCG}$ & $O(m)$ & - & - & $O(2^n)$ & $O(n + 2^n/(n+m))$ \\ \hline

\cite{yuan2023optimal}&Arbitrary Unitary& $O(m)$ & - & - & $O(\sqrt{m}2^{3n/2})$ & $O(n2^{n/2} + \sqrt{n/m}2^{3n/2})$ \\ \hline

\multirow{2}{*}{\cite{Low2024tradingtgatesdirty}}& Arbitrary Unitary& $O(m)$ & - & - & $O(2^{2n})$ & $O(n2^n + 2^{2n}/m)$ \\ 

& $\mathsf{QRAM}$ & $2\lceil\log_2{n}\rceil + 2$ & - & - & $O(n)$ & $O(n)$ \\ \hline

\multirow{6}{*}{This work}& $f$-$\mathsf{UCG}$ & $O(n2^n)$ & $O(2^n)$ & - & - & $O(1)$ \\  

& $f$-$\mathsf{UCG}$ & $O(2^n)$ & - & $5$ & - & $O(1)$ \\

& $f$-$\mathsf{FIN}$ & $O(2^n)$ & - & $2$ & - & $O(1)$ \\

&$\mathsf{QRAM}$/$\mathsf{QRAG}$&$O(n\log^{(d)}\!{n}\log^{(d+1)}\!{n})$& $O(n\log^{(d)}\!{n})$ & - & - & $O(d)$ \\ 

& $\mathsf{QRAG}$ & $O(n\log^{(d)}\!{n})$ & - & $21d-12$ & - & $O(d)$ \\ 

& $\mathsf{QRAM}$ & $O(n\log^{(d)}\!{n})$ & - & $16d-10$ & - & $O(d)$ \\ 

& $\mathsf{QRAM}$ & $O(n^{(1/(1-2^{-d})})$ & - & $8d-6$ & - & $O(d)$ \\ \hline

\end{tabular}}
\caption{Summary of some known constructions in the literature regarding number of ancillae, Fan-Out and $\mathsf{GT}$ gates, size (number of single and two-qubit gates), and depth. Here $d$ and $m$ are tunable parameters, $\log^{(d)}{n}$ is the $d$-times iterated logarithm, and $\log^\ast{n}$ is the star-log function. We disregard the size of circuits with Fan-Out and/or $\mathsf{GT}$ gates. In~\cite{yuan2023optimal}, $m\in[\Omega(2^n/n),O(2^{2n})]$, while $m\in[\Omega(n),O(n2^n)]$ in~\cite{Low2024tradingtgatesdirty}.  Ref.~\cite{Low2024tradingtgatesdirty} focuses on minimising the $\mathsf{T}$-count, which is $O(2^n m + 2^{2n}/m)$ for arbitrary unitary. The result of~\cite{maslov2022depth} is implicit.}
\label{table:known_results}
\end{table}

\subsection{Related work}

\paragraph*{Constant-depth complexity classes.} Recall the main classical classes computed by constant-depth and polynomial-size circuits:
\begin{itemize}[noitemsep,topsep=-0.5pt]
    \item $\mathsf{NC}^0$ with $\mathsf{NOT}$ and bounded $\mathsf{AND},\mathsf{OR}$ gates;
    \item $\mathsf{AC}^0$ with $\mathsf{NOT}$ and unbounded   $\mathsf{AND},\mathsf{OR}$ gates;
    \item $\mathsf{TC}^0$ with $\mathsf{NOT}$ and unbounded  $\mathsf{AND},\mathsf{OR},\mathsf{THRESHOLD}[t]$ gates for all $t$;
    \item $\mathsf{AC}^0[q]$ with $\mathsf{NOT}$ and unbounded  $\mathsf{AND},\mathsf{OR},\mathsf{MOD}[q]$ gates;
    \item $\mathsf{ACC}^0 = \bigcup_{q>1}\mathsf{AC}^0[q]$.
\end{itemize}
The study of shallow quantum circuit classes was initiated in~\cite{moore1998some,moore2001parallel}, which introduced a definition of $\mathsf{QNC}^0$, the quantum analogue of the class $\mathsf{NC}^0$.  The remaining quantum analogs of the above circuit classes such as $\mathsf{QAC}^0,\mathsf{QTC}^0,\mathsf{QAC}^0[q]$, and $\mathsf{QACC}^0$ were later defined in~\cite{green2001counting}. In the same paper, the authors introduced expanded versions of the aforementioned classes in which Fan-Out gates are also allowed. For example, the class $\mathsf{QAC}^0_f$ consists of problems solvable by constant-depth and polynomial-size quantum circuits composed by Fan-Out gates and unbounded $\mathsf{AND},\mathsf{OR}$ gates (and similarly for the remaining classes $\mathsf{QTC}_f^0,\mathsf{QAC}^0_f[q],\mathsf{QACC}^0_f$).

Moore~\cite{moore1999quantum} and Green et al.~\cite{green2001counting} proved that $\mathsf{QAC}^0_f = \mathsf{QAC}^0[q] = \mathsf{QACC}^0$ for any $q>1$. This result differs greatly from the classical result~\cite{smolensky1987algebraic} that $\mathsf{AC}^0[p] \neq \mathsf{AC}^0[q]$ for primes $p\neq q$. The power of Fan-Out was further explored in~\cite{hoyer2005quantum} who proved that the bounded-error versions of $\mathsf{QNC}^0_f,\mathsf{QAC}^0_f,\mathsf{QTC}^0_f$ are equal. Later, \cite{takahashi2016collapse} managed to  collapse the hierarchy of constant-depth exact quantum circuits: $\mathsf{QNC}_f^0 = \mathsf{QAC}^0_f = \mathsf{QTC}^0_f$. This is in sharp contrast to the classical result $\mathsf{NC}^0 \subset \mathsf{AC}^0 \subset \mathsf{TC}^0$. Still, open problems abound, e.g., whether $\mathsf{QAC}^0$ and $\mathsf{QAC}_f^0$ are equal or not. In this direction, see~\cite{pius2015parallel,pade2020depth,rosenthal2020bounds,nadimpalli2023pauli,anshu2024computational}.

Regarding the class $\mathsf{QNC}^0$ more specifically, it has been an object of great interest since its proposal. A series of works~\cite{terhal2002adaptive,aaronson2005quantum,aaronson2016complexity,bouland2018quantum,boixo2018characterizing} gave evidence that sampling from the output distribution of shallow quantum circuits cannot be simulated by polynomial-time classical computers. Recently, a new line of research starting in~\cite{bravyi2018quantum} is focused in proving unconditional separation between the classical and quantum constant-depth circuits~\cite{coudron2021trading,gall2018average,watts2019exponential,bravyi2020quantum,grier2020interactive,watts2023unconditional,caha2023colossal,briet2023noisy,grewal2024improved,grilo2024power,hsieh2024unconditionally}.

\paragraph*{Quantum state preparation.} Quantum state preparation is the problem of constructing an $n$-qubit quantum state $\ket{\psi}$ starting from the initial state $\ket{0}^{\otimes n}$ and classical knowledge of the amplitudes of $\ket{\psi}$. To our knowledge, the first results for efficient state preparation are~\cite{grover2000synthesis,grover2002creating}, the latter using oracle access (implementable with a $\mathsf{QRAM}$) to a set of pre-computed partial integrals. Since then, several constructions have been proposed~\cite{bergholm2005quantum,plesch2011quantum,cortese2018loading,sanders2019black,araujo2021divide,bausch2022fast,rattew2022preparing-arbitrary,mcardle2022quantum,plesch2011quantum,zhang2021lowdepth,rosenthal2021query,Zhang2022quantum,STY-asymptotically,bouland2023state,rosenthal2023efficient}.

\section{Preliminaries}\label{sec:preliminaries}

Denote $\mathbb{N} = \{1,2,\dots\}$ and $[n] := \{0,\dots,n-1\}$. Given $d\in\mathbb{N}$, let $\log^{(d)}{n}$ be the $d$-times iterated logarithm defined recursively by $\log^{(d)}{n} = \log(\log^{(d-1)}{n})$ and $\log^{(1)}{n} = \log{n}$. The \emph{star-log function} $\log^\ast{n}$ is the maximum number of iterations $d$ such that $\log^{(d)}{n}$ exists and is real. Let $[n]^{m}$ and $[n]^{\leq m}$ be the set of sequences of size $m$ and size at most $m$, respectively. We shall often equate the decimal and binary representations of a given number. Given $x=x_0x_1\dots x_{n-1}\in\{0,1\}^n$, let $|x|$ be its Hamming weight and $\overline{x}$ its bit-wise negation, i.e., $\overline{x}_i = x_i\oplus 1$ for all $i\in[n]$.
The one-hot encoding $e(x)\in\{0,1\}^{2^n}$ of a string $x\in\{0,1\}^n$ is defined such that $e(x)_j = 1$ if and only if $j=x$, $j\in\{0,1\}^n$, and can be calculated as $e(x)_j = \bigwedge_{k\in[n]} (x\oplus \overline{j})_k$. We take logarithms to the base $2$. Given $\mathsf{A}\in\mathbb{C}^{n\times n}$, its spectral norm is $\|\mathsf{A}\| := \max_{v\in\mathbb{C}^n:\|v\|_2=1}\|\mathsf{A}v\|_2$. Let $\mathcal{U}(\mathbb{C}^{n\times n})$ be the set of $n\times n$ unitary matrices. Let $\mathbb{I}_n$ be the $2^n\times 2^n$ identity matrix, $\mathsf{X},\mathsf{Y},\mathsf{Z}$ the usual Pauli matrices, and $\mathsf{H}$ the Hadamard gate. For $\theta\in[-1,1]$, define $\mathsf{Z}(\theta) := \bigl(\begin{smallmatrix} 1&0 \\ 0&e^{i\pi\theta} \end{smallmatrix} \bigr)$. 
\begin{itemize}
    \item Given an ordered sequence $I\in[n]^m$ of $m$ distinct elements and a unitary $\mathsf{U}\in\mathcal{U}(\mathbb{C}^{2^m\times 2^m})$, let $\mathsf{U}_{\to I}\in\mathcal{U}(\mathbb{C}^{2^n\times 2^n})$ be the unitary that applies $\mathsf{U}$ onto qubits in $I$ and the identity onto the remaining qubits, i.e., $\mathsf{U}_{\to I}|x\rangle = (\mathsf{U}|x_I\rangle)|x_{\overline{I}}\rangle$. If $I = (i) \in [n]$, write $\mathsf{U}_{\to i}$.
    \item Given an ordered sequence $I\in[n]^m$ of $m$ distinct elements, $S\subseteq[n]\setminus I$, and a unitary $\mathsf{U}\in\mathcal{U}(\mathbb{C}^{2^m\times 2^m})$, let $\mathsf{C}_S$-$\mathsf{U}_{\to I}\in\mathcal{U}(\mathbb{C}^{2^n\times 2^n})$ be the unitary that applies $\mathsf{U}$ onto qubits in $I$ controlled on all qubits in $S$ being in the $|1\rangle$ state and the identity onto the remaining qubits (define $\mathsf{C}_\emptyset$-$\mathsf{U}_{\to I} := \mathsf{U}_{\to I}$ if $S=\emptyset$). As an example, $\mathsf{C}_S$-$\mathsf{X}_{\to i}$ is the $\mathsf{X}$ gate applied onto qubit $i$ controlled on qubits in $S$ being in the $|1\rangle$ state (if $|S| = 1$, this is just a $\mathsf{CNOT}$ gate).
\end{itemize}
Let $\mathsf{SWAP}_{i\leftrightarrow j}$ be the gate that swaps qubits $i,j\in[n]$ and $\mathsf{C}_k$-$\mathsf{SWAP}_{i\leftrightarrow j}$ its controlled version on qubit $k\in[n]\setminus\{i,j\}$. In the present work, we use a one and two-qubit universal gate set $\mathcal{G}$, e.g., $\{\mathsf{H},\mathsf{CNOT},\mathsf{Z}(\theta)\}$, supplemented with the multi-qubit Fan-Out and $\mathsf{GT}$ gates defined in \Cref{sec:multi-qubit_gates}.
\begin{fact}[$\mathsf{Z}$-decomposition, {\cite[Theorem~4.1]{nielsen2002quantum}}]
    Let $f:\{0,1\}^n\to\mathcal{U}(\mathbb{C}^{2\times 2})$ be a function onto single-qubit gates. Then there are functions $\alpha,\beta,\gamma,\delta:\{0,1\}^n \to [-1,1]$ such that
    \begin{align*}
        f(x) = e^{i\pi\alpha(x)}\mathsf{Z}(\beta(x))\mathsf{H}\mathsf{Z}(\gamma(x))\mathsf{H}\mathsf{Z}(\delta(x)), \quad x\in\{0,1\}^n.
    \end{align*}
    We say that the tuple $(\alpha,\beta,\gamma,\delta)$ is the \emph{$\mathsf{Z}$-decomposition} of $f$.
\end{fact}
  
The size/arity of a gate is the number of qubits on which it depends and effects, e.g., a $\mathsf{C}_S$-$\mathsf{U}_{\to i}$ gate has arity $|S|+1$. For clarity, we may explicitly denote the arity $k$ of a gate $\mathsf{U}$ by writing $\mathsf{U}^{(k)}$.
Circuit diagrams in this paper use the following convention for controlled gates. A black circle ($\bullet$) denotes a control that is active when the qubit is in the $|1\rangle$ state, while a white circle ($\circ$) denotes a control that is active when the qubit is in the $|0\rangle$ state (see \Cref{fig:UCG}).

\subsection{Boolean analysis}
\label{sec:boolean_analysis}

For an introduction to Boolean analysis, see~\cite{de2008brief,o2014analysis}. In the following, we identify a set $S\subseteq[n]$ with its characteristic vector $S\in\{0,1\}^n$ such that $S_i = 1$ if and only if $i\in S$. Given a real-valued Boolean function $f:\{0,1\}^n\to\mathbb{R}$, its (unique) real-polynomial representation, or Fourier expansion, is
\begin{align*}
    f(x) = \sum_{S\subseteq[n]} \widehat{f}(S)\chi_S(x), \quad\text{where}~ \chi_S(x) := (-1)^{S\cdot x} = (-1)^{\sum_{i\in S}x_i}
\end{align*}
 and its Fourier coefficients $\widehat{f}:2^{[n]}\to\mathbb{R}$ are given by $\widehat{f}(S) = \frac{1}{2^{n}}\sum_{x\in\{0,1\}^n}f(x)\chi_S(x)$. The Fourier expansion is a multipolynomial expansion over $\{-1,1\}$, i.e., it uses $\mathsf{PARITY}$ functions. It is possible to represent a function over $\{0,1\}$, i.e., using $\mathsf{AND}$ functions instead. Given $f:\{0,1\}^n\to\mathbb{R}$, its (unique) real-polynomial $\{0,1\}$-representation is
\begin{align*}
    f(x) = \sum_{S\subseteq[n]} \widetilde{f}(S) x^S, \quad\text{where}~x^S := \prod_{i\in S} x_i ~~\text{and}~~ \widetilde{f}:2^{[n]} \to \mathbb{R} ~\text{is}~ \widetilde{f}(S) = \sum_{T\subseteq S}(-1)^{|S| - |T|}f(T)
\end{align*}
(the formula $\sum_{T\subseteq S}(-1)^{|S| - |T|}f(T)$ is called M\"obius inversion). The Fourier expansion (using $\mathsf{AND}$ or $\mathsf{PARITY}$  functions) is a representation over the \emph{real} field. In the special case of functions with codomain $\{0,1\}$, it is possible to represent them over the field $\mathbb{F}_2$ instead. Given $f:\{0,1\}^n\to\{0,1\}$, its (unique) $\mathbb{F}_2$-polynomial representation (also called algebraic normal form) is
\begin{align*}
    f(x) = \bigoplus_{S\subseteq[n]} \widetilde{f}_{\mathbb{F}_2}(S) x^S, \quad\text{where}~\widetilde{f}_{\mathbb{F}_2}:2^{[n]} \to \{0,1\} ~\text{is}~ \widetilde{f}_{\mathbb{F}_2}(S) = \widetilde{f}(S)~(\text{mod}~2) = \bigoplus_{x:\operatorname{supp}(x)\subseteq S}f(x), 
\end{align*}
with $\operatorname{supp}(x) := \{i\in[n]:x_i\neq 0\}$. The $\mathbb{F}_2$-representation can be obtained from the real $\{0,1\}$-representation by changing the sum over the reals to a sum over $\mathbb{F}_2$, i.e., $\widetilde{f}_{\mathbb{F}_2}(S) = \widetilde{f}(S)~(\text{mod}~2)$.

Given the above representations of a function $f$, there are several important quantities that can be extracted from them. We list various concepts that will be used in the paper.
\begin{enumerate}
    \item $\operatorname{supp}(f) := \{S\subseteq[n]:\widehat{f}(S) \neq 0\}$ is the Fourier support of $f$;
    \item $|\operatorname{supp}(f)|$ is the sparsity of $f$;
    \item $\operatorname{supp}^{>k}(f) := \{S\subseteq[n]:|S|>k,\widehat{f}(S)\neq 0\}$ (similarly for $\operatorname{supp}^{\leq k}(f)$ and $\operatorname{supp}^{=k}(f)$);
    \item $\operatorname{deg}(f) := \max\{|S|:S\in\supp(f)\}$ is the Fourier degree of $f$;
    \item $f^{> k} := \sum_{S\subseteq[n]:|S|> k}\widehat{f}(S)\chi_S$ is the part of $f$ with degree greater than $k$ (similarly for $f^{\leq k}$);
    \item $\hat{\|}f\hat{\|}_1 := \sum_{S\subseteq[n]}|\widehat{f}(S)|$ is the Fourier $1$-norm of $f$;
    \item $\operatorname{supp}_{\{0,1\}}(f) := \{S\subseteq[n]:\widetilde{f}(S)\neq 0\}$ is the $\{0,1\}$-support of $f$;
    \item $\operatorname{supp}^{>k}_{\{0,1\}}(f) := \{S\subseteq[n]:|S|>k,\widetilde{f}(S)\neq 0\}$ (similarly for $\operatorname{supp}^{\leq k}_{\{0,1\}}(f)$ and $\operatorname{supp}^{=k}_{\{0,1\}}(f)$);
    \item $\operatorname{deg}_{\{0,1\}}(f) := \max\{|S|:S\in\operatorname{supp}_{\{0,1\}}(f)\} = \operatorname{deg}(f)$ is the $\{0,1\}$-degree of $f$;
    \item $\operatorname{supp}_{\mathbb{F}_2}(f) := \{S\subseteq[n]:\widetilde{f}_{\mathbb{F}_2}(S) \neq 0\}$ is $\mathbb{F}_2$-support of $f$;
    \item $\operatorname{deg}_{\mathbb{F}_2}(f) := \max\{|S|:S\in\supp_{\mathbb{F}_2}(f)\} \leq \operatorname{deg}(f)$ is the $\mathbb{F}_2$-degree of $f$;
    \item $\operatorname{supp}^{>k}_{\mathbb{F}_2}(f) := \{S\subseteq[n]:|S|>k,\widetilde{f}_{\mathbb{F}_2}(S) \neq 0\}$ (similarly for $\operatorname{supp}^{\leq k}_{\mathbb{F}_2}(f)$ and $\operatorname{supp}^{=k}_{\mathbb{F}_2}(f)$). 
\end{enumerate}

Given $f:\{0,1\}^n\to\mathcal{U}(\mathbb{C}^{2\times 2})$, consider its $\mathsf{Z}$-decomposition $\alpha,\beta,\gamma,\delta:\{0,1\}^n\to[-1,1]$. We extend the above Boolean definitions to $f$ by defining $\operatorname{supp}(f) := \operatorname{supp}(\alpha)\cup \operatorname{supp}(\beta)\cup \operatorname{supp}(\gamma)\cup \operatorname{supp}(\delta)$ and $\operatorname{deg}(f) := \max\{\operatorname{deg}(\alpha),\operatorname{deg}(\beta),\operatorname{deg}(\gamma),\operatorname{deg}(\delta)\}$. Similar definitions apply to $\operatorname{supp}^{>k}(f)$, $\operatorname{supp}^{\leq k}(f)$, $\operatorname{supp}^{=k}(f)$, $\operatorname{supp}_{\{0,1\}}(f)$, and $\operatorname{supp}^{> k}_{\{0,1\}}(f)$.

More generally, consider a function $f:\{0,1\}^n\to V$, where $V$ is a complex vector space. Given a partition $(J,\overline{J})$ of $[n]$ and $z\in\{0,1\}^{|\overline{J}|}$, we write $f_{J|z}:\{0,1\}^{|J|} \to V$ for the subfunction of $f$ given by fixing the coordinates in $\overline{J}$ to the bit values $z$. We say that $f:\{0,1\}^n\to V$ is an $r$-junta for $r\in\mathbb{N}$ if it depends on at most $r$ of its input coordinates, i.e., $f(x) = g(x_{i_1},\dots,x_{i_r})$ for some $g:\{0,1\}^r\to V$ and $i_1,\dots,i_r\in[n]$. We say that $f:\{0,1\}^n\to V$ is a $(J,r)$-junta for $J\subseteq[n]$ and $r\in\mathbb{N}$ if $f_{J|z}:\{0,1\}^{|J|}\to V$ is an $r$-junta for any $z\in\{0,1\}^{|\overline{J}|}$. Note that a $(J,r)$-junta with codomain $\{0,1\}$ can be computed by a $|J|$-partial depth-$r$ decision tree, see~\cite[Definition~3.4]{kumar2023tight}.

\section{Quantum memory architectures}\label{ssec:memory}

In this section, we formally define a model of a quantum computer with quantum access to memory. A simplified model of classical computers can be thought of as (i) a central processing unit (CPU), (ii) a Random Access Memory (RAM) that serves as a temporary storage medium for the CPU to quickly retrieve data, and (iii) auxiliary permanent storage mediums. A RAM constitutes a memory array, an address/input register, and a target/bus/output register. Data is accessed or modified via address lines. When the CPU requires access to the memory, it sends the value from the address register down the address lines, and, depending on the read or write signal, the content of a memory cell is either copied into the target register or stored from the target register into the memory cell. To define a model of a quantum computer with quantum access to memory, it will first be helpful to formally define the quantum processing unit ($\mathsf{QPU}$). 

\begin{definition}[Quantum Processing Unit]\label{def:qpu}
    A Quantum Processing Unit $(\mathsf{QPU})$ of size $m$ is defined as a tuple $(\mathtt{I}, \mathtt{W},\mathcal{G})$ consisting of
    \begin{enumerate}
        \item an $m_{\mathtt{I}}$-qubit Hilbert space called \emph{input register} $\mathtt{I}$;
        \item an $(m-m_{\mathtt{I}})$-qubit Hilbert space called \emph{workspace} $\mathtt{W}$;
        \item a constant-size universal gate set $\mathcal{G}\subset\mathcal{U}(\mathbb{C}^{4\times 4})$.
    \end{enumerate}
    The qubits in the workspace $\mathtt{W}$ are called ancillary qubits or simply ancillae. An input to the $\mathsf{QPU}$, or quantum circuit, is a tuple $(T,|\psi_{\mathtt{I}}\rangle,C_1,\dots,C_T)$ where $T\in\mathbb{N}$, $|\psi_{\mathtt{I}}\rangle\in\mathtt{I}$, and, for each $t\in\{1,\dots,T\}$, $C_t\in\mathcal{I}(\mathcal{G})$ is an instruction from a set $\mathcal{I}(\mathcal{G})$ of possible instructions. Starting from the state $|\psi_0\rangle := |\psi_\mathtt{I}\rangle|0\rangle_{\mathtt{W}}^{\otimes (m-m_\mathtt{I})}$, at each time step $t\in\{1,\dots, T\}$ we obtain the state $|\psi_t\rangle = C_t|\psi_{t-1}\rangle\in\mathtt{I}\otimes\mathtt{W}$. The instruction set $\mathcal{I}(\mathcal{G})\subset\mathcal{U}(\mathbb{C}^{2^m\times 2^m})$ consists of all $m$-qubit unitaries on $\mathtt{I}\otimes\mathtt{W}$ of the form
    \begin{align*}
        \prod_{i=1}^k (\mathsf{U}_i)_{\to I_i}
    \end{align*}
    for some $k\in\mathbb{N}$, $\mathsf{U}_1,\dots,\mathsf{U}_k\in\mathcal{G}$ and pair-wise disjoint non-repeating sequences $I_1,\dots,I_k\in[m]^{\leq 2}$ of at most $2$ elements. We say that $\sum_{i=1}^k |I_i|$ is the \emph{arity/size} of the corresponding instruction. We say that $T$ is the \emph{depth} of the input to the $\mathsf{QPU}$, while its \emph{size} is the sum of the arities/sizes of the instructions $C_1,\dots,C_T$. 
\end{definition}
The extension of this definition to incorporate a quantum memory device ($\mathsf{QMD}$) is then:
\begin{definition}[Quantum Processing Unit and Quantum Memory Device]\label{def:computational_model} 

We consider a model of computation comprising a $\mathsf{QPU}$ of size $\poly\log(n)$ and a Quantum Memory Device $(\mathsf{QMD})$ of $n$ memory registers, where each register is of $\ell$-qubit size (for $n$ a power of $2$). A $\mathsf{QPU}$ and a $\mathsf{QMD}$ are collectively defined by a tuple $(\mathtt{I}, \mathtt{W}, \mathtt{A}, \mathtt{T}, \mathtt{Aux}, \mathtt{M}, \mathcal{G}, \mathsf{V})$ consisting of
\begin{enumerate}
    \item two $(\operatorname{poly}\log{n})$-qubit Hilbert spaces called \emph{input register} $\mathtt{I}$ and \emph{workspace} $\mathtt{W}$ owned solely by the $\mathsf{QPU}$;
    \item a $(\log{n})$-qubit Hilbert space called \emph{address register} $\mathtt{A}$ shared by both $\mathsf{QPU}$ and $\mathsf{QMD}$;
    \item an $\ell$-qubit Hilbert space called \emph{target register} $\mathtt{T}$ shared by both $\mathsf{QPU}$ and $\mathsf{QMD}$;
    \item a $(\poly{n})$-qubit Hilbert space called \emph{auxiliary register} $\mathtt{Aux}$  owned solely by the $\mathsf{QMD}$;
    \item an $n\ell$-qubit Hilbert space called \emph{memory} $\mathtt{M}$ comprising $n$ registers $\mathtt{M}_0, \ldots, \mathtt{M}_{n-1}$, each containing $\ell$ qubits, owned solely by the $\mathsf{QMD}$;
    \item a constant-size universal gate set $\mathcal{G}\subset\mathcal{U}(\mathbb{C}^{4\times 4})$;
    \item a function $\mathsf{V} : [n] \to \mathcal{V}$, where $\mathcal{V}\subset \mathcal{U}(\mathbb{C}^{2^{2\ell}\times 2^{2\ell}})$ is a $O(1)$-size subset of $2\ell$-qubit gates.
\end{enumerate}
The qubits in $\mathtt{W}$, $\mathtt{A}$, $\mathtt{T}$, and $\mathtt{Aux}$ are called ancillary qubits or simply ancillae. An input to the $\mathsf{QPU}$ with a $\mathsf{QMD}$, or quantum circuit, is a tuple $(T,|\psi_\mathtt{I}\rangle,|\psi_{\mathtt{M}}\rangle,C_1,\dots,C_T)$ where $T\in\mathbb{N}$, $|\psi_{\mathtt{I}}\rangle\in\mathtt{I}$, $|\psi_{\mathtt{M}}\rangle\in\mathtt{M}$, and, for each $t\in\{1,\dots,T\}$, $C_t\in\mathcal{I}(\mathcal{G},\mathsf{V})$ is an instruction from a set $\mathcal{I}(\mathcal{G},\mathsf{V})$ of possible instructions. The instruction set $\mathcal{I}(\mathcal{G},\mathsf{V})$ is the set $\mathcal{I}(\mathcal{G})$ from {\rm\Cref{def:qpu}} of instructions on $\mathtt{I}\otimes\mathtt{W}\otimes\mathtt{A}\otimes\mathtt{T}$ augmented with the call-to-the-$\mathsf{QMD}$ instruction that implements the unitary
\begin{align*}
    |i\rangle_{\mathtt{A}}|b\rangle_{\mathtt{T}}|x_i\rangle_{\mathtt{M}_i}|0\rangle^{\otimes \poly{n}}_{\mathtt{Aux}} \mapsto  |i\rangle_{\mathtt{A}}\mathsf{V}(i)\big(|b\rangle_{\mathtt{T}}|x_i\rangle_{\mathtt{M}_i}\big)|0\rangle^{\otimes \poly{n}}_{\mathtt{Aux}}, \qquad \forall i\in[n],b,x_i\in\{0,1\}^\ell.
\end{align*}
Starting from $|\psi_0\rangle|0\rangle^{\otimes \poly{n}}_{\mathtt{Aux}}$, where $|\psi_0\rangle := |\psi_\mathtt{I}\rangle|0\rangle^{\otimes\poly\log{n}}_{\mathtt{W}}|0\rangle_{\mathtt{A}}^{\otimes \log{n}}|0\rangle_{\mathtt{T}}^{\otimes \ell}|\psi_\mathtt{M}\rangle$, at each time step $t\in\{1,\dots, T\}$ we obtain the state $|\psi_t\rangle|0\rangle^{\otimes \poly{n}}_{\mathtt{Aux}} = C_t(|\psi_{t-1}\rangle|0\rangle^{\otimes \poly{n}}_{\mathtt{Aux}})$, where $|\psi_t\rangle\in \mathtt{I}\otimes \mathtt{W}\otimes \mathtt{A}\otimes \mathtt{T}\otimes \mathtt{M}$. 
\end{definition}

We depict the architecture of a quantum processing unit with access to a quantum memory device in \Cref{fig:qvonneumann}. The address register $\mathtt{A}$ (shared by the $\mathsf{QPU}$ and $\mathsf{QMD}$) is used to select a unitary from $\mathcal{V}$
% set $\{\mathsf{V}_0,\dots,\mathsf{V}_{n-1}\}$
and apply it to the target and memory registers $\mathtt{T}$ and $\mathtt{M}$ with the help of the auxiliary register $\mathtt{Aux}$. Even though a call to the $\mathsf{QMD}$ might require gates from a universal gate set, we stress that the underlying quantum circuit implementing such a call is \emph{fixed}, i.e., does not change throughout the execution of a quantum algorithm by the $\mathsf{QPU}$, or even between different quantum algorithms. This allows for highly specialized circuits for the $\mathsf{QMD}$.

We did not include measurements in our models of computation, but these can easily be performed on the output state $|\psi_T\rangle$ if desired. We do not fix the position of qubits within our architecture and thus allow for long-range interactions, generally through multi-qubit entangling gates like Fan-Out and $\mathsf{GT}$ gates (see \Cref{sec:multi-qubit_gates}). Our model, nonetheless, could be augmented with algorithms for efficiently moving and addressing qubits within physically realistic devices. In this direction, we point the reader to Ref.~\cite{beals2013efficient}.

Note that our definition of circuit size in \Cref{def:qpu} differs slightly from the standard notion of circuit size (number of gates from $\mathcal{G}$) up to a factor of at most $2$. In this work, we focus on constant-depth circuits, and since the size of a constant-depth circuit is just a constant times the number of input qubits plus ancillary qubits, we shall specify only the input and the number of ancillae of a circuit, with its size thus being implicit. In the rest of the paper we assume that each memory cell has size $\ell=1$.

\begin{figure}[t]
    \centering
    \includegraphics[trim={1cm 15.7cm 12cm 0.7cm},clip,width=0.5\textwidth]{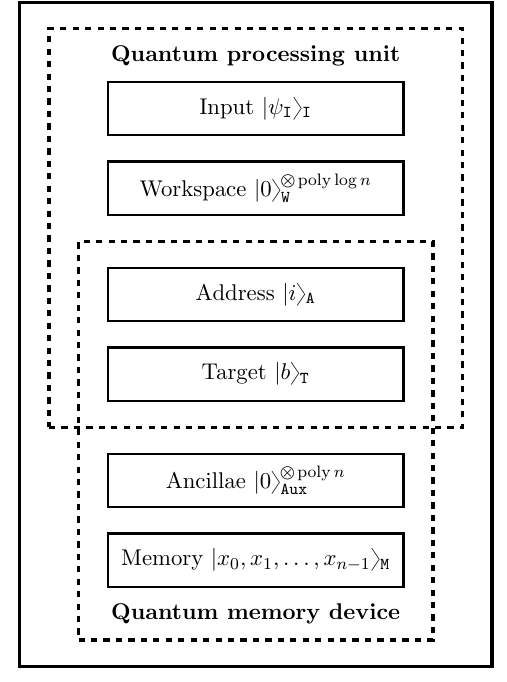}
    \caption{
    The architecture of a Quantum Processing Unit ($\mathsf{QPU}$) with access to a quantum memory device ($\mathsf{QMD}$). The $\mathsf{QPU}$ encompasses a ($\poly\log{n}$)-qubit input register $\mathtt{I}$ and workspace $\mathtt{W}$, a ($\log{n}$)-qubit address register $\mathtt{A}$, and an $\ell$-qubit target register $\mathtt{T}$, while the $\mathsf{QMD}$ encompasses the address register $\mathtt{A}$, the target register $\mathtt{T}$, an $n\ell$-qubit memory array $\mathtt{M}$ composed of $n$ cells $x_0,\dots,x_{n-1} \in \{0,1\}^\ell$ of $\ell$ qubits each, and a ($\poly{n}$)-qubit auxiliary register $\mathtt{Aux}$.}
    \label{fig:qvonneumann}
\end{figure}

A call to the $\mathsf{QMD}$ is defined by the function $\mathsf{V}$ and we shall often equate the quantum memory device with the unitary that it implements. In many applications, one is interested in some form of reading a specific entry from the memory, which corresponds to the special cases where the $\mathsf{V}(i)$ unitaries are made of controlled single-qubit gates, and to which the traditional $\QRAM$ belongs. 
\begin{definition}[$f$-$\mathsf{QRAM}$]\label{def:fqram}
    Let $n\in\mathbb{N}$ be a power of $2$ and $f:\{0,1\}^{\log{n}} \to \mathcal{U}(\mathbb{C}^{2\times 2})$. An $f$-\emph{quantum random access memory} $(f\text{-}\QRAM)$ of memory size $n$ is a $\mathsf{QMD}$ with $\mathsf{V}(i) = \mathsf{C}_{\mathtt{M}_i}$-$f(i)_{\to\mathtt{T}}$, $\forall i\in[n]$. Equivalently, it is a $\mathsf{QMD}$ that maps
    \begin{align*}
        |i\rangle_{\mathtt{A}}|b\rangle_{\mathtt{T}}|x_0,\dots,x_{n-1}\rangle_{\mathtt{M}} \mapsto |i\rangle_{\mathtt{A}}(f(i)^{x_i}|b\rangle_{\mathtt{T}}) |x_0,\dots,x_{n-1}\rangle_{\mathtt{M}} \quad\quad \forall &i\in[n], b,x_0,\dots,x_{n-1}\in\{0,1\}.
    \end{align*}
    A special case, normally called simply $\mathsf{QRAM}$, is when $f(i) = \mathsf{X}$ for all $i\in[n]$, i.e., $\mathsf{V}(i) = \mathsf{C}_{\mathtt{M}_i}$-$\mathsf{X}_{\to\mathtt{T}}$.
\end{definition}
Note that $f$-$\mathsf{QRAM}$s are $\mathsf{QMD}$s that can be implemented via an $\mathsf{UCG}$ (see comment at the end of \Cref{sec:ucg} and \Cref{fig:ucg_qmd_venn}). Another case of interest is writing content from the workspace into memory using $\mathsf{SWAP}$ gates. 
\begin{definition}[$\mathsf{QRAG}$]\label{def:qrag}
    Let $n\in\mathbb{N}$ be a power of $2$. A \emph{quantum random access gate} $\mathsf{QRAG}$ of memory size $n$ is a $\mathsf{QMD}$ with $\mathsf{V}(i) = \mathsf{SWAP}_{\mathtt{M}_i\leftrightarrow \mathtt{T}}$, $\forall i\in[n]$. Equivalently, it is a $\mathsf{QMD}$ that maps
    \begin{align*}
        |i\rangle_{\mathtt{A}}|b\rangle_{\mathtt{T}}|x_0,\dots,x_{n-1}\rangle_{\mathtt{M}} \mapsto |i\rangle_{\mathtt{A}}|x_i\rangle_{\mathtt{T}} |x_0,\dots,x_{i-1},b,x_{i+1},\dots,x_{n-1}\rangle_{\mathtt{M}} \quad \forall i\in[n], b,x_0,\dots,x_{n-1}\in\{0,1\}.
    \end{align*}
\end{definition}
The following lemma shows that a $\mathsf{QRAG}$ is at least as powerful as a $\mathsf{QRAM}$. 

\begin{lemma}[Simulating $\mathsf{QRAM}$ with $\mathsf{QRAG}$]\label{thm:qrag-gives-qram}
A query to a $\mathsf{QRAM}$ of memory size $n$ can be simulated using $2$ queries to a $\mathsf{QRAG}$ of memory size $n$, $3$ two-qubit gates, and $1$ workspace qubit.
\end{lemma}
\begin{proof}
Start with the input $\ket{i}_{\mathtt{A}}\ket{0}_{\mathtt{Tmp}}\ket{b}_{\mathtt{T}}\ket{x_0, \dots, x_{n-1}}_{\mathtt{M}}$ by using an ancillary qubit $\mathtt{Tmp}$ from the workspace. Use a $\mathsf{SWAP}_{\mathtt{T}\leftrightarrow\mathtt{Tmp}}$ gate to obtain $\ket{i}_{\mathtt{A}}\ket{b}_{\mathtt{Tmp}}\ket{0}_{\mathtt{T}}\ket{x_0, \dots, 0, \dots, x_{n-1}}_{\mathtt{M}}$. A query to the $\mathsf{QRAG}$ then leads to  $\ket{i}_{\mathtt{A}}\ket{b}_{\mathtt{Tmp}}\ket{x_i}_{\mathtt{T}}\ket{x_0, \dots, 0, \dots, x_{n-1}}_{\mathtt{M}}$. Use a $\mathsf{C}_{\mathtt{T}}$-$\mathsf{X}_{\to\mathtt{Tmp}}$ gate from register $\mathtt{T}$ to register $\mathtt{Tmp}$, and query again the $\mathsf{QRAG}$, followed by a $\mathsf{SWAP}_{\mathtt{T}\leftrightarrow\mathtt{Tmp}}$ gate, to obtain the desired state $\ket{i}_{\mathtt{A}}\ket{b \oplus x_i}_{\mathtt{T}}\ket{x_0, \dots, x_{n-1}}_{\mathtt{M}}$ after discarding the ancillary qubit.
\end{proof}
On the other hand, in our model, the converse is not true. It is possible, though, to simulate a $\QRAG$ using a constant number of $\QRAM$ queries in a model where single-qubit gates are allowed to be freely applied to the memory register $\mathtt{M}$. The next lemma formalizes these results.
\begin{lemma}[Simulating $\mathsf{QRAG}$ with $\mathsf{QRAM}$]\:
\begin{itemize}
    \item In the model from {\rm \Cref{def:computational_model}}, a query to a $\QRAG$ cannot be simulated by any number of queries to a $\QRAM$.
    \item Suppose that single-qubit gates can be freely applied onto the memory register $\mathtt{M}$ of any $\QRAM$. Then a $\QRAG$ of memory size $n$ can be simulated using $3$ queries to a $\QRAM$ of memory size $n$ and $2(n+1)$ Hadamard gates.
\end{itemize}    
\end{lemma}
\begin{proof}
    For the first statement, consider the simplest case of trying to implement a $\QRAG$ with zero address qubits (i.e., there is only one memory cell): we are given memory qubit $\mathtt{M}$, target qubit $\mathtt{T}$, and an arbitrary number of workspace qubits $\mathtt{W}$. A single action of the $\QRAM$ followed by an arbitrary unitary $\mathsf{U}$ acting on the target and workspace maps $\ket{x_0}_{\mathtt{M}}\ket{b}_{\mathtt{T}}\ket{\psi}_{\mathtt{W}} \mapsto \ket{x_0}_{\mathtt{M}} \mathsf{U}(\ket{b\oplus x_0}_{\mathtt{T}}\ket{\psi}_{\mathtt{W}}) = \ket{x_0}_{\mathtt{M}}\ket{\Phi}_{\mathtt{TW}}$ and thus leaves the memory register invariant. As we cannot modify the memory register, it is not possible to swap the state of the memory with the contents of the target.

    The second statement follows from the simple fact that three $\mathsf{CNOT}$s can implement a $\mathsf{SWAP}$, i.e., $\mathsf{SWAP}_{\mathtt{B}\leftrightarrow \mathtt{D}} = \mathsf{C}_{\mathtt{B}}\text{-}\mathsf{X}_{\to \mathtt{D}}\cdot \mathsf{C}_{\mathtt{D}}\text{-}\mathsf{X}_{\to \mathtt{B}}\cdot \mathsf{C}_{\mathtt{B}}\text{-}\mathsf{X}_{\to \mathtt{D}}$, and that one can swap control and target registers of a $\mathsf{CNOT}$ as $(\mathsf{H}_{\to \mathtt{B}}\cdot\mathsf{H}_{\to \mathtt{D}})\mathsf{C}_{\mathtt{B}}\text{-}\mathsf{X}_{\to \mathtt{D}}(\mathsf{H}_{\to \mathtt{B}}\cdot\mathsf{H}_{\to \mathtt{D}}) = \mathsf{C}_{\mathtt{D}}\text{-}\mathsf{X}_{\to \mathtt{B}}$, for registers $\mathtt{B},\mathtt{D}$. Then, starting from the input $|i\rangle_{\mathtt{A}}|b\rangle_{\mathtt{T}}|x_0,\dots,x_{n-1}\rangle_{\mathtt{M}}$, apply a $\QRAM$ followed by the $n+1$ Hadamard gates $\mathsf{H}_{\to\mathtt{T}}\cdot\prod_{j\in[n]}\mathsf{H}_{\to\mathtt{M}_j}$, and then another $\QRAM$ query followed by $\mathsf{H}_{\to\mathtt{T}}\cdot\prod_{j\in[n]}\mathsf{H}_{\to\mathtt{M}_j}$, and a final $\QRAM$ query.
\end{proof}

Our computational model can be seen as a refined version of the one described in~\cite{buhrman2022memory}. Similar to our \Cref{def:computational_model}, the authors divide the qubits of a quantum computer into work and memory qubits. Given $M$ memory qubits, their workspace consists of $O(\log M)$ qubits, of which the address and target qubits are always the first $\lceil\log M\rceil + 1$ qubits. However, address and target qubits are not considered to be shared by the $\mathsf{QMD}$, and there is no mention of ancillary qubits mediating a call to the $\mathsf{QMD}$. The inner structure of the $\mathsf{QMD}$ is abstracted away by assuming access to the unitary of a $\mathsf{QRAG}$ as in \Cref{def:qrag}. Our model, in contrast, ``opens'' the quantum memory device, and allows for general fixed unitaries, including $\mathsf{QRAM}$ and $\mathsf{QRAG}$. 

The first efficient architectures for $\QRAM$ were formalized and proposed in~\cite{giovannetti2008architectures,giovannetti2008quantum}, namely the Fan-Out and bucket-brigade architectures. These architectures can readily be used for $\mathsf{QRAG}$s, with a simple modification: replacing the last layer of $\mathsf{CNOT}$ gates with $\mathsf{SWAP}$ gates. Both schemes access the memory cells through a binary tree of size $O(n)$ and depth $\log{n}$. Each qubit of the address register $|i\rangle_{\mathtt{A}}$ specifies the direction to follow from the root to the correct memory cell, i.e., the $k$-th qubit of the address register tells whether to go left or right at a router (or bifurcation) on the $k$-th level of the binary tree. The target qubit is sent down the binary tree to the memory cell corresponding to the address register, and the information in the memory cell is copied ($\mathsf{QRAM}$) or swapped ($\mathsf{QRAG}$), and the target qubit is then sent back up the tree to the~root. 

The Fan-Out and bucket-brigade architectures differ in how the target qubit is routed down the binary tree. In the Fan-Out architecture, the $k$-th address qubit controls all the $2^k$ routers on the $k$-th level via a Fan-Out gate. The drawback of this scheme is that it requires simultaneous control of all $n-1$ routers, even though only $\log{n}$ routers (in each branch of the wavefunction) are necessary to route the target down the tree. This in turn makes the Fan-Out architecture highly susceptible to noise since each router is maximally entangled with the rest of the system. In the bucket-brigade architecture, on the other hand, all routers are initially in an ``idle'' state. Each address qubit is sequentially sent down the binary tree and its state is transferred to the first idle router it encounters. This creates a path for the following address qubits to the next idle router and, after all address qubits have been routed down the tree, a path for the target qubits to the correct memory cells. One main advantage of the bucket-brigade architecture is reducing the number of active routers down to $\log{n}$ in each component of the superposition. Another advantage is its high resilience to noise due to limited entanglement between the memory components~\cite{giovannetti2008architectures,giovannetti2008quantum,arunachalam2015robustness,hann2021resilience}.

Several other architectures for $\mathsf{QRAM}$ have been proposed, including Flip-Flop $\mathsf{QRAM}$~\cite{park2019circuit}, Entangling Quantum Generative Adversarial Network $\QRAM$~\cite{niu2022entangling}, approximate Parametric-Quantum-Circuit-based $\QRAM$~\cite{phalak2022approximate}, and others~\cite{chen2021scalable,zoufal2019quantum,niu2022entangling,agliardi2022optimized}. Roughly speaking, one can classify the proposals for $\QRAM$ with classical memory in two ways~\cite{di2020fault}. One way is to explicitly lay out the classical memory in physical hardware at the end of the quantum circuit implementing a $\QRAM$, e.g., at the end of the ancillary binary tree in the Fan-Out and bucket-brigade architectures, and then be copied via a $\mathsf{CNOT}$ gate. The advantage of such ``explicit'' $\QRAM$s is that their underlying circuits must be optimized and compiled just once, while the contents of the memory array can be modified freely. The other way is to encode the memory implicitly in the quantum circuit. This can be achieved by employing multicontrolled $\mathsf{CNOT}$ gates controlled by bits representing the memory address containing a $1$. The advantage of such ``implicit'' $\QRAM$s is that in some cases they can be heavily optimized using techniques from Boolean circuits~\cite{mishchenko2001fast,shafaei2013reversible}. 
Another way to distinguish between $\mathsf{QMD}$s is in the way the routing operation, i.e., the memory cell selection, is implemented: passively or actively. For example, the architecture in~\cite{chen2021scalable} is passive: when the routers (the ancillary qubits of the device) are configured, a photon gets absorbed into a cavity, and then subsequent incoming photons acquire a phase shift depending on the state of the cavity. 
Active architectures~\cite{hann2021resilience}, on the other hand, are similar to a traditional gate-based quantum computer, where each $\mathsf{SWAP}$ or controlled-$\mathsf{SWAP}$ gate is executed by some control pulse.
We point the reader to a few recent surveys on the state of the art of $\mathsf{QRAM}$s for more information~\cite{hann2021practicality,phalak2023quantum,jaques2023qram}.

\subsection{Uniformly controlled gates}\label{sec:ucg}

An $f$-Uniformly Controlled Gate ($f$-$\mathsf{UCG}$ or simply $\mathsf{UCG}$) is a unitary that, conditioned on the state of a set of control qubits, implements one of a set of single-qubit gates on a target qubit.
\begin{definition}[$f$-Uniformly Controlled Gate]\label{def:ucg}
    Let $m,n\in\mathbb{N}$, $n< m$. 
    Let $i \in [m]$.
    Consider a function $f:\{0,1\}^n\to\mathcal{U}(\mathbb{C}^{2\times 2})$, and let $S \in ([m]\setminus\{i\})^n$ be a sequence of $n$ \emph{non-repeating} elements from $[m]\setminus\{i\}$. The Uniformly Controlled Gate $f\text{-}\mathsf{UCG}_{S\to i}^{(n)}$ of arity\footnote{Even though the gate depends on $n+1$ qubits, we define its arity according to $f$.} $n$ is defined as
    \begin{align*}
        \ket{x_0}\ket{x_1}\dots\ket{x_{m-1}} \mapsto \ket{x_0}\dots\ket{x_{i-1}}(f(x_S)\ket{x_i})\ket{x_{i+1}}\dots\ket{x_{m-1}}, \quad \forall x_0,\dots,x_{m-1}\in\{0,1\},
    \end{align*}
    where $x_S=x_{S_1}\dots x_{S_n}$. When it is clear from context, we shall omit either the superscript $(n)$, or the subscripts corresponding to the target $i$ and/or control $S$ from $f\text{-}\mathsf{UCG}^{(n)}_{S\to i}$. By $f\text{-}\mathsf{UCG}$ we mean a generic $f\text{-}\mathsf{UCG}_{S \to i}^{(n)}$ for some $n,S,i$.
\end{definition}

An $f$-$\mathsf{UCG}$ is normally defined in the literature by listing a set $\{\mathsf{U}_0,\dots,\mathsf{U}_{2^n-1}\}$ of single-qubit gates (corresponding to $f(0^n),\dots,f(1^n)$), and writing
\begin{align*}
    f\text{-}\mathsf{UCG}_{[n] \to n}^{(n)} = \sum_{x\in\{0,1\}^n}|x\rangle\langle x|\otimes f(x) = \sum_{x\in\{0,1\}^n}|x\rangle\langle x|\otimes \mathsf{U}_x,
\end{align*}
where we ignored the qubits on which $f\text{-}\mathsf{UCG}_{S\to i}^{(n)}$ does not depend (so $m=n+1$) and took the target qubit $i$ to be the last one. Equivalently, its matrix representation is
\begin{align*}
    f\text{-}\mathsf{UCG}_{[n]\to n}^{(n)} = \begin{pmatrix}
        \mathsf{U}_0 & & &  \\
        & \mathsf{U}_1 & &  \\
        & & \ddots &  \\
        & & & \mathsf{U}_{2^n - 1}
    \end{pmatrix} \in \mathbb{C}^{2^{(n+1)}\times 2^{(n+1)}}.
\end{align*}
A possible way to implement $f\text{-}\mathsf{UCG}_{[n]\to n}^{(n)}$ is shown in \Cref{fig:UCGa}, where each gate $f(x)=\mathsf{U}_x$ is sequentially performed controlled on the state $|x\rangle$. It is possible to simplify such sequential implementation by using Gray codes, see~\cite{barenco1995elementary,bullock2004asymptotically,STY-asymptotically}.

Well-known examples of $f$-$\mathsf{UCG}$s can be found in~\cite{grover2002creating,kerenidis2017quantum,montanaro2015quantum,harrow2009quantum}. These algorithms perform a set of controlled rotations $\ket{x}\ket{0}\mapsto \ket{x}(\cos\theta(x)\ket{0} + \sin\theta(x)\ket{1})$ on a single qubit for a function $\theta:\{0,1\}^n \to [0, 2\pi]$.
Another example is the special subclass of $f$-$\mathsf{UCG}$s known as Fan-In gates ($f$-$\mathsf{FIN}$), for which $f:\{0,1\}^n\to\{\mathbb{I}_1,\mathsf{X}\}$, i.e., the $\mathsf{Z}$-decomposition of $f$ is simply $f(x) = \mathsf{H}\mathsf{Z}(\gamma(x))\mathsf{H} = \mathsf{X}^{\gamma(x)}$ for $\gamma:\{0,1\}^n\to \{0,1\}$. Fan-In gates are thus equivalent to gates for which a Boolean function is computed on a subset of the registers and the result is added to a specified register $|x_i\rangle$. Other $f$-$\mathsf{UCG}$s include phase oracles for which $f:\{0,1\}^n\to\{\mathbb{I}_1,\mathsf{Z}\}$.

\begin{figure}
    \centering
    \begin{subfigure}[b]{0.45\textwidth}
        \includegraphics[trim={1cm 19.4cm 10cm 0.7cm},clip,width=\textwidth]{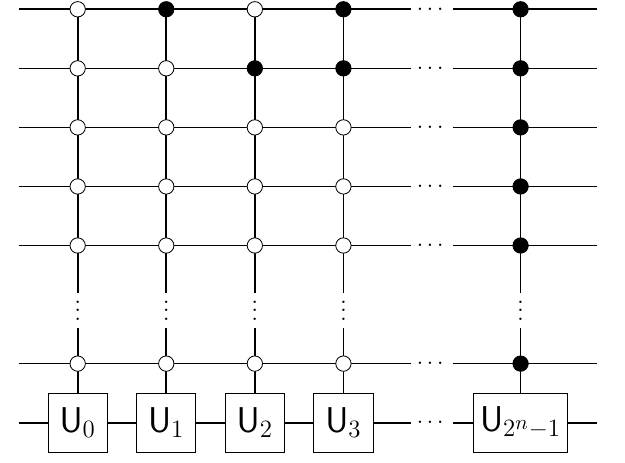}
        \caption{}
        \label{fig:UCGa}
    \end{subfigure}
    \qquad
    \begin{subfigure}[b]{0.395\textwidth}
        \includegraphics[trim={1cm 19.4cm 11.3cm 0.6cm},clip,width=\textwidth]{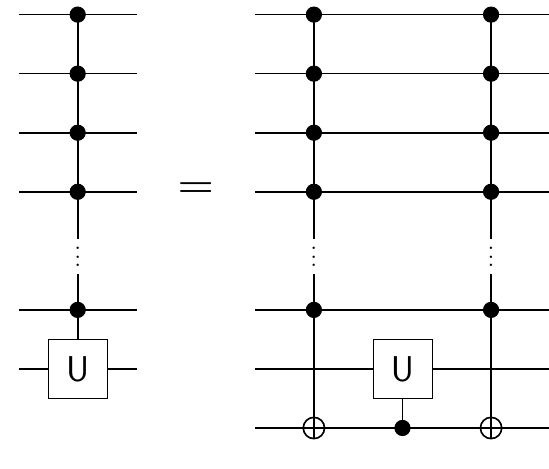}
        \caption{}
        \label{fig:UCGb}
    \end{subfigure}
    \caption{(a) A sequential implementation of $f\text{-}\mathsf{UCG}_{[n]\to n}^{(n)}$.
    (b) The gate $|x\rangle\langle x|\otimes \mathsf{U} + \sum_{j\in\{0,1\}^n\setminus\{x\}}|j\rangle\langle j|\otimes \mathbb{I}_1$ can be implemented by employing two $\mathsf{AND}$ gates onto an ancillary qubit and controlling $\mathsf{U}\in\mathcal{U}(\mathbb{C}^{2\times 2})$ on it being in the $|1\rangle$ state~\cite{barenco1995elementary}.}
    \label{fig:UCG}
\end{figure}

\begin{definition}[$f$-Fan-In gate]\label{def:ffingate}
    Let $m,n\in\mathbb{N}$, $n<m$. Let $i\in[m]$. Consider a Boolean function $f:\{0,1\}^{n}\to\{0,1\}$ on $n$ bits, and let $S\in([m]\setminus\{i\})^n$ be a sequence of $n$ \emph{non-repeating} elements from $[m]\setminus\{i\}$. The Fan-In gate $f\text{-}\mathsf{FIN}^{(n)}_{S\to i}$ of arity $n$ is defined as
    \begin{align*}
        \ket{x_0}\ket{x_1}\dots\ket{x_{m-1}} \mapsto \ket{x_0}\dots\ket{x_{i-1}}\ket{x_i\oplus f(x_S)}\ket{x_{i+1}}\dots\ket{x_{m-1}}, \quad \forall x_0,\dots,x_{m-1}\in\{0,1\},
    \end{align*}
    where $x_S = x_{S_1}\dots x_{S_n}$. When it is clear from context, we shall omit either the superscript $(n)$, or the subscripts corresponding to the target $i$ and/or control $S$ from $f\text{-}\mathsf{FIN}^{(n)}_{S\to i}$. By $f\text{-}\mathsf{FIN}$ we mean a generic $f\text{-}\mathsf{FIN}_{S \to i}^{(n)}$ for some $n,S,i$.
\end{definition}

\noindent Examples of Boolean functions $f:\{0,1\}^n\to\{0,1\}$ include
\begin{align*}
    f(x) &= 1 \text{ if and only if } \begin{cases}
    |x|>0 &\qquad \mathsf{OR}^{(n)},\\
    |x|=n &\qquad \mathsf{AND}^{(n)} \text{ (generalized Toffoli)},\\
    |x|\ge n/2 &\qquad \mathsf{MAJORITY}^{(n)}, \\
    |x|\ge t &\qquad \mathsf{THRESHOLD}^{(n)}[t], \\
    |x| = t &\qquad \mathsf{EXACT}^{(n)}[t],\\
    |x| ~(\operatorname{mod}~q) = 0 &\qquad \mathsf{MOD}^{(n)}[q],\\
    |x|~\text{is odd} &\qquad \mathsf{PARITY}^{(n)}.
    \end{cases}
\end{align*}
Another example of $f$-$\mathsf{FIN}$ is the $\mathsf{QRAM}$ itself. Indeed, $\mathsf{QRAM}$ is simply the $f$-$\mathsf{FIN}$ with $f:\{0,1\}^n\times \{0,1\}^{\log{n}} \to \{0,1\}$ defined by $f(x,i) = x_i$ (also known as selection function).

The following simple fact is behind our constructions based on one-hot encoding in \Cref{sec:constructions_onehot}.

\begin{fact}
    Given $x\in\{0,1\}^n$ and $\mathsf{U}\in\mathcal{U}(\mathbb{C}^{2\times 2})$, the gate $|x\rangle\langle x|\otimes \mathsf{U} + \sum_{j\in\{0,1\}^n\setminus\{x\}}|j\rangle\langle j|\otimes \mathbb{I}_1$ can be implemented using two $\mathsf{AND}^{(n)}$ gates and one ancillary qubit.
\end{fact} 
\begin{proof}
    Given $n$-qubit register $|k\rangle_{\mathtt{I}} = \bigotimes_{j\in[n]} |k_j\rangle_{\mathtt{I}_j}$ and single-qubit register $|b\rangle_{\mathtt{T}}$, simply note that
    \begin{align*}
    \begin{multlined}[b][\textwidth]
        \left(|x\rangle\langle x|\otimes \mathsf{U} + \sum_{j\in\{0,1\}^n\setminus\{x\}}|j\rangle\langle j|\otimes \mathbb{I}_1\right)\otimes \mathbb{I}_1|k\rangle_{\mathtt{I}}|b\rangle_{\mathtt{T}}|0\rangle_{\mathtt{Tmp}} \\
        = \left(\prod_{j\in[n]}\mathsf{X}^{\overline{x}_j}_{\to\mathtt{I}_j}\right)\mathsf{AND}^{(n)}_{\mathtt{I}\to\mathtt{Tmp}}\cdot\mathsf{C}_{\mathtt{Tmp}}\text{-}\mathsf{U}_{\to\mathtt{T}}\cdot\mathsf{AND}^{(n)}_{\mathtt{I}\to\mathtt{Tmp}}\left(\prod_{j\in[n]}\mathsf{X}^{\overline{x}_j}_{\to\mathtt{I}_j}\right)|k\rangle_{\mathtt{I}}|b\rangle_{\mathtt{T}}|0\rangle_{\mathtt{Tmp}}
    \end{multlined}
    \end{align*}
    for all $k\in\{0,1\}^n$ and $b\in\{0,1\}$ (see \Cref{fig:UCGb}).
\end{proof}

\paragraph{Relation between $\mathsf{QMD}$ and $f$-$\mathsf{UCG}$.}

Uniformly Controlled Gates and quantum memory devices are similar but distinct concepts. Since $\mathcal{V}\subset\mathcal{U}(\mathbb{C}^{4\times 4})$, i.e., $\mathsf{V}(i)$ can act non-trivially on two qubits for all $i\in\{0,1\}^{\log{n}}$ (registers $\mathtt{T}$ and $\mathtt{M}_i$), it is clear that $f$-$\mathsf{UCG}$s cannot simulate general $\mathsf{QMD}$s. However, if, for all $i\in\{0,1\}^{\log{n}}$, $\mathsf{V}(i)$ is of the form $f(i)\otimes \mathbb{I}_1$ for some  $f:\{0,1\}^{\log{n}}\to\mathcal{U}(\mathbb{C}^{2\times 2})$, then such $\mathsf{QMD}$ is simply the $f$-$\mathsf{UCG}^{(\log{n})}$. Similarly, an $f$-$\QRAM$ for $f:\{0,1\}^{\log{n}}\to\mathcal{U}(\mathbb{C}^{2\times 2})$ (which is a $\mathsf{QMD}$ such that $\mathsf{V}(i)=\mathbb{I}_1\otimes |0\rangle\langle 0|_{\mathtt{M}_i} + f(i)\otimes |1\rangle\langle 1|_{\mathtt{M}_i}$) is an $f'$-$\mathsf{UCG}^{(n+\log{n})}$ for some $f':\{0,1\}^{n}\times\{0,1\}^{\log{n}}\to\mathcal{U}(\mathbb{C}^{2\times 2})$ that is a $(J=[n],1)$-junta.

In the other direction, the requirement that $\mathcal{V}$ be a constant-size set limits the kind of $f$-$\mathsf{UCG}$ that can be simulated by a $\mathsf{QMD}$ to those where $f$ has a constant range. Relaxing the restriction on the size of $\mathcal{V}$, an $f$-$\mathsf{UCG}^{(1+\log{n})}$ can be simulated by a $\mathsf{QMD}$ of memory size $n$.

\section{Multi-qubit gates as building blocks}
\label{sec:multi-qubit_gates}
\subsection{The Fan-Out gate}
\label{sec:fanout}

The Fan-Out gate copies a specific register $|x_i\rangle$ into a subset of other registers. It can be thought of as a single-control multiple-target $\mathsf{CNOT}$ gate.
\begin{definition}[Fan-Out gate]\label{def:fanout}
    Let $m\in\mathbb{N}$. Let $i\in[m]$ and $S\subseteq[m]\setminus\{i\}$, with $|S\cup\{i\}|=:n$. The Fan-Out gate $\mathsf{FO}_{i\to S}^{(n)}$ of arity $n$ is defined as
    \begin{align*}
        |x_0\rangle|x_1\rangle\dots|x_{m-1}\rangle \mapsto \bigotimes_{j\in[m]} \begin{cases}
            |x_j\oplus x_i\rangle ~&\text{if}~j\in S,\\
            |x_j\rangle ~&\text{if}~j\notin S,
        \end{cases}
        \quad\quad\quad \forall x_0,\dots,x_{m-1}\in\{0,1\},
    \end{align*}
    which copies the bit $x_i$ into the registers in $S$. Similarly to $f$-$\mathsf{UCG}$, we shall sometimes omit either the superscript $(n)$, or the subscripts corresponding to the control $i$ and/or target $S$ from $\mathsf{FO}^{(n)}_{i\to S}$.
\end{definition}

The Fan-Out gate is known to be powerful, in that other multi-qubit gates can be efficiently implemented if one has access to Fan-Out. In particular, we have the following fact.
\begin{fact}[\cite{moore1999quantum,green2001counting}]\label{fact:fanoutparity}
    The Fan-Out gate is equivalent to the $\mathsf{PARITY}$ gate up to a Hadamard conjugation, i.e., for $i\in[n]$ and $S\subseteq [n]\setminus\{i\}$,
    \begin{align*}
        \mathsf{PARITY}_{S\to i} = \left(\prod_{j\in S\cup\{i\}} \mathsf{H}_{\to j} \right)\mathsf{FO}_{i\to S}\left(\prod_{j\in S\cup\{i\}} \mathsf{H}_{\to j} \right).   
\end{align*}
\end{fact}

It is known that the $\mathsf{EXACT}^{(n)}$ gate (including $\mathsf{OR}^{(n)}$ and $\mathsf{AND}^{(n)}$) can be simulated exactly in constant depth using Fan-Out and single-qubit gates~\cite{takahashi2016collapse}. Other known constructions with Fan-Out include $\mathsf{MAJORITY}$ and $\mathsf{THRESHOLD}$~\cite{hoyer2005quantum,takahashi2016collapse}.\footnote{The power of quantum $\mathsf{THRESHOLD}$ has recently been  explored in~\cite{grier2024quantum} one year after our work appeared online and shortly before publication.}
\begin{fact}[{\cite[Theorem~1]{takahashi2016collapse}}]
    \label{thr:or_constantdepth}
    The $\mathsf{EXACT}^{(n)}[t]$ gate can be implemented in $O(1)$-depth using $2n\log{n} + O(n)$ ancillae and $6n + O(\log{n})$ Fan-Out gates with arity $\leq 2n$.
\end{fact}
The above result comes from a useful $\mathsf{OR}$ reduction from $n$ to $\lceil\log(n+1)\rceil$ qubits developed in~\cite{hoyer2005quantum}. We include the proof for completeness, and explicitly count the resources required.
\begin{fact}[{\cite[Lemma~5.1]{hoyer2005quantum}}]\label{fact:or_reduction}
    The $\mathsf{OR}^{(n)}$ gate can be reduced to $\mathsf{OR}^{(p)}$, $p=\lceil\log(n+1)\rceil$, in $O(1)$-depth using $2n\lceil\log(n+1)\rceil$ ancillae and $2n + 2\lceil\log(n+1)\rceil$ Fan-Out gates with arity at most~$n$. In other words, there is a $O(1)$-depth circuit that maps $|x\rangle|0\rangle^{\otimes p} \mapsto |x\rangle|\psi_x\rangle$ for $x\in\{0,1\}^n$, where $|\psi_x\rangle\in\mathbb{C}^{2^p}$ is such that $\langle 0^p|\psi_x\rangle = 1$ if $\mathsf{OR}(x) = 0$ and $\langle 0^p|\psi_x\rangle = 0$ if $\mathsf{OR}(x) = 1$.
\end{fact}
\begin{proof}
    Given the input $|x\rangle|0\rangle$, $x\in\{0,1\}^n$, we first show how to compute $|x\rangle|0\rangle \mapsto |x\rangle|\mu_\theta^{|x|}\rangle$ in constant depth, where $|\mu_\theta^{|x|}\rangle := \frac{1}{2}(1+e^{i\pi\theta|x|})|0\rangle + \frac{1}{2}(1-e^{i\pi\theta|x|})|1\rangle$, $\theta\in[-1,1]$. Attach an ancillary register $|0\rangle^{\otimes (n-1)}$ and apply a Hadamard gate on the first qubit of $|0\rangle^{\otimes n}$ followed by a Fan-Out gate copying this first qubit onto the remaining $n-1$ qubits. This leads to
    \begin{align*}
        |x\rangle|0\rangle^{\otimes n} \mapsto |x\rangle\frac{|0\rangle + |1\rangle}{\sqrt{2}}|0\rangle^{\otimes (n-1)} \mapsto |x\rangle\frac{|0\rangle^{\otimes n} + |1\rangle^{\otimes n}}{\sqrt{2}}.
    \end{align*}
    Apply a $\mathsf{Z}(\theta x_i)$ gate on the $i$-th qubit of $\frac{1}{\sqrt{2}}(|0\rangle^{\otimes n} + |1\rangle^{\otimes n})$ controlled on $|x_i\rangle$, for $i\in[n]$. Thus
    \begin{align*}
        |x\rangle\frac{|0\rangle^{\otimes n} + |1\rangle^{\otimes n}}{\sqrt{2}} \mapsto |x\rangle\frac{|0\rangle^{\otimes n} + e^{i\pi\theta|x|}|1\rangle^{\otimes n}}{\sqrt{2}}.
    \end{align*}
    Uncomputing the first step leads to $|x\rangle|\mu_\theta^{|x|}\rangle$ as required. In total, we have used $n-1$ ancillae and $2$ Fan-Out gates with arity $n$.

    The reduction works by computing in parallel the states $|\psi_k\rangle = |\mu_{\theta_k}^{|x|}\rangle$ with $\theta_k = 1/2^k$, for all $k\in[p]$, which requires copying the register $|x\rangle$ a number of $p-1$ times by using $n$ Fan-Out gates with arity $p$.
    The output $|\psi\rangle = |\psi_0\rangle|\psi_1\rangle\dots |\psi_{p-1}\rangle$ is the desired state. Indeed, if $|x| = 0$, then $\langle 0^p|\psi\rangle = 1$, since $|\psi_k\rangle = |0\rangle$ for each $k\in[p]$. On the other hand, if $|x| \neq 0$, then $\langle 0^p|\psi\rangle = 0$, since at least one qubit $|\psi_k\rangle$ is $|1\rangle$ with certainty. Indeed, there are integers $a\in[p]$ and $b\geq 0$ such that $|x| = 2^a(2b+1)$. Then a direct calculation shows that $\langle 1|\psi_a\rangle = 1$. This proves the correctness of the reduction. Finally, the whole reduction uses $n(p-1)+p(n-1) \leq 2np$ ancillae and $2m+2p$ Fan-Out gates with arity at most $n$.
\end{proof} 

Pham and Svore~\cite{pham20132d} showed how to implement an $n$-arity Fan-Out using a $O(1)$-depth and $O(n)$-size quantum circuit with intermediary measurements, classical feedback, and classical $O(\log{n})$-depth computation. Without intermediary measurements and classical feedback, it is folklore that $n$-arity Fan-Out gates can be implemented by a $\mathsf{CNOT}$ circuit of depth $O(\log n)$ and size $O(n)$. 
If low-arity Fan-Out gates are available, it is possible to improve the number and depth of $\mathsf{CNOT}$ gates required as follows. 

\begin{lemma}\label{lem:fanout-cascade}
    For $y \in \{0,1\}$, the unitary $\ket{y}\ket{0}^{\otimes n} \mapsto \ket{y}^{\otimes(n+1)}$ can be implemented with $\lceil n/(k-1)\rceil$ $k$-arity Fan-Out gates in depth $\lceil\log_k(n+1)\rceil$.
\end{lemma}
\begin{proof}
    Note that, starting from the initial state at depth $d=0$ until the final state at maximum depth $d=d_c$, the $i$-th Fan-Out layer maps the $i$-th state onto the $(i+1)$-th state, for $i\in[d_c].$ We prove by induction that, at depth $d$, our state is $|y\rangle^{\otimes k^d}|0\rangle^{\otimes (n+1-k^d)}$. The case $d=0$ is obvious. Assume the induction hypothesis for $d$. Then, after applying one layer of $k^d$ $k$-arity Fan-Out gates we obtain
    \begin{align*}
        |y\rangle^{\otimes k^d}|0\rangle^{\otimes (n+1-k^d)} \mapsto |y\rangle^{\otimes k^{d+1}}|0\rangle^{\otimes (n+1-k^{d+1})},
    \end{align*}
    as wanted. The circuit depth $d_c$ is the minimum $d$ such that $n+1 - k^{d} \leq 0$, i.e., $d_c = \lceil \log_k(n+1) \rceil$. Regarding the size, from depth $d=0$ to $d=d_c-1$ we require $\sum_{j=0}^{d_c-2}k^j = (k^{d_c-1} - 1)/(k-1)$ Fan-Out gates. In the final layer there are only $n+1 - k^{d_c-1}$ qubits left in the $|0\rangle$ state, thus another $\lceil\frac{n+1}{k} - k^{d_c-2} \rceil$ Fan-Out gates are required. In total, the number of Fan-Outs is at most
    \begin{equation*}
    \scalemath{0.975}{\frac{k^{d_c-1} - 1}{k-1} + \left\lceil \frac{n+1}{k} - k^{d_c-2} \right\rceil = \left\lceil\frac{n+1}{k} + \frac{k^{d_c-2} - 1}{k-1}\right\rceil \leq \left\lceil\frac{n+1}{k} + \frac{(n+1)/k - 1}{k-1}\right\rceil = \left\lceil\frac{n}{k-1}\right\rceil}. \qedhere\end{equation*}
\end{proof}

\subsection{The Global Tunable gate}\label{sec:GT}

\begin{definition}[Global Tunable gate]\label{def:globaltunable}
    Let $\Theta\in[-1,1]^{n\times n}$. The $n$-arity Global Tunable gate $\mathsf{GT}^{(n)}_\Theta$ is the unitary operator
    \begin{align*}
        \mathsf{GT}^{(n)}_{\Theta} = \prod_{1\leq i < j \leq n} \mathsf{C}_i\text{-}\mathsf{Z}(\Theta_{ij})_{\to j}.
    \end{align*}
\end{definition}
The $\mathsf{GT}$ gate is powerful in that it can perform many Fan-Out gates in parallel.
\begin{lemma}\label{claim:fanoutasGT}
    A number $l$ of pair-wise commuting Fan-Out gates $\mathsf{FO}^{(n_0)},\dots,\mathsf{FO}^{(n_{l-1})}$ can be performed in depth-$3$ using one $\mathsf{GT}$ gate with arity at most $n$ and at most $2(n-1)$ Hadamard gates, where $n := \sum_{j\in[l]} n_j$.
\end{lemma}
\begin{proof}
    Let $T\subseteq[l]$.
    Without lost of generality, for each $i\in[l]$, consider a Fan-Out gate $\mathsf{FO}_{q_i\to S_i}$ controlled on qubit $q_i\in T$ with target qubits in $S_i\subseteq[n]\setminus T$. All Fan-Out gates $\mathsf{FO}_{q_i\to S_i}$ commute since the sets of target and control qubits $T$ and $\bigcup_{i\in[l]}S_i$, respectively, are disjoint. Therefore
    \begin{align*}
        \scalemath{0.93}{\prod_{i\in[l]}\mathsf{FO}_{q_i\to S_i} = \prod_{i\in[l]}\prod_{j\in S_i}\mathsf{C}_{q_i}\text{-}\mathsf{X}_{\to j} = \prod_{i\in[l]}\prod_{j\in S_i}\mathsf{H}_{\to j}\cdot\mathsf{C}_{q_i}\text{-}\mathsf{Z}_{\to j}\cdot\mathsf{H}_{\to j} = \left(\prod_{j\in \bigcup_{i\in[l]}S_i}\mathsf{H}_{\to j}\right)\mathsf{GT}_{\Theta}\left(\prod_{j\in \bigcup_{i\in[l]}S_i}\mathsf{H}_{\to j}\right)}.
    \end{align*}
    Here $\Theta\in\{0,1\}^{n\times n}$ is the matrix $\Theta = \bigoplus_{i\in[l]}\Theta_{q_i}$, where $\Theta_{q_i}\in\{0,1\}^{n\times n}$ is the matrix whose $q_i$-th row is the characteristic vector of the set $S_i$, while the remaining rows are zero (the parity is taken entry-wise). In other words, the $(i,j)$-entry of $\Theta$ is the parity of the number of sets $S_0,\dots,S_{l-1}$ that contain $j\in [n]\setminus T$ and are controlled on qubit $i\in T$. This is because a $\mathsf{Z}$ gate is applied onto a qubit $j\in [n]\setminus T$ only if it is the target to an odd number of Fan-Out gates. The maximum arity of the $\mathsf{GT}$ gate happens when $(S_i \cup \{q_i\}) \cap (S_j \cup \{q_j\})=\emptyset$ for every $i\neq j \in [l]$, i.e., when each Fan-Out gate acts on a separate set of qubits. The maximum number of Hadamard gates happens when $q_0 = \dots = q_{l-1}$ and $S_i \cap S_j=\emptyset$ for every $i\neq j \in [l]$, i.e., when all Fan-Out gates share the same control qubit but copy it into separate sets of qubits.
\end{proof}

Thus, up to conjugation by Hadamards, a single $\mathsf{GT}$ gate can copy a control register into an arbitrary number of target registers. Moreover, from the above fact follows a simple yet interesting result concerning constant-depth quantum circuits.
\begin{lemma}
    \label{lem:fanout_to_gt_circuits}
    Consider a constant-depth circuit that uses $l$ Fan-Out gates $\mathsf{FO}^{(n_0)},\dots,\mathsf{FO}^{(n_{l-1})}$. There is an equivalent constant-depth circuit that uses $O(1)$ $\mathsf{GT}^{(n)}$ gates, where $n \leq \sum_{j\in[l]} n_j$.
\end{lemma}
\begin{proof}
    The circuit consists of a constant number of layers, each of which contains at most $l$ disjoint Fan-Out gates. The result follows from \Cref{claim:fanoutasGT}.
\end{proof}
An example of the above result applied to \Cref{thr:or_constantdepth} (up to an exact gate count) is the following result concerning the $\mathsf{EXACT}^{(n)}$ gate.
\begin{fact}[\cite{bravyi2022constant}]
    \label{thr:or_constantdepth_gt}
    The $\mathsf{EXACT}^{(n)}[t]$ gate can be implemented in $O(1)$-depth using $2n + O(\log{n})$ ancillae and $4$ $\mathsf{GT}$ gates with arity $\leq n + O(\log{n})$.
\end{fact}

Note that the original construction of~\cite{bravyi2022constant} is for $\mathsf{OR}^{(n)}$ and requires fewer than $2n$ ancillae because they implement a slightly different gate, $|x\rangle \mapsto (-1)^{\mathsf{OR}(x)}|x\rangle$. Their construction is similar to the one from~\cite{takahashi2016collapse} and uses the $\mathsf{OR}$ reduction from~\cite{hoyer2005quantum} adapted to $\mathsf{GT}$ gates. We include the proof of the $\mathsf{OR}$ reduction for completeness and explicitly count the resources required.

\begin{fact}[\cite{bravyi2022constant}]
    The $\mathsf{OR}^{(n)}$ gate can be reduced to $\mathsf{OR}^{(p)}$, $p=\lceil\log(n+1)\rceil$, in $O(1)$-depth using $1$ $\mathsf{GT}$ gate with arity $n + \lceil\log(n+1)\rceil$ and no ancillae.
\end{fact}
\begin{proof}
    The general construction is the same as in \Cref{fact:or_reduction}, the difference being the number of ancillary qubits. When constructing the states $|\mu_\theta^{|x|}\rangle$, there is no need for the ancillary register $|0\rangle^{\otimes (n-1)}$, since all the $\mathsf{Z}(\theta x_i)$ gates controlled on $|x_i\rangle$, $i\in[n]$, can be performed using a single $\mathsf{GT}$ gate with arity $n+1$ on the state $\frac{1}{\sqrt{2}}(|0\rangle + |1\rangle)$, i.e., the mapping
    \begin{align*}
        |x\rangle\frac{|0\rangle + |1\rangle}{\sqrt{2}} \mapsto |x\rangle\frac{|0\rangle + e^{i\pi\theta |x|}|1\rangle}{\sqrt{2}}
    \end{align*}
    requires only one $\mathsf{GT}$ gate and no ancillae. This procedure actually scales to computing $|x\rangle|0\rangle^{\otimes p} \mapsto |x\rangle|\psi_0\rangle\dots|\psi_{p-1}\rangle$, i.e., all states $|\psi_k\rangle = |\mu^{|x|}_{\theta_k}\rangle$ can be computed in parallel with a single $\mathsf{GT}$ gate with arity $n+p$ and no ancillary qubits.
\end{proof}

As another example, it was shown in~\cite[Theorem~6]{rosenthal2021query} that every $n$-qubit state can be constructed by a $\mathsf{QAC}^0_f$ circuit with $\widetilde{O}(2^n)$ ancillae. It follows from \Cref{lem:fanout_to_gt_circuits} that every $n$-qubit state can be constructed with a constant number of $\mathsf{GT}$ gates and $\widetilde{O}(2^n)$ ancillae.

\paragraph{Comment on physical implementation of multi-qubit gates.}  The constant-depth architectures we consider make use of the multi-qubit Fan-Out and $\mathsf{GT}$ gates. However, the complexity and time required to implement such gates in practice may differ and may be both hardware and code-dependent. For example, if one considers logical qubits encoded via the surface code, then for a fixed code distance $d$, Fan-Out gates can be performed in a constant number of surface code cycles via braiding~\cite{fowler2012surface}. On the other hand, in the non-error-corrected ion-trap $\mathsf{GT}$ gate implementation proposed in~\cite{grzesiak2020efficient}, each of the $n$ qubits is simultaneously acted on by a separate sequence of at least $2n$ constant-duration laser pulses. Assuming a practical lower bound on the duration of any pulse in this sequence, the wall-clock time required to implement a single $\mathsf{GT}$ gate according to this scheme scales linearly with $n$ (and uses a linear number of laser sources) and the constant-depth $\mathsf{GT}$ gate constructions do not necessarily translate to constant-time constructions. This is not surprising, since the $\mathsf{GT}$ gate is strictly more powerful than Fan-Out.

\section{Constant-depth circuits based on one-hot encoding}
\label{sec:constructions_onehot}

In this section, we provide constant-depth circuits for $f$-$\mathsf{UCG}$s via our first technique based one-hot encoding. We rely on a simple fact regarding the unitary
\begin{align*}
    \mathsf{C}_{\mathtt{E}_{n-1}}\text{-}(\mathsf{U}_{n-1})_{\to\mathtt{T}}\cdots \mathsf{C}_{\mathtt{E}_{1}}\text{-}(\mathsf{U}_{1})_{\to\mathtt{T}}\cdot\mathsf{C}_{\mathtt{E}_{0}}\text{-}(\mathsf{U}_{0})_{\to\mathtt{T}}
\end{align*}
that sequentially applies the gates $\mathsf{U}_0,\dots,\mathsf{U}_{n-1}$ onto a target qubit $\mathtt{T}$ controlled on the single-qubit registers $\mathtt{E}_0,\dots,\mathtt{E}_{n-1}$, respectively, being in the $|1\rangle$ state (see \Cref{fig:circuit_reorderinga}). Let $|e\rangle_{\mathtt{E}} = \bigotimes_{j\in[n]}|e_j\rangle_{\mathtt{E}_j}$ be the state of the registers $\mathtt{E}_0,\dots,\mathtt{E}_{n-1}$. If $e\in\{0,1\}^n$ has Hamming weight at most $1$ (i.e., $|e| \leq 1$), then we can rearrange the gates from the $\mathsf{Z}$-decomposition of $\mathsf{U}_j$ as
\begin{align*}
    &\mathsf{C}_{\mathtt{E}_{n-1}}\text{-}(\mathsf{U}_{n-1})_{\to\mathtt{T}}\cdots \mathsf{C}_{\mathtt{E}_{1}}\text{-}(\mathsf{U}_{1})_{\to\mathtt{T}}\cdot\mathsf{C}_{\mathtt{E}_{0}}\text{-}(\mathsf{U}_{0})_{\to\mathtt{T}}|e\rangle_{\mathtt{E}}|b\rangle_{\mathtt{T}} = \\
    &\left(\prod_{j\in[n]} \mathsf{Z}(\alpha_j)_{\to\mathtt{E}_j}\right)\left(\prod_{j\in[n]} \mathsf{C}_{\mathtt{E}_j}\text{-}\mathsf{Z}(\beta_j)_{\to\mathtt{T}}\right) \mathsf{H}_{\to \mathtt{T}} \left(\prod_{j\in[n]} \mathsf{C}_{\mathtt{E}_j}\text{-}\mathsf{Z}(\gamma_j)_{\to\mathtt{T}}\right) \mathsf{H}_{\to \mathtt{T}} \left(\prod_{j\in[n]} \mathsf{C}_{\mathtt{E}_j}\text{-}\mathsf{Z}(\delta_j)_{\to\mathtt{T}}\right)|e\rangle_{\mathtt{E}}|b\rangle_{\mathtt{T}}.
\end{align*}
The above identity holds because at most one controlled gate $\mathsf{C}$-$\mathsf{U}_j$ is ``active'' for the state $|e\rangle_{\mathtt{E}}$. This allows us to group the $\mathsf{Z}$ operators from the $\mathsf{Z}$-decomposition of $\mathsf{U}_0,\dots,\mathsf{U}_{n-1}$ as shown in \Cref{fig:circuit_reorderingb}.

The idea of our quantum circuits is to compute the one-hot encoding of the input string $x\in\{0,1\}^n$, after which the correct phases $\alpha_j,\beta_j,\gamma_j,\delta_j$ can be selected. We note that Low et al.~\cite{Low2024tradingtgatesdirty} proposed a $O(n^2)$-depth circuit with $O(2^n)$ single and two-qubit gates and no ancillae to compute the one-hot encoding of $x\in\{0,1\}^n$.

\begin{figure}
    \centering
    \begin{subfigure}[b]{0.7\textwidth}
        \includegraphics[trim={1.4cm 24cm 10cm 0.7cm},clip,width=\textwidth]{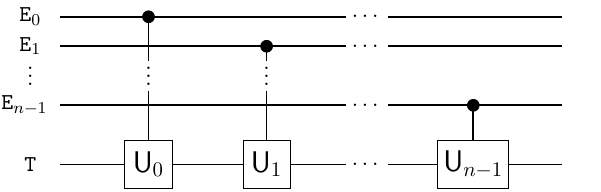}
        \caption{}
        \label{fig:circuit_reorderinga}
    \end{subfigure}
    \begin{subfigure}[b]{\textwidth}
        \includegraphics[trim={1.4cm 23.5cm 1cm 0.7cm},clip,width=\textwidth]{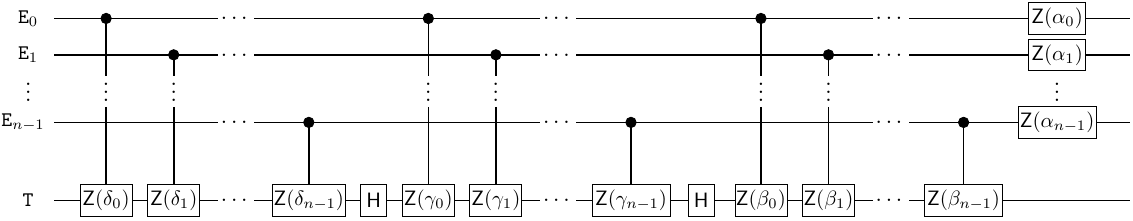}
        \caption{}
        \label{fig:circuit_reorderingb}
    \end{subfigure}
    \caption{(a) A serial circuit of unitaries $\mathsf{U}_j = e^{i\pi \alpha_j}\mathsf{Z}(\beta_j)\mathsf{H}\mathsf{Z}(\gamma_j)\mathsf{H}\mathsf{Z}(\delta_j)$ controlled on the single-qubit register $\mathtt{E}_j$ being in the $|1\rangle$ state, $j\in[n]$. (b) If the control qubits $|e\rangle_{\mathtt{E}} = \bigotimes_{j\in[n]}|e_j\rangle_{\mathtt{E}_j}$ are such that $e\in\{0,1\}^n$ has Hamming weight at most $1$ ($|e|\leq 1$), then the gates composing $\mathsf{U}_0,\dots,\mathsf{U}_{n-1}$ can be rearranged as shown in the equivalent circuit.}
    \label{fig:circuit_reordering}
\end{figure}

\subsection{Constant-depth circuits for $f$-$\mathsf{UCG}$s}

\begin{theorem}[One-hot-encoding implementation of $f$-$\mathsf{UCG}$]\label{thr:ucg_construction}
    Let $f:\{0,1\}^n\to\mathcal{U}(\mathbb{C}^{2\times 2})$ be a $(J,r)$-junta for $J\subseteq[n]$ with $|\overline{{J}}| = t$ and $r\in\mathbb{N}$. There is a $O(1)$-depth circuit for $f$-$\mathsf{UCG}$ that uses
    \begin{itemize}
        \item either $2(t+r)2^{t+r}\log(t+r) + O((t+r)2^{t+r})$ ancillae and $2n + 6(t+r)2^{t+r} + O(2^{t+r}\log(t+r))$ Fan-Out gates with arity $\leq 1+2^{t+r}$,
        \item or $3(t+r)2^{t+r} + O(2^{t+r}\log(t+r))$ ancillae and $9$ $\mathsf{GT}$ gates with arity $\leq n+(t+r)2^{t+r} + O(2^{t+r}\log(t+r))$.
    \end{itemize}
\end{theorem}
\begin{proof}
    Given the initial state $|x\rangle_{\mathtt{I}}|b\rangle_{\mathtt{T}}$ for $x\in\{0,1\}^n$ and $b\in\{0,1\}$, we wish to perform the mapping $|x\rangle_{\mathtt{I}}|b\rangle_{\mathtt{T}} \mapsto |x\rangle_{\mathtt{I}}f(x)|b\rangle_{\mathtt{T}}$. For each $z\in\{0,1\}^{t}$, let $J_z\subseteq J$, with $|J_z| \leq r$, be the subset of coordinates that $f_{J|z}$ depends on. For $z\in\{0,1\}^t$, let $g_z:\{0,1\}^{|J_z|}\to \mathcal{U}(\mathbb{C}^{2\times 2})$ be such that $f_{J|z}(x_J) = g_z(x_{J_z})$. In the following, split the register $\mathtt{I}$ into registers $\mathtt{\overline{J}}$ and $\mathtt{J}$ such that $\mathtt{\overline{J}}$ contains the coordinates of $x$ in $\overline{J}$ and $\mathtt{J}$ contains the coordinates of $x$ in $J$, i.e., $|x\rangle_{\mathtt{I}} = |x_{\overline{J}}\rangle_{\mathtt{\overline{J}}}|x_{J}\rangle_{\mathtt{J}}$. For $i\in J$, let $m_i := \sum_{z\in\{0,1\}^t:i\in J_z} 2^{|J_z|} \leq 2^r\cdot|\{z\in\{0,1\}^t:i\in J_z\}|$ and $m:= \sum_{z\in\{0,1\}^t} 2^{|J_z|} \leq 2^{t+r}$. Let $f(x) = e^{i\pi\alpha(x)}\mathsf{Z}(\beta(x))\mathsf{H}\mathsf{Z}(\gamma(x))\mathsf{H}\mathsf{Z}(\delta(x))$ be the $\mathsf{Z}$-decomposition of each single-qubit gate $f(x)$, $x\in\{0,1\}^n$. Equivalently, write $g_z(j) = e^{i\pi\alpha_{z}(j)}\mathsf{Z}(\beta_{z}(j))\mathsf{H}\mathsf{Z}(\gamma_{z}(j))\mathsf{H}\mathsf{Z}(\delta_{z}(j))$ for the $\mathsf{Z}$-decomposition of the single-qubit gate $g_z(j)$, $z\in\{0,1\}^t$ and $j\in\{0,1\}^{|J_z|}$. From $f_{J|z}(x_J) = g_z(x_{J_z})$ we can establish the correspondence that $\alpha(x) = \alpha_z(j)$ for all $z\in\{0,1\}^t$, $j\in\{0,1\}^{|J_z|}$, and $x\in\{0,1\}^n$ such that $x_{\overline{J}} = z$ and $x_{J_z} = j$ (and similarly for $\beta,\gamma,\delta$).
    
    The main idea of the circuits is to compute the one-hot encoding of a compressed version of $x$. Naively, one could compute the one-hot encoding of the whole string $x$. However, by breaking $x$ into the sub-strings $x_{\overline{J}}$ and $x_J$ indexed by $\overline{J}$ and $J$, respectively, it is possible to compute the one-hot encoding of $x_{\overline{J}}$ separately from the one-hot encoding of $x_{J}$. In principle, this would not offer any advantage, but since $f_{J|z}$ depends only on a few coordinates of $x_J$, we can compute the one-hot encoding of the much shorter sub-string $x_{J_z}$ for $z=x_{\overline{J}}$ instead. Since we do not know $x_{\overline{J}}$, we must compute the one-hot encoding of $x_{J_z}$ for all $z\in\{0,1\}^t$. We first consider the Fan-Out-based circuit (see \Cref{fig:ucg_fanout}) and then the one based on $\mathsf{GT}$ gates. The algorithm is as follows:
    \begin{figure}
        \centering
        \includegraphics[trim={1.4cm 1.0cm 3.4cm 0.7cm},clip,width=0.79\textwidth]{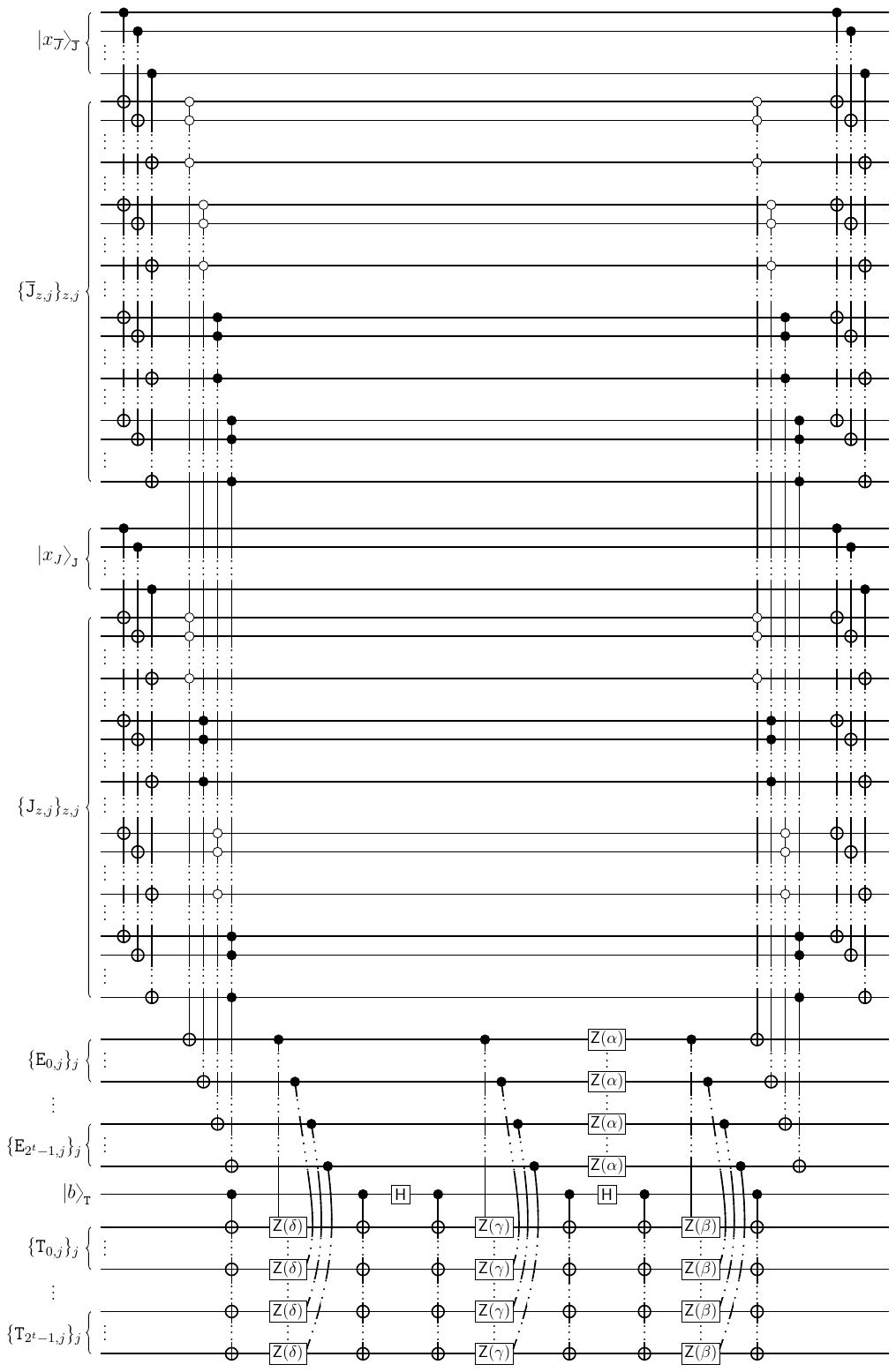}
        \caption{The circuit for an $f$-$\mathsf{UCG}^{(n)}$ using Fan-Out gates, where $f$ is a $(J,r)$-junta with $|\overline{J}| = t$. For simplicity, we include targets from $|x_i\rangle_{\mathtt{J}}$ onto all registers $\{\mathtt{J}_{z,j}\}_{z,j}$, but in reality $x_i$ is copied onto the registers $\{\mathtt{J}_{z,j}\}_j$ for all $z\in\{0,1\}^t$ such that $i\in J_z$, where $J_z$ is the set of coordinated that $f_{J|z}$ depends on. Moreover, we omit the indices in the parameters $\alpha_{z}(j)$, $\beta_{z}(j)$, $\gamma_{z}(j)$, $\delta_{z}(j)$. %The registers $\mathtt{E}_{z,j}$ and $\mathtt{T}_{z,j}$ have both size $|J_z|$.
        }
        \label{fig:ucg_fanout}
    \end{figure}
    \begin{enumerate}
        \item For each $z\in\{0,1\}^t$ and $j\in\{0,1\}^{|J_z|}$, attach a $t$-qubit ancillary register $\mathtt{\overline{J}}_{z,j}$ and copy the contents of the $|x_{\overline{J}}\rangle_{\mathtt{\overline{J}}}$ register onto it. Every qubit of $|x_{\overline{J}}\rangle_{\mathtt{\overline{J}}}$ is thus copied $m$ times.

        \item For each $z\in\{0,1\}^t$ and $j\in\{0,1\}^{|J_z|}$, attach a $|J_z|$-qubit ancillary register $\mathtt{J}_{z,j}$. The registers $\mathtt{J}_{z,j}$, for $j\in\{0,1\}^{|J_z|}$, will be used to compute the one-hot encoding of $x_{J_z}$. For each $i\in J$ in parallel, copy  $m_i$ number of times the qubit $|x_i\rangle_{\mathtt{J}}$ onto the registers $\{\mathtt{J}_{z,j}\}_j$, where $j\in\{0,1\}^{|J_z|}$, for all $z$ such that $i\in J_z$. Steps~$1$ and~$2$ lead to
        \begin{align*}
            |x\rangle_{\mathtt{I}}|b\rangle_{\mathtt{T}} \mapsto   |x\rangle_{\mathtt{I}}|b\rangle_{\mathtt{T}} \left(\bigotimes_{z\in\{0,1\}^t} \bigotimes_{j\in\{0,1\}^{|J_z|}}|x_{\overline{J}}\rangle_{\mathtt{\overline{J}}_{z,j}} |x_{J_z}\rangle_{\mathtt{J}_{z,j}} \right).
        \end{align*}
        
        \item For each $z\in\{0,1\}^t$ and $j\in\{0,1\}^{|J_z|}$, apply the gate $\bigotimes_{k\in[t]} \mathsf{X}^{\overline{z}_k}$ onto $|x_{\overline{J}}\rangle_{\mathtt{\overline{J}}_{z,j}}$ and the gate $\bigotimes_{k\in[|J_z|]} \mathsf{X}^{\overline{j}_k}$ onto $|x_{J_z}\rangle_{\mathtt{{J}}_{z,j}}$. This leads to
        \begin{align*}
            |x\rangle_{\mathtt{I}}|b\rangle_{\mathtt{T}} \left(\bigotimes_{z\in\{0,1\}^t}\bigotimes_{j\in\{0,1\}^{|J_z|}} |x_{\overline{J}}\oplus\overline{z}\rangle_{\mathtt{\overline{J}}_{z,j}}|x_{J_z}\oplus\overline{j}\rangle_{\mathtt{J}_{z,j}} \right).
        \end{align*}

        \item For each $z\in\{0,1\}^t$ and $j\in\{0,1\}^{|J_z|}$, attach a single-qubit ancillary register $\mathtt{E}_{z,j}$ and apply an $\mathsf{AND}^{(t+|J_z|)}_{\{\mathtt{\overline{J}}_{z,j},\mathtt{J}_{z,j}\}\to\mathtt{E}_{z,j}}$ gate from registers $\mathtt{\overline{J}}_{z,j}$ and $\mathtt{J}_{z,j}$ onto register $\mathtt{E}_{z,j}$ to obtain
        \begin{align*}
            &|x\rangle_{\mathtt{I}}|b\rangle_{\mathtt{T}} \left(\bigotimes_{z\in\{0,1\}^t} \bigotimes_{j\in\{0,1\}^{|J_z|}} |x_{\overline{J}}\oplus\overline{z}\rangle_{\mathtt{\overline{J}}_{z,j}}|x_{J_z}\oplus\overline{j}\rangle_{\mathtt{J}_{z,j}} |0\rangle_{\mathtt{E}_{z,j}}\right) \\
            \mapsto~ &|x\rangle_{\mathtt{I}}|b\rangle_{\mathtt{T}} \left(\bigotimes_{z\in\{0,1\}^t} \bigotimes_{j\in\{0,1\}^{|J_z|}} |x_{\overline{J}}\oplus\overline{z}\rangle_{\mathtt{\overline{J}}_{z,j}}|x_{J_z}\oplus\overline{j}\rangle_{\mathtt{J}_{z,j}} \big|{\bigwedge}_{k\in[t]}(x_{\overline{J}}\oplus\overline{z})_k\cdot {\bigwedge}_{l\in[|J_z|]}(x_{J_z}\oplus\overline{j})_l\big\rangle_{\mathtt{E}_{z,j}}\right)\\
            =~ &|x\rangle_{\mathtt{I}}|b\rangle_{\mathtt{T}} \left(\bigotimes_{z\in\{0,1\}^t} \bigotimes_{j\in\{0,1\}^{|J_z|}} |x_{\overline{J}}\oplus\overline{z}\rangle_{\mathtt{\overline{J}}_{z,j}}|x_{J_z}\oplus\overline{j}\rangle_{\mathtt{J}_{z,j}} |e(x_{\overline{J}})_z\cdot e(x_{J_z})_j\rangle_{\mathtt{E}_{z,j}}\right),
        \end{align*}
        where $e(x_{\overline{J}}) \in \{0,1\}^{2^t}$ and $e(x_{J_z}) \in \{0,1\}^{2^{|J_z|}}$ are the one-hot encodings of $x_{\overline{J}}$ and $x_{J_z}$, respectively. Note that $e(x_{\overline{J}})_z$ is the $z$-th bit of $e(x_{\overline{J}})$ and $e(x_{J_z})_j$ is $j$-th bit of $e(x_{J_z})$.

        \item [5a.] For each $z\in\{0,1\}^t$ and $j\in\{0,1\}^{|J_z|}$, attach a single-qubit ancillary register $\mathtt{T}_{z,j}$. Apply a $(1+m)$-arity Fan-Out gate $\mathsf{FO}^{(1+m)}_{\mathtt{T}\to\{\mathtt{T}_{z,j}\}_{z,j}}$ from register $\mathtt{T}$ onto registers $\{\mathtt{T}_{z,j}\}_{z,j}$, $z\in\{0,1\}^t$ and $j\in\{0,1\}^{|J_z|}$ (remember that $m:= \sum_{z\in\{0,1\}^t}2^{|J_z|}$). For each $z\in\{0,1\}^t$ and $j\in\{0,1\}^{|J_z|}$, apply a $\mathsf{C}_{\mathtt{E}_{z,j}}$-$\mathsf{Z}(\delta_{z}(j))_{\to\mathtt{T}_{z,j}}$ gate controlled on register $\mathtt{E}_{z,j}$ onto register $\mathtt{T}_{z,j}$. Finally, apply $\mathsf{FO}^{(1+m)}_{\mathtt{T}\to\{\mathtt{T}_{z,j}\}_{z,j}}$ again. We shall omit the registers $\mathtt{\overline{J}}_{z,j}$ and $\mathtt{J}_{z,j}$ from now on for simplicity. This chain of operations leads to
        \begin{align*}
            &|x\rangle_{\mathtt{I}}|b\rangle_{\mathtt{T}} \left(\bigotimes_{z\in\{0,1\}^t} \bigotimes_{j\in\{0,1\}^{|J_z|}}|e(x_{\overline{J}})_z e(x_{J_z})_j\rangle_{\mathtt{E}_{z,j}}|b\rangle_{\mathtt{T}_{z,j}}\right)\\
            \mapsto~ &|x\rangle_{\mathtt{I}}|b\rangle_{\mathtt{T}} \left(\bigotimes_{z\in\{0,1\}^t} \bigotimes_{j\in\{0,1\}^{|J_z|}} |e(x_{\overline{J}})_z e(x_{J_z})_j\rangle_{\mathtt{E}_{z,j}}\mathsf{Z}(\delta_{z}(j))^{e(x_{\overline{J}})_z e(x_{J_z})_j}|b\rangle_{\mathtt{T}_{z,j}}\right)\\
            \mapsto~ &|x\rangle_{\mathtt{I}}\mathsf{Z}(\delta(x))|b\rangle_{\mathtt{T}} \left(\bigotimes_{z\in\{0,1\}^t} \bigotimes_{j\in\{0,1\}^{|J_z|}} |e(x_{\overline{J}})_z e(x_{J_z})_j\rangle_{\mathtt{E}_{z,j}}\right),
        \end{align*}
        where we have used the definition of one-hot encoding, i.e., $e(x_{\overline{J}})_z = 1$ if and only if $z=x_{\overline{J}}$ and $e(x_{J_z})_j = 1$ if and only if $j=x_{J_z}$, and also that $\delta_{z}(j) = \delta(x)$ for $z=x_{\overline{J}}$ and $j=x_{J_z}$.
        \item [5b.] Apply a $\mathsf{H}_{\to\mathtt{T}}$ gate to register $\mathtt{T}$ followed by a $(1+m)$-arity Fan-Out gate $\mathsf{FO}^{(1+m)}_{\mathtt{T}\to\{\mathtt{T}_{z,j}\}_{z,j}}$ from register $\mathtt{T}$ to registers $\{\mathtt{T}_{z,j}\}_{z,j}$. For each $z\in\{0,1\}^t$ and $j\in\{0,1\}^{|J_z|}$, apply a $\mathsf{C}_{\mathtt{E}_{z,j}}$-$\mathsf{Z}(\gamma_{z}(j))_{\to\mathtt{T}_{z,j}}$ gate controlled on register $\mathtt{E}_{z,j}$ onto register $\mathtt{T}_{z,j}$. Finally, apply $\mathsf{FO}^{(1+m)}_{\mathtt{T}\to\{\mathtt{T}_{z,j}\}_{z,j}}$ again. For simplicity, write $\mathsf{H} \mathsf{Z}(\delta(x))|b\rangle_{\mathtt{T}} = r_{b,x}|0\rangle_{\mathtt{T}} + s_{b,x}|1\rangle_{\mathtt{T}}$. This chain of operations yields
        \begin{align*}
            &~~~~|x\rangle_{\mathtt{I}}\mathsf{H}\mathsf{Z}(\delta(x))|b\rangle_{\mathtt{T}} \left(\bigotimes_{z\in\{0,1\}^t} \bigotimes_{j\in\{0,1\}^{|J_z|}} |e(x_{\overline{J}})_z e(x_{J_z})_j\rangle_{\mathtt{E}_{z,j}}\right) \\
            &\begin{multlined}[b][0.94\textwidth]
                \mapsto r_{b,x}|x\rangle_{\mathtt{I}}|0\rangle_{\mathtt{T}}\bigotimes_{z\in\{0,1\}^t}\bigotimes_{j\in\{0,1\}^{|J_z|}}\big(|e(x_{\overline{J}})_z e(x_{J_z})_j\rangle_{\mathtt{E}_{z,j}}|0\rangle_{\mathtt{T}_{z,j}}\big) \\
                + s_{b,x}|x\rangle_{\mathtt{I}}|1\rangle_{\mathtt{T}}\bigotimes_{z\in\{0,1\}^t}\bigotimes_{j\in\{0,1\}^{|J_z|}}\big(|e(x_{\overline{J}})_z e(x_{J_z})_j\rangle_{\mathtt{E}_{z,j}}|1\rangle_{\mathtt{T}_{z,j}}\big)
            \end{multlined}\\
            &\begin{multlined}[b][0.94\textwidth]
                \mapsto r_{b,x}|x\rangle_{\mathtt{I}}|0\rangle_{\mathtt{T}}\bigotimes_{z\in\{0,1\}^t}\bigotimes_{j\in\{0,1\}^{|J_z|}}\big(|e(x_{\overline{J}})_z e(x_{J_z})_j\rangle_{\mathtt{E}_{z,j}}\mathsf{Z}(\gamma_{z}(j))^{e(x_{\overline{J}})_z e(x_{J_z})_j}|0\rangle_{\mathtt{T}_{z,j}}\big) \\
                + s_{b,x}|x\rangle_{\mathtt{I}}|1\rangle_{\mathtt{T}}\bigotimes_{z\in\{0,1\}^t}\bigotimes_{j\in\{0,1\}^{|J_z|}}\big(|e(x_{\overline{J}})_z e(x_{J_z})_j\rangle_{\mathtt{E}_{z,j}}\mathsf{Z}(\gamma_{z}(j))^{e(x_{\overline{J}})_z e(x_{J_z})_j}|1\rangle_{\mathtt{T}_{z,j}}\big)
            \end{multlined}\\
            &\mapsto |x\rangle_{\mathtt{I}}\mathsf{Z}(\gamma(x))\mathsf{H}\mathsf{Z}(\delta(x))|b\rangle_{\mathtt{T}} \left(\bigotimes_{z\in\{0,1\}^t} \bigotimes_{j\in\{0,1\}^{|J_z|}} |e(x_{\overline{J}})_z e(x_{J_z})_j\rangle_{\mathtt{E}_{z,j}}\right).
        \end{align*}

        \item [5c.] Apply a $\mathsf{H}_{\to\mathtt{T}}$ gate to register $\mathtt{T}$ followed by a $(1+m)$-arity Fan-Out gate $\mathsf{FO}^{(1+m)}_{\mathtt{T}\to\{\mathtt{T}_{z,j}\}_{z,j}}$ from register $\mathtt{T}$ to registers $\{\mathtt{T}_{z,j}\}_{z,j}$. For each $z\in\{0,1\}^t$ and $j\in\{0,1\}^{|J_z|}$, apply a $\mathsf{C}_{\mathtt{E}_{z,j}}$-$\mathsf{Z}(\beta_{z}(j))_{\to\mathtt{T}_{z,j}}$ gate controlled on register $\mathtt{E}_{z,j}$ onto register $\mathtt{T}_{z,j}$ followed by a $\mathsf{Z}(\alpha_{z}(j))_{\to\mathtt{E}_{z,j}}$ gate applied onto register $\mathtt{E}_{z,j}$. Finally, apply $\mathsf{FO}^{(1+m)}_{\mathtt{T}\to\{\mathtt{T}_{z,j}\}_{z,j}}$ again. Similarly to the previous step, we get (consider only registers $\mathtt{I}$ and $\mathtt{T}$ for simplicity)
        \begin{align*}
            |x\rangle_{\mathtt{I}}\mathsf{Z}(\gamma(x))\mathsf{H}\mathsf{Z}(\delta(x))|b\rangle_{\mathtt{T}} \mapsto |x\rangle_{\mathtt{I}}e^{i\pi\alpha(x)}\mathsf{Z}(\beta(x))\mathsf{H}\mathsf{Z}(\gamma(x))\mathsf{H}\mathsf{Z}(\delta(x))|b\rangle_{\mathtt{T}} = |x\rangle_{\mathtt{I}}f(x)|b\rangle_{\mathtt{T}}.
        \end{align*}
        \item [6.] Uncompute Steps $4$, $3$, $2$, and $1$. This leads to the desired state $|x\rangle_{\mathtt{I}}f(x)|b\rangle_{\mathtt{T}}$.
    \end{enumerate}

    We now consider the $\mathsf{GT}$-gate-based circuit (see \Cref{fig:ucg_gt}), which is basically the same as the Fan-Out-based circuit, but replacing Steps $5$a, $5$b, and $5$c with the following Step~$5$:
    \begin{figure}
        \centering
        \includegraphics[trim={1.4cm 5.3cm 0.7cm 0.7cm},clip,width=\textwidth]{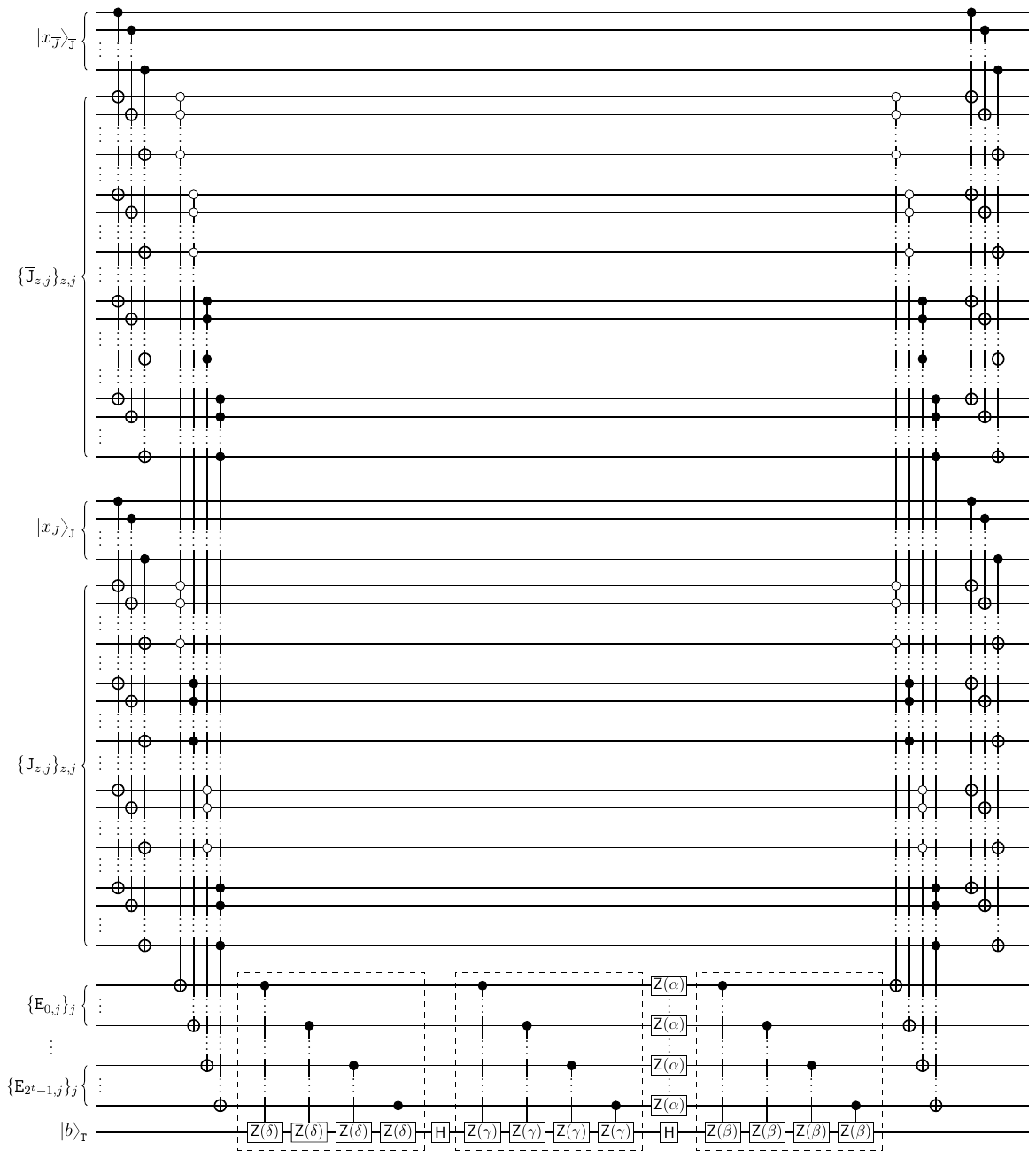}
        \caption{The circuit for an $f$-$\mathsf{UCG}^{(n)}$ using $\mathsf{GT}$ gates, where $f$ is a $(J,r)$-junta with $|\overline{J}| = t$. We highlight the $\mathsf{GT}$ gates inside the dashed boxes. For simplicity, we include targets from $|x_i\rangle_{\mathtt{J}}$ onto all registers $\{\mathtt{J}_{z,j}\}_{z,j}$, but in reality $x_i$ is copied onto the registers $\{\mathtt{J}_{z,j}\}_j$ for all $z\in\{0,1\}^t$ such that $i\in J_z$, where $J_z$ is the set of coordinated that $f_{J|z}$ depends on. Moreover, we omit the indices in the parameters $\alpha_{z}(j)$, $\beta_{z}(j)$, $\gamma_{z}(j)$, $\delta_{z}(j)$.}
        \label{fig:ucg_gt}
    \end{figure}
    \begin{enumerate}
        \item [5.] Apply the gate
        \begin{align*}
        \begin{multlined}[b][0.94\textwidth]
            \left(\prod_{z\in\{0,1\}^t}\prod_{j\in\{0,1\}^{|J_z|}}\mathsf{Z}(\alpha_{z}(j))_{\to\mathtt{E}_{z,j}} \right)\left(\prod_{z\in\{0,1\}^t}\prod_{j\in\{0,1\}^{|J_z|}}\mathsf{C}_{\mathtt{E}_{z,j}}\text{-}\mathsf{Z}(\beta_{z}(j))_{\to\mathtt{T}}\right)\mathsf{H}_{\to\mathtt{T}}\\
            \cdot\left(\prod_{z\in\{0,1\}^t}\prod_{j\in\{0,1\}^{|J_z|}}\mathsf{C}_{\mathtt{E}_{z,j}}\text{-}\mathsf{Z}(\gamma_{z}(j))_{\to\mathtt{T}}\right)\mathsf{H}_{\to\mathtt{T}}\left(\prod_{z\in\{0,1\}^t}\prod_{j\in\{0,1\}^{|J_z|}}\mathsf{C}_{\mathtt{E}_{z,j}}\text{-}\mathsf{Z}(\delta_{z}(j))_{\to\mathtt{T}}\right)
        \end{multlined}
        \end{align*}
        using $3$ $\mathsf{GT}$ gates (one for each $\prod_{z\in\{0,1\}^t}\prod_{j\in\{0,1\}^{|J_z|}}\mathsf{C}_{\mathtt{E}_{z,j}}\text{-}\mathsf{Z}(\cdot)_{\to\mathtt{T}}$). Since $|e(x_{\overline{{J}}})| = |e(x_{J_z})| = 1$, this leads to $|x\rangle_{\mathtt{I}}f(x)|b\rangle_{\mathtt{T}}$ up to the ancillary registers.
    \end{enumerate}

    We now analyse the resources required for each step:
    \begin{itemize}
        \item Step $1$: the registers $\mathtt{\overline{J}}_{z,j}$, for $z\in\{0,1\}^t$ and $j\in\{0,1\}^{|J_z|}$, use at most $t2^{t+r}$ ancillae and copying the register $\mathtt{\overline{J}}$ requires either $t$ Fan-Out gates with arity at most $1+2^{t+r}$ or $1$ $\mathsf{GT}$ gate with arity at most $t(1+2^{t+r})$;
        
        \item Step $2$: the registers $\mathtt{J}_{z,j}$, for $z\in\{0,1\}^t$ and $j\in\{0,1\}^{|J_z|}$, use at most $r2^{t+r}$ ancillae and copying the register $\mathtt{J}$ requires either $|J|=n-t$ Fan-Out gates with arity at most $1+\max_{i\in J}m_i \leq 1+2^{t+r}$ or $1$ $\mathsf{GT}$ (which can be absorbed by the one from Step $1$) with arity at most $n-t + \sum_{i\in J} m_i \leq n-t + r2^{t+r}$;
        
        \item Step $4$: the registers $\mathtt{E}_{z,t}$, for $z\in\{0,1\}^t$ and $j\in\{0,1\}^{|J_z|}$, use at most $2^{t+r}$ ancillae. The $m\leq 2^{t+r}$ $\mathsf{AND}^{(t+|J_z|)}_{\{\mathtt{\overline{J}}_{z,j},\mathtt{J}_{z,j}\}\to\mathtt{E}_{z,j}}$ gates require either $2(t+r)2^{t+r}\log{(t+r)} + O((t+r)2^{t+r})$ ancillae by using the construction based on Fan-Out gates (\Cref{thr:or_constantdepth}) or $2(t+r)2^{t+r} + O(2^{t+r}\log(t+r))$ ancillae by using the construction based on $\mathsf{GT}$ gates (\Cref{thr:or_constantdepth_gt}). Naively, one would expect the $m\leq 2^{t+r}$ $\mathsf{AND}^{(t+|J_z|)}_{\{\mathtt{\overline{J}}_{z,j},\mathtt{J}_{z,j}\}\to\mathtt{E}_{z,j}}$ gates to require either $6(t+r)2^{t+r} + O(2^{t+r}\log(t+r))$ Fan-Out gates with arity at most $2(t+r)$ (\Cref{thr:or_constantdepth}) or $4$ $\mathsf{GT}$ gates with arity at most $(t+r)2^{t+r} + O(2^{t+r}\log(t+r))$ (\Cref{thr:or_constantdepth_gt}), but we can postpone their inner uncomputation part until Step~$6$ and carry over all the required ancillae. This means that Step~$4$ requires only $3(t+r)2^{t+r} + O(2^{t+r}\log(t+r))$ Fan-Out gates or $2$ $\mathsf{GT}$ gates;
        
        \item Step $5$: the Fan-Out-based circuit requires at most $2^{t+r}$ ancillae for the registers $\mathtt{T}_{z,j}$, where $z\in\{0,1\}^t$ and $j\in\{0,1\}^{|J_z|}$, and $6$ Fan-Out gates with arity at most $1+2^{t+r}$. The $\mathsf{GT}$-gate-based circuit does not require any ancillae, and only $3$ $\mathsf{GT}$ gates with arity $\leq 1+2^{t+r}$;
        
        \item Step $6$: uncomputing uses either $n + 3(t+r)2^{t+r} + O(2^{t+r}\log(t+r))$ Fan-Outs or $3$ $\mathsf{GT}$ gates.
    \end{itemize} 

    In total, we require either $2(t+r)2^{t+r}\log(t+r) + O((t+r)2^{t+r})$ ancillae and $2n + 6(t+r)2^{t+r} + O(2^{t+r}\log(t+r))$ Fan-Out gates with arity at most $1+2^{t+r}$, or $3(t+r)2^{t+r} + O(2^{t+r}\log(t+r))$ ancillae and $9$ $\mathsf{GT}$ gates with arity at most $n+(t+r)2^{t+r} + O(2^{t+r}\log(t+r))$.
\end{proof}

\subsection{Constant-depth circuits for $f$-$\mathsf{FIN}$}
\label{sec:fin_onehot}

The circuits from the previous section can be used to implement an $f$-$\mathsf{FIN}$, since they are a special case of $f$-$\mathsf{UCG}$s, as explained before \Cref{def:ffingate}. Nonetheless, the circuits from the previous section can be simplified due to their simpler structure, i.e., the $\mathsf{Z}$-decomposition of an $f$-$\mathsf{FIN}$ is $\mathsf{H}\mathsf{Z}(f(x))\mathsf{H} = \mathsf{X}^{f(x)}$. In particular, the controlled gates $\mathsf{H}_{\to\mathtt{T}}\mathsf{C}_{\mathtt{E}_{z,j}}$-$\mathsf{Z}(\gamma_z(j))_{\rightarrow \mathtt{T}}\mathsf{H}_{\to\mathtt{T}} = \mathsf{C}_{\mathtt{E}_{z,j}}$-$\mathsf{X}_{\rightarrow \mathtt{T}}^{\gamma_z(j)}$, where $\gamma_z:\{0,1\}^{|J_z|}\to\{0,1\}$, that arise from the $\mathsf{Z}$-decomposition can be replaced by a single $\mathsf{PARITY}$ gate (recall that the $\mathsf{PARITY}$ gate can be implemented by a single Fan-Out gate (\Cref{fact:fanoutparity})).
We show how this can be done in the next result.
\begin{theorem}[One-hot-encoding implementation of $f$-$\mathsf{FIN}$]\label{thr:fin_onehot}
    Let $f:\{0,1\}^n\to\{0,1\}$ be a $(J,r)$-junta for $J\subseteq[n]$ with $|\overline{{J}}| = t$ and $r\in\mathbb{N}$. There is a $O(1)$-depth circuit for $f$-$\mathsf{FIN}$ that uses
    \begin{itemize}
        \item either $2(t+r)2^{t+r}\log(t+r) + O((t+r)2^{t+r})$ ancillae and $2n + 6(t+r)2^{t+r} + O(2^{t+r}\log(t+r))$ Fan-Out gates with arity $\leq 1+2^{t+r}$,
        \item or $3(t+r)2^{t+r} + O(2^{t+r}\log(t+r))$ ancillae and $6$ $\mathsf{GT}$ gates with arity $\leq n+(t+r)2^{t+r} + O(2^{t+r}\log(t+r))$.
    \end{itemize}
\end{theorem}
\begin{proof}
    Construct the state
    \begin{align*}
        |x\rangle_{\mathtt{I}}|b\rangle_{\mathtt{T}} \left(\bigotimes_{z\in\{0,1\}^t} \bigotimes_{j\in\{0,1\}^{|J_z|}} |x_{\overline{J}}\oplus\overline{z}\rangle_{\mathtt{\overline{J}}_{z,j}}|x_{J_z}\oplus\overline{j}\rangle_{\mathtt{J}_{z,j}} |e(x_{\overline{J}})_z\cdot e(x_{J_z})_j\rangle_{\mathtt{E}_{z,j}}\right)
    \end{align*}
    by following the same steps as in \Cref{thr:ucg_construction}. To perform the $f$-$\mathsf{FIN}$ gate, we must apply a $\mathsf{C}_{\mathtt{E}_{z,j}}$-$\mathsf{X}_{\to\mathtt{T}}$ gate for all $z\in\{0,1\}^t$ and $j\in g_{z}^{-1}(1)$, where $g_z:\{0,1\}^{|J_z|}\to \{0,1\}$ is such that $f_{J|z}(x_J) = g_z(x_{J_z})$, since this leads to (consider only register $\mathtt{T}$)
    \begin{align*}
        |b\rangle_{\mathtt{T}} \mapsto \big|b\oplus {\bigoplus}_{z\in\{0,1\}^t}{\bigoplus}_{j\in g^{-1}_{z}(1)}e(x_{\overline{J}})_z\cdot e(x_{J_z})_j \big\rangle_{\mathtt{T}} = |b\oplus g_{x_{\overline{J}}}(x_{J_{x_{\overline{J}}}})\rangle_{\mathtt{T}} = |b\oplus f(x) \rangle_{\mathtt{T}},
    \end{align*}
    where we used (i) the definition of one-hot encoding, $e(x_{\overline{J}})_z = 1$ if and only if $z=x_{\overline{J}}$, and $e(x_{J_z})_j = 1$ if and only if $j=x_{J_z}$; (ii) the fact that $\bigoplus_{j\in g^{-1}_{z}(1)} e(y)_j = g_{z}(y)$ for any $z\in\{0,1\}^{t}$ and $y\in \{0,1\}^{|J_z|}$; (iii) the identity $g_{x_{\overline{J}}}(x_{J_{x_{\overline{J}}}}) = f_{J|x_{\overline{J}}}(x_{J}) = f(x)$.

    \begin{figure}
        \centering
        \includegraphics[trim={1.45cm 23.5cm 2.2cm 0.7cm},clip,width=\textwidth]{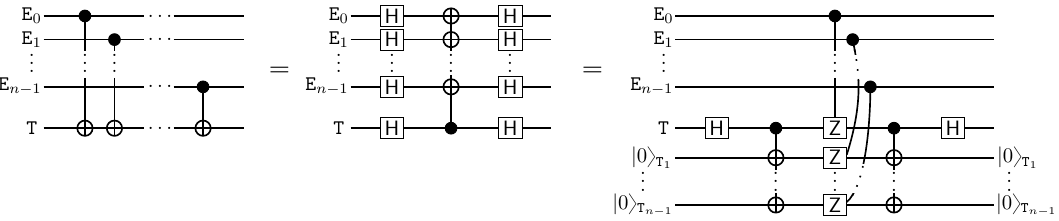}
        \caption{The gate $\prod_{j\in[n]}\mathsf{C}_{\mathtt{E}_j}\text{-}\mathsf{X}_{\to\mathtt{T}}$ (first circuit) is equivalent to $\mathsf{PARITY}_{\{\mathtt{E}_j\}_{j\in[n]}\to\mathtt{T}}$ (second circuit) and to $\mathsf{H}_{\to\mathtt{T}}\mathsf{FO}_{\mathtt{T}\to\{\mathtt{T}_{j}\}_{j=1}^{n-1}}\big(\prod_{j\in[n]}\mathsf{C}_{\mathtt{E}_{j}}\text{-}\mathsf{Z}_{\to\mathtt{T}_{j}}\big)\mathsf{FO}_{\mathtt{T}\to\{\mathtt{T}_{j}\}_{j=1}^{n-1}}\mathsf{H}_{\to\mathtt{T}}$ (third circuit) (define $\mathtt{T}_0 := \mathtt{T}$), where the ancillary registers $\{\mathtt{T}_j\}_{j=1}^{n-1}$ are initialized in the $|0\rangle$ state.
        }
        \label{fig:two_methods}
    \end{figure}
    
    There are two methods to apply the $\mathsf{C}_{\mathtt{E}_{z,j}}$-$\mathsf{X}_{\to\mathtt{T}}$ gates in parallel (see \Cref{fig:two_methods}). The first method is via~\cite{moore1999quantum}
    \begin{align*}
        \prod_{z\in\{0,1\}^t}\prod_{j\in g_{z}^{-1}(1)}\mathsf{C}_{\mathtt{E}_{z,j}}\text{-}\mathsf{X}_{\to\mathtt{T}} = \mathsf{PARITY}_{\{\mathtt{E}_{z,j}\}_{z\in\{0,1\}^t,j\in g^{-1}_{z}(1)}\to\mathtt{T}},
    \end{align*}
    i.e., applying an $\mathsf{X}$ onto $|b\rangle_{\mathtt{T}}$ controlled on $|e(x_{\overline{J}})_z\cdot e(x_{J_z})_j\rangle_{\mathtt{E}_{z,j}}$ for all $z\in\{0,1\}^t$ and $j\in g_{z}^{-1}(1)$ is equivalent to applying a $\mathsf{PARITY}$ gate onto $|b\rangle_{\mathtt{T}}$ from input $\bigotimes_{z\in\{0,1\}^t} \bigotimes_{j\in g^{-1}_{z}(1)}  |e(x_{\overline{J}})_z e(x_{J_z})_j\rangle_{\mathtt{E}_{z,j}}$. The $\mathsf{PARITY}$ gate costs $1$ Fan-Out or $\mathsf{GT}$ gate with arity $1 + \sum_{z\in\{0,1\}^t}|g_{z}^{-1}(1)|$.  

    The second method is to use the parallelisation method from~\cite{green2001counting,moore2001parallel,hoyer2005quantum}. More specifically, by using the $\sum_{z\in\{0,1\}^t}|g^{-1}_z(1)| \leq 2^{t+r}$ ancillary registers $\mathtt{T}_{z,j}$, $z\in\{0,1\}^t$ and $j\in g^{-1}_z(1)$, initialized in the $|0\rangle$ state (note that we do not require all registers $\mathtt{T}_{z,j}$ from \Cref{thr:ucg_construction}), then
    \begin{align*}
    \begin{multlined}[\textwidth]
        \prod_{z\in\{0,1\}^t}\prod_{j\in g_{z}^{-1}(1)}\mathsf{C}_{\mathtt{E}_{z,j}}\text{-}\mathsf{X}_{\to\mathtt{T}} \\
        = \mathsf{H}_{\to\mathtt{T}}\mathsf{FO}_{\mathtt{T}\to\{\mathtt{T}_{z,j}\}_{z,j}}\left(\prod_{z\in\{0,1\}^t}\prod_{j\in g_{z}^{-1}(1)}\mathsf{C}_{\mathtt{E}_{z,j}}\text{-}\mathsf{Z}_{\to\mathtt{T}_{z,j}}\right)\mathsf{FO}_{\mathtt{T}\to\{\mathtt{T}_{z,j}\}_{z,j}}\mathsf{H}_{\to\mathtt{T}} \bigotimes_{z\in\{0,1\}^t}\bigotimes_{j\in g^{-1}_z(1)}|0\rangle_{\mathtt{T}_{z,j}}.
    \end{multlined}
    \end{align*}
    In more details, we have (consider only registers $\mathtt{T}$ and $\mathtt{T}_{z,j}$)
    \begin{align*}
        |b\rangle_{\mathtt{T}} &\mapsto \frac{1}{\sqrt{2}}|0\rangle_{\mathtt{T}}\bigotimes_{z\in\{0,1\}^t}\bigotimes_{j\in g^{-1}_z(1)}|0\rangle_{\mathtt{T}_{z,j}} + \frac{(-1)^b}{\sqrt{2}}|1\rangle_{\mathtt{T}}\bigotimes_{z\in\{0,1\}^t}\bigotimes_{j\in g^{-1}_z(1)}|1\rangle_{\mathtt{T}_{z,j}}\\
        &\mapsto \frac{1}{\sqrt{2}}|0\rangle_{\mathtt{T}}\bigotimes_{z\in\{0,1\}^t}\bigotimes_{j\in g^{-1}_z(1)}|0\rangle_{\mathtt{T}_{z,j}} + \frac{(-1)^b}{\sqrt{2}}|1\rangle_{\mathtt{T}}\bigotimes_{z\in\{0,1\}^t}\bigotimes_{j\in g^{-1}_z(1)}(-1)^{e(x_{\overline{J}})_ze(x_{J_z})_j}|1\rangle_{\mathtt{T}_{z,j}}\\
        &\mapsto \frac{|0\rangle_{\mathtt{T}} + (-1)^{b+f(x)}|1\rangle_{\mathtt{T}}}{\sqrt{2}} \\
        &\mapsto |b\oplus f(x)\rangle_{\mathtt{T}}.
    \end{align*}
    The above requires $2$ Fan-Out or $\mathsf{GT}$ gates with arity at most $1+2^{t+r}$. We crucially remark that the $\mathsf{GT}$ gates can be absorbed by the ones from computing and uncomputing registers $\mathtt{E}_{z,j}$. 
    
    The rest of the circuit is identical to \Cref{thr:ucg_construction}: uncompute the ancillary registers. The number of ancillae and Fan-Out gates is asymptotically the same as in \Cref{thr:ucg_construction}. The number of $\mathsf{GT}$ gates is reduced from $9$ to $7$ by using the first method or to $6$ by using the second method.
\end{proof}

\subsection{Constant-depth circuits for quantum memory devices via one-hot encoding}

In this section, we apply our previous circuit constructions based on one-hot encoding to the case of $\mathsf{QRAM}$ and $\QRAG$. As mentioned before, $\mathsf{QRAM}$ is simply an $f$-$\mathsf{FIN}$ with the Boolean function $f:\{0,1\}^n\times\{0,1\}^{\log{n}}\to\{0,1\}$ defined by $f(x,i) = x_i$. Furthermore, this Boolean function is a $(J,1)$-junta with $J = [n]$ and $|\overline{J}|= \log{n}$. Indeed, by fixing the coordinates of $i\in[n]$, $f_{[n]|i}(x) = x_i$ depends only on one input coordinate. \Cref{thr:fin_onehot} thus immediately applies to any $\QRAM$ by setting $r=1$ and $t = |\overline{J}| = \log{n}$. (Actually, there is no need to compute the one-hot encoding of register $\mathtt{J}$ since $r=1$. This means that we only require registers $\mathtt{E}_z$ for $z\in\{0,1\}^t$. The number of ancillae is thus halved). For completeness we depict the circuit in \Cref{fig:qram_construction}.

\begin{figure}[t]
    \centering
    \includegraphics[trim={1.4cm 16.65cm 3cm 0.8cm},clip,width=\textwidth]{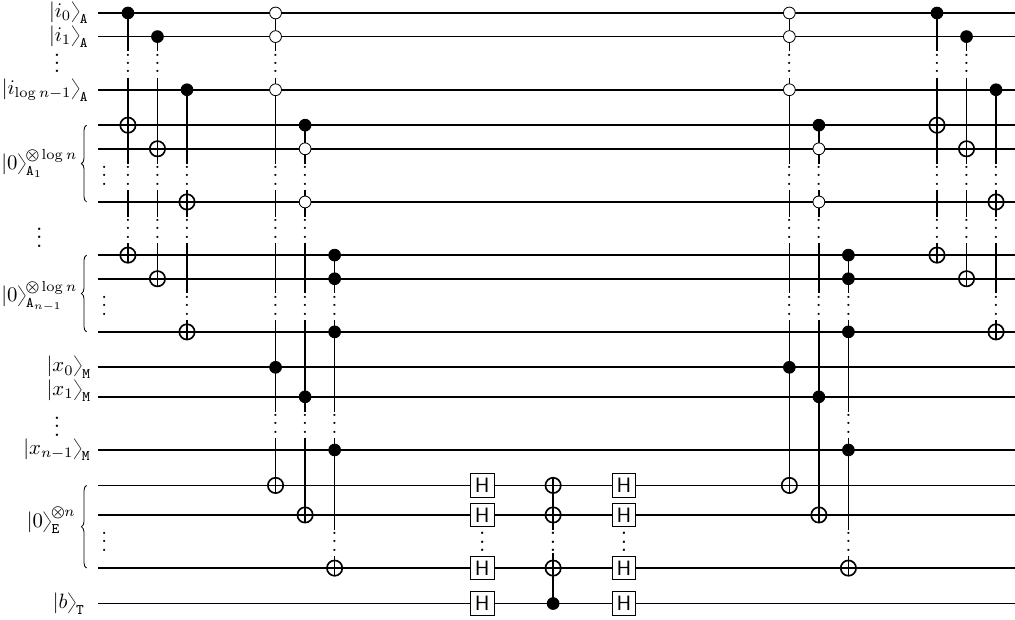}
    \caption{The circuit for the $\mathsf{QRAM}$ $|i\rangle_{\mathtt{A}}|b\rangle_{\mathtt{T}}|x\rangle_{\mathtt{M}} \mapsto |i\rangle_{\mathtt{A}}|b\oplus x_i\rangle_{\mathtt{T}}|x\rangle_{\mathtt{M}}$.}
    \label{fig:qram_construction}
\end{figure}

\begin{theorem}[One-hot-encoding implementation of $\QRAM$]\label{thr:qram in qac0f}
    For every $n \in \mathbb{N}$ a power of $2$, a $\mathsf{QRAM}$ of memory size $n$ can be implemented in $O(1)$-depth using
    \begin{itemize}
        \item either $2n\log{n}\log\log{n} + O(n\log{n})$ ancillae and $6n\log{n} + O(n\log\log{n})$ Fan-Out gates with arity $\leq n+1$, 
        \item or $3n\log{n} + O(n\log\log{n})$ ancillae and $6$ $\mathsf{GT}$ gates with arity $\leq n\log{n} + O(n\log\log{n})$.
    \end{itemize}
\end{theorem}

Even though $\mathsf{QRAG}$ is not an $f$-$\mathsf{FIN}$ or even an $f$-$\mathsf{UCG}$, it is possible to use the one-hot encoding framework from the previous circuit constructions to implement $\mathsf{QRAG}$ in constant depth. We mention that a similar $O(1)$-depth and $O(n\log{n}\log\log{n})$-size circuit for $\QRAG$ based on one-hot encoding and using Fan-Out gates had previously appeared in~\cite[Lemma~4.3]{rosenthal2021query}. We note that author missed the $\log\log{n}$-factor by ignoring the $\log{n}$-size overhead in \Cref{thr:or_constantdepth}.

\begin{theorem}[One-hot-encoding implementation of $\QRAG$]
    \label{thr:qrag in qac0f}
     For every $n \in \mathbb{N}$ a power of $2$, a $\mathsf{QRAG}$ of memory size $n$ can be implemented in $O(1)$-depth using
     \begin{itemize}
         \item either $2n\log{n}\log\log{n} + O(n\log{n})$ ancillae and $6n\log{n} + O(n\log\log{n})$ Fan-Out gates with arity $\leq n+1$,
         \item or $3n\log{n} + O(n\log\log{n})$ ancillae and $9$ $\mathsf{GT}$ gates with arity $\leq n\log{n} + O(n\log\log{n})$.
     \end{itemize}
\end{theorem}
\begin{proof}
    Given the state $|i\rangle_{\mathtt{A}}|b\rangle_{\mathtt{T}}|x_0,\dots,x_{n-1}\rangle_{\mathtt{M}}$, we shall compute the one-hot encoding $e(i)\in\{0,1\}^n$ of the address $i\in\{0,1\}^{\log{n}}$ given by $e(i)_j = \bigwedge_{k\in[\log{n}]} (i\oplus \overline{j})_k$, where $j\in\{0,1\}^{\log{n}}$. Since $e(i)_j = 1$ if and only if $j=i$, the one-hot encoding is then used to swap the correct entry $x_i$ from the memory $\mathtt{M}$ onto an $n$-qubit ancillary register $\mathtt{B}$. The swapped entry in register $\mathtt{B}$ is then mapped onto the target register $\mathtt{T}$ by using a $\mathsf{PARITY}$ gate. At this point, both registers $\mathtt{M}$ and $\mathtt{T}$ are in the desired state. The final step is uncomputing register $\mathtt{B}$ with an additional ancillary register $\mathtt{C}$.
    Consider the following circuit (see \Cref{fig:qrag_construction}):
    \begin{figure}[t]
        \centering
        \includegraphics[trim={1.4cm 13cm 3cm 0.8cm},clip,width=\textwidth]{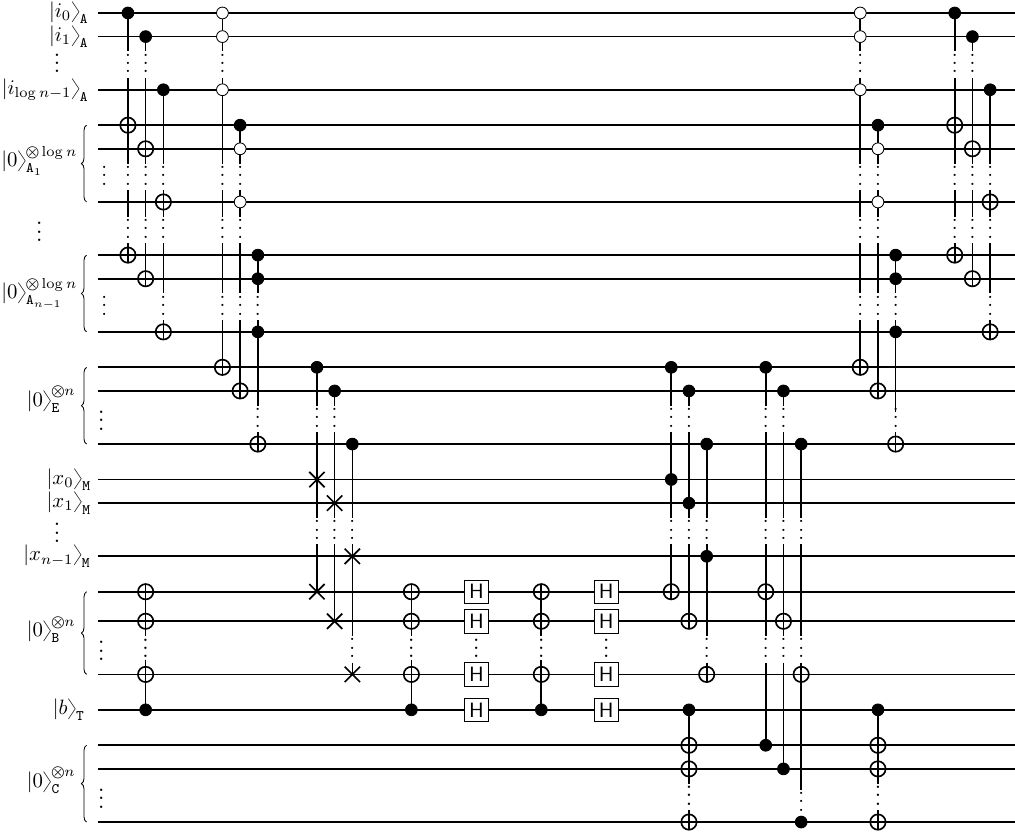}
        \caption{The circuit for the $\mathsf{QRAG}$ $|i\rangle_{\mathtt{A}}|b\rangle_{\mathtt{T}}|x\rangle_{\mathtt{M}} \mapsto |i\rangle_{\mathtt{A}}|x_i\rangle_{\mathtt{T}}|x_0,\dots,x_{i-1},b,x_{i+1},\dots,x_{n-1}\rangle_{\mathtt{M}}$. 
        The symbol $\bigtimes$---$\bigtimes$ means a $\mathsf{SWAP}$ gate.}
        \label{fig:qrag_construction}
    \end{figure}
    \begin{enumerate}
         \item Attach an $((n-1)\log{n})$-qubit ancillary register $\bigotimes_{j=1}^{n-1}|0\rangle^{\otimes \log{n}}_{\mathtt{A}_j}$ and copy $n - 1$ times the register $|i\rangle_{\mathtt{A}}$ using either $\log{n}$ Fan-Out gates with arity $n$ or $1$ $\mathsf{GT}$ gate with arity $n\log{n}$.

         \item Attach an $n$-qubit ancillary register $|0\rangle^{\otimes n}_{\mathtt{B}} = \bigotimes_{j\in[n]}|0\rangle_{\mathtt{B}_j}$ and apply an $(n+1)$-arity Fan-Out gate $\mathsf{FO}^{(n+1)}_{\mathtt{T}\to\mathtt{B}}$ from register $\mathtt{T}$ to register $\mathtt{B}$ to copy $n$ times the register $|b\rangle_{\mathtt{T}}$.
        
        \item For each $j\in[n]$, apply the gate $\bigotimes_{k\in[\log{n}]} \mathsf{X}^{\overline{j}_k}$ to $|i\rangle_{\mathtt{A}_j}$ (define $\mathtt{A}_0 := \mathtt{A}$). This leads to 
        \begin{align*}
            \left(\bigotimes_{j\in[n]}\ket{i}_{\mathtt{A}_j}|b\rangle_{\mathtt{B}_j}|x_j\rangle_{\mathtt{M}_j}\right)\ket{b}_{\mathtt{T}} \mapsto \left(\bigotimes_{j\in[n]} \ket{i\oplus \overline{j}}_{\mathtt{A}_j}|b\rangle_{\mathtt{B}_j}|x_j\rangle_{\mathtt{M}_j}\right)\ket{b}_{\mathtt{T}}.
        \end{align*}

        \item Attach a new $n$-qubit ancillary register $|0\rangle^{\otimes n}_{\mathtt{E}} = \bigotimes_{j\in[n]}|0\rangle_{\mathtt{E}_j}$ and apply an $\mathsf{AND}^{(\log{n})}_{\mathtt{A}_j\to\mathtt{E}_j}$ gate from register $\mathtt{A}_j$ onto register $\mathtt{E}_j$ for all $j\in[n]$ to obtain
        \begin{align*}
            \left(\bigotimes_{j\in[n]} \ket{i\oplus \overline{j}}_{\mathtt{A}_j}|b\rangle_{\mathtt{B}_j}|x_j\rangle_{\mathtt{M}_j}|0\rangle_{\mathtt{E}_j}\right)\ket{b}_{\mathtt{T}} \mapsto \left(\bigotimes_{j\in[n]}\ket{i\oplus \overline{j}}_{\mathtt{A}_j}|b\rangle_{\mathtt{B}_j}|x_j\rangle_{\mathtt{M}_j}|e(i)_j\rangle_{\mathtt{E}_j}\right)\ket{b}_{\mathtt{T}},
        \end{align*}
        where $e(i)\in \{0,1\}^{n}$ is the one-hot encoding of $i$.

        \item Apply $n$ $\mathsf{C}_{\mathtt{E}_j}$-$\mathsf{SWAP}_{\mathtt{B}_j\leftrightarrow\mathtt{M}_j}$ gates in parallel for $j\in[n]$, i.e., swap registers $\mathtt{B}_j$ and $\mathtt{M}_j$ controlled on $\mathtt{E}_j$. Since $e(i)_j = 1$ if and only if $j=i$, we obtain (ignore the register $\bigotimes_{j\in[n]} \ket{i\oplus \overline{j}}_{\mathtt{A}_j}$ for clarity)
        \begin{align*}
            \ket{e(i)}_{\mathtt{E}}\ket{b,\dots,b, x_i,b,\dots,b}_{\mathtt{B}}\ket{x_0,\dots,x_{i-1},b,x_{i+1},\dots,x_{n-1}}_{\mathtt{M}}\ket{b}_{\mathtt{T}}.
        \end{align*}
        \item Apply an $(n+1)$-arity Fan-Out gate $\mathsf{FO}^{(n+1)}_{\mathtt{T}\to\mathtt{B}}$ from register $\mathtt{T}$ onto register $\mathtt{B}$ to get
        \begin{align*}
            \ket{e(i)}_{\mathtt{E}}\ket{0,\dots,0, b\oplus x_i,0,\dots,0}_{\mathtt{B}}\ket{x_0,\dots,x_{i-1},b,x_{i+1},\dots,x_{n-1}}_{\mathtt{M}}\ket{b}_{\mathtt{T}}.
        \end{align*}
        \item Apply a $\mathsf{PARITY}_{\mathtt{B}\to\mathtt{T}}$ gate from register $\mathtt{B}$ onto register $\mathtt{T}$ to obtain
        \begin{align*}
            \ket{e(i)}_{\mathtt{E}}\ket{0,\dots,0, b\oplus x_i,0,\dots,0}_{\mathtt{B}}\ket{x_0,\dots,x_{i-1},b,x_{i+1},\dots,x_{n-1}}_{\mathtt{M}}\ket{x_i}_{\mathtt{T}}.
        \end{align*}
        \item Apply $n$ $\mathsf{C}_{\{\mathtt{E}_j,\mathtt{M}_j\}}$-$\mathsf{X}_{\to\mathtt{B}_j}$ gates in parallel for $j\in[n]$, i.e., apply an $\mathsf{X}$ gate onto register $\mathtt{B}_j$ controlled on registers $\mathtt{E}_j$ and $\mathtt{M}_j$. This yields
        \begin{align*}
            \ket{e(i)}_{\mathtt{E}}\ket{0,\dots,0,  x_i,0,\dots,0}_{\mathtt{B}}\ket{x_0,\dots,x_{i-1},b,x_{i+1},\dots,x_{n-1}}_{\mathtt{M}}\ket{x_i}_{\mathtt{T}}.
        \end{align*}
        \item Attach a new $n$-qubit ancillary register $|0\rangle^{\otimes n}_{\mathtt{C}}$ and apply an $(n+1)$-arity Fan-Out gate $\mathsf{FO}^{(n+1)}_{\mathtt{T}\to\mathtt{C}}$ from register $\mathtt{T}$ to register $\mathtt{C}$ to get (ignore register $\mathtt{M}$ for clarity)
        \begin{align*}
        \ket{e(i)}_{\mathtt{E}}\ket{0,\dots,0,x_i,0,\dots,0}_{\mathtt{B}}\ket{0}_{\mathtt{C}}^{\otimes n}\ket{x_i}_{\mathtt{T}} \mapsto \ket{e(i)}_{\mathtt{E}}\ket{0,\dots,0,x_i,0,\dots,0}_{\mathtt{B}}\ket{x_i}_{\mathtt{C}}^{\otimes n}\ket{x_i}_{\mathtt{T}}.
        \end{align*}
        \item Apply $n$ $\mathsf{C}_{\{\mathtt{E}_j,\mathtt{C}_j\}}$-$\mathsf{X}_{\to\mathtt{B}_j}$ gates in parallel for $j\in[n]$ to uncompute register $\mathtt{B}$ and get
        \begin{align*}
            \ket{e(i)}_{\mathtt{E}}\ket{0,\dots,0,x_i,0,\dots,0}_{\mathtt{B}}\ket{x_i}_{\mathtt{C}}^{\otimes n}\ket{x_i}_{\mathtt{T}} \mapsto \ket{e(i)}_{\mathtt{E}}\ket{0}_{\mathtt{B}}^{\otimes n}\ket{x_i}_{\mathtt{C}}^{\otimes n}\ket{x_i}_{\mathtt{T}}.
        \end{align*}
        \item Uncompute Steps $9$, $4$, $3$, $2$, and $1$. This leads to $|i\rangle_{\mathtt{A}}|x_i\rangle_{\mathtt{T}}|x_0,\dots,x_{i-1},b,x_{i+1},\dots,x_{n-1}\rangle_{\mathtt{M}}$.
    \end{enumerate}
    
\noindent    We now analyse the resources for each step:
    \begin{itemize}
        \item Steps $1$ and $2$: the registers $\mathtt{A}_1,\dots,\mathtt{A}_{n-1}$ and $\mathtt{B}_0,\dots,\mathtt{B}_{n-1}$ use $(n-1)\log{n} + n$ ancillae. Copying register $\mathtt{A}$ onto $\mathtt{A}_1,\dots,\mathtt{A}_{n-1}$ and register $\mathtt{T}$ onto $\mathtt{B}_0,\dots,\mathtt{B}_{n-1}$ requires either $\log{n} + 1$ Fan-Out gates with arity at most $n+1$ or $1$ $\mathsf{GT}$ gate with arity $n\log{n} + n+1$;
        
        \item Step $3$: the $n$ $\mathsf{AND}^{(\log{n})}_{\mathtt{A}_j\to\mathtt{E}_j}$ gates use either $2n\log{n}\log\log{n} + O(n\log{n})$ ancillae and $3n\log{n} + O(n\log\log{n})$ Fan-Out gates with arity at most $2\log{n}$ (\Cref{thr:or_constantdepth}) or $2n\log{n} + O(n\log\log{n})$ ancillae and $2$ $\mathsf{GT}$ gates with arity $n\log{n} + O(n\log\log{n})$ (\Cref{thr:or_constantdepth_gt}), where for both cases we pushed the uncomputation part of the $\mathsf{AND}^{(\log{n})}_{\mathtt{A}_j\to\mathtt{E}_j}$ gates to Step~$11$;
        \item Step $6$: copying register $\mathtt{T}$ onto $\mathtt{B}$ uses $1$ Fan-Out or $1$ $\mathsf{GT}$ gate with arity $n+1$;
        \item Step $7$: the $\mathsf{PARITY}_{\mathtt{B}\to\mathtt{T}}$ gate uses $1$ Fan-Out or $1$ $\mathsf{GT}$ gate with arity $n+1$;
        \item Step $9$: copying register $\mathtt{T}$ onto $\mathtt{C}$ uses $1$ Fan-Out or $1$ $\mathsf{GT}$ gate with arity $n+1$;
        \item Step $11$: uncomputing the previous steps uses $3n\log{n} + O(n\log\log{n})$ Fan-Out gates or $3$ $\mathsf{GT}$ gates, since the $\mathsf{GT}$ gate from uncomputing Step $9$ can be absorbed by the ones from uncomputing Steps $1$, $2$, or $3$.
    \end{itemize}

   In total, we require either $2n\log{n}\log\log{n} + O(n\log{n})$ ancillae and $6n\log{n} + O(n\log\log{n})$ Fan-Out gates with arity at most $n+1$ or $3n\log{n} + O(n\log\log{n})$ ancillae and $9$ $\mathsf{GT}$ gates with arity at most $n\log{n} + O(n\log\log{n})$.
\end{proof}

We note that the number of ancillae for $\mathsf{QRAM}$ and $\mathsf{QRAG}$ can be asymptotically reduced at the expense of at most a constant factor increase in depth. This is achieved by reducing the problem into small blocks, each of which is solved using a small $\mathsf{QRAM}$/$\mathsf{QRAG}$ circuit. The outcome of all the small circuits is then solved by another small $\mathsf{QRAM}$/$\mathsf{QRAG}$ circuit. This can be done recursively in a tree-wise fashion, where the output of one level of $\mathsf{QRAM}$/$\mathsf{QRAG}$ circuits is broken into small blocks and inputted into a new level of $\mathsf{QRAM}$/$\mathsf{QRAG}$ circuits. This is formalised in the next result. Note that the same idea was used for the $\mathsf{OR}$ function, see~\cite[Theorem~6.3]{hoyer2005quantum} and~\cite[Section~3.2]{takahashi2016collapse}.
\begin{theorem}\label{thr:qram_recursive_procedure}
    For every $n,d \in \mathbb{N}$, a $\mathsf{QRAM}$ of memory size $n$ can be performed in $O(d)$-depth~using
    \begin{itemize}
        \item either $O(n\log^{(d)}{n}\log^{(d+1)}{n})$ ancillae and $O(n\log^{(d)}{n})$ Fan-Out gates,
        \item or $O(n\log^{(d)}{n})$ ancillae and $16d-10$ $\mathsf{GT}$ gates.
    \end{itemize}
    Moreover, a $\mathsf{QRAG}$ of memory size $n$ can be performed in $O(d)$-depth using
     \begin{itemize}
         \item either $O(n\log^{(d)}{n}\log^{(d+1)}{n})$ ancillae and $O(n\log^{(d)}{n})$ Fan-Out gates,
         \item or $O(n\log^{(d)}{n})$ ancillae and $21d-12$ $\mathsf{GT}$ gates.
     \end{itemize}
\end{theorem}
\begin{proof}
    We first focus on $\mathsf{QRAM}$. The proof is by induction on $d$. For $d=1$, the result follows from \Cref{thr:qram in qac0f}. Assume that the result is true for $d-1$. Divide the input $x\in\{0,1\}^n$ into $m := n/\log^{(d-1)}{n}$ blocks of $b := \log^{(d-1)}{n}$ qubits each. Given $i\in[n]$, compute $r \equiv i~(\text{mod}~ b)$ into an ancillary register $\mathtt{B}_0$ using a $O(1)$-depth $\poly\log{n}$-size quantum circuit~\cite{hoyer2005quantum,siu1993depth}. We then copy $|r\rangle_{\mathtt{B}_0}$ a number of $m-1$ times to obtain $\bigotimes_{j=0}^{m-1}|r\rangle_{\mathtt{B}_j}$. For each input $|x_{jb},x_{jb+1},\dots,x_{j(b+1)-1}\rangle_{\mathtt{M}}|r\rangle_{\mathtt{B}_j}$, $j\in[m]$, we apply the $\mathsf{QRAM}$ circuit from \Cref{thr:qram in qac0f} with target qubit $|0\rangle_{\mathtt{T}_j}$, which yields $|x_{jb+r}\rangle_{\mathtt{T}_j}$. This $d$-th $\mathsf{QRAM}$-level uses either $O(mb\log{b}\log\log{b}) = O(n\log^{(d)}{n}\log^{(d+1)}{n})$ ancillae and $O(mb\log{b}) = O(n\log^{(d)}{n})$ Fan-Outs, or $O(n\log^{(d)}{n})$ ancillae and $6$ $\mathsf{GT}$ gates (the $\mathsf{GT}$ gates from different blocks can be done in parallel). We are left with $m = n/\log^{(d-1)}{n}$ output qubits $\bigotimes_{j=0}^{m-1}|x_{jb+r}\rangle_{\mathtt{T}_j}$. Compute now $q := \lfloor i /b \rfloor$ into a separate quantum register using a $O(1)$-depth $\poly\log{n}$-size quantum circuit~\cite{hoyer2005quantum,siu1993depth}. Using the induction hypothesis, we can input $|q\rangle\bigotimes_{j=0}^{m-1}|x_{jb+r}\rangle_{\mathtt{T}_j}$ into a $O(d)$-depth $\mathsf{QRAM}$ circuit (the remaining $d-1$ $\mathsf{QRAM}$-levels) that uses either $O(m\log^{(d-1)}{m}\log^{(d)}{m}) = O(n\log^{(d)}{n})$ ancillae and $O(m\log^{(d-1)}{m}) = O(n)$ Fan-Outs, or $O(n)$ ancillae and $16(d-1) - 10$ $\mathsf{GT}$ gates. The output qubit is $|b\oplus x_{bq+r}\rangle_{\mathtt{T}} = |b\oplus x_{i}\rangle_{\mathtt{T}}$ as required. We then uncompute all the intermediary steps. 
    
    Let us compute the amount of resources. First note that computing and uncomputing the $d$-th $\mathsf{QRAM}$-level uses $6$ $\mathsf{GT}$ gates since we can wait and perform the second half of the $\mathsf{QRAM}$ circuits (see \Cref{fig:qram_construction}) until after $|b\oplus x_{i}\rangle_{\mathtt{T}}$ is outputted. Moreover, copying and uncopying $|r\rangle_{\mathtt{B}_0}$ requires $2$ $\mathsf{GT}$ gates. Finally, we must take into consideration the resources to compute $q$ and~$r$. For the Fan-Out-based construction, it only requires $\poly\log{n}$ Fan-Outs. The $\mathsf{GT}$-gate-based construction, on the other hand, requires more care. It is  well known that $q$ can be computed using a depth-$4$ polynomial-size threshold circuit, while $r$ can be computed with a depth-$2$ polynomial-size threshold circuit~\cite{siu1993depth}. In order to compute $\mathsf{THRESHOLD}$ functions, we employ the Boolean construction from \Cref{thr:fin_boolean_construction1}. Due to that, we shall postpone the proof till \Cref{sec:boolean_qram_construction} (see proof of \Cref{thr:qram_recursive_procedure_boolean}) and just claim for now that computing (plus uncomputing) $q$ uses $O(n)$ ancillae and $16$ $\mathsf{GT}$ gates, while computing (plus uncomputing) $r$ uses $O(n)$ ancillae and $8$ $\mathsf{GT}$ gates. Computing $q$ ($16$ $\mathsf{GT}$s) can be done in parallel to computing plus copying $r$ and performing the $d$-th $\mathsf{QRAM}$ circuit ($8+2+6 = 16$ $\mathsf{GT}$s), so the $d$-th level-$\mathsf{QRAM}$ uses $16$ $\mathsf{GT}$ gates. In total, we use $16d - 10$ $\mathsf{GT}$ gates. 

    The analysis is the same for $\mathsf{QRAG}$, the difference being the number of $\mathsf{GT}$ gates. Computing and uncomputing the $d$-th $\mathsf{QRAG}$ circuit from \Cref{thr:qrag in qac0f} uses $11$ $\mathsf{GT}$s instead of $9$, since the $\mathsf{GT}$ gates next to the Hadamard layers in \Cref{fig:qrag_construction} must be uncomputed. Computing $q$ ($16$ $\mathsf{GT}$s) can be done in parallel to computing plus copying $r$ and performing the $d$-th $\mathsf{QRAG}$ circuit ($8+2+11 = 21$ $\mathsf{GT}$s), so the $d$-th level-$\mathsf{QRAG}$ uses $21$ $\mathsf{GT}$ gates. In total, we use $21d-12$ $\mathsf{GT}$ gates. %the only difference is that, since we are using $\poly\log{n}$ Fan-Out gates to perform the small computation of the quotient $q = \lfloor i/m\rfloor$ and the remainder $r\equiv i~(\text{mod}~b)$, $2$ $\mathsf{GT}$ gates in the $\mathsf{QRAG}$ circuit can be replaced with $2$ Fan-Out gates. This reduces the $\mathsf{GT}$-gate count from $9$ to $7$ per $\mathsf{QRAG}$ call. However, the copying of $|r\rangle_{\mathtt{B}_0}$ requires another $\mathsf{GT}$ gate. Thus we require $8(d-1) + 7 = 8d - 1$ $\mathsf{GT}$ gates in total.
\end{proof}

\begin{remark}
    If $d = \log^\ast{n}$, then $\mathsf{QRAM}$ requires only $O(n)$ ancillae and either $O(n)$ Fan-Outs or $16\log^\ast{n} - 10$ $\mathsf{GT}$ gates. Similarly for $\mathsf{QRAG}$. However, the depth becomes $O(\log^\ast{n})$, which is no longer constant. Nonetheless, the log-star of the estimated number of atoms in the universe is $5$, thus one can consider $\log^\ast{n}$ to be a constant in practice.
\end{remark}
% \begin{remark}
%     It is possible to reduce the number of $\mathsf{GT}$ gates by trading in extra Fan-Outs. The $q$ and $r$ computation and copying uses $\poly\log{n}$ Fan-Outs, and $2$ $\mathsf{GT}$ gates in the basic $\mathsf{QRAG}$ circuit (cf.\ {\rm \Cref{fig:qrag_construction}}) can be traded by $2$ Fan-Outs. Thus $\mathsf{QRAM}$ and $\mathsf{QRAG}$ can use $6d$ and $7d$ $\mathsf{GT}$ gates, respectively, plus $\poly\log{n}$ Fan-Outs.
% \end{remark}

\section{Constant-depth circuits based on Boolean analysis}\label{sec:fourier}

In this section, we explore the Boolean analysis connection between constant-depth gates and Fan-Outs made by Takahashi and Tani~\cite{takahashi2016collapse} and propose several constructions for $f$-$\mathsf{UCG}$s. Given $f:\{0,1\}^n\to\mathcal{U}(\mathbb{C}^{2\times 2})$, consider its $\mathsf{Z}$-decomposition $\alpha,\beta,\gamma,\delta:\{0,1\}^n\to[-1,1]$. Recall that $\operatorname{supp}(f) := \operatorname{supp}(\alpha)\cup \operatorname{supp}(\beta)\cup \operatorname{supp}(\gamma)\cup \operatorname{supp}(\delta)$ and $\operatorname{deg}(f) := \max\{\operatorname{deg}(\alpha),\operatorname{deg}(\beta),\operatorname{deg}(\gamma),\operatorname{deg}(\delta)\}$. Similar definitions apply to $\operatorname{supp}^{>k}(f)$, $\operatorname{supp}^{\leq k}(f)$, $\operatorname{supp}^{=k}(f)$, $\operatorname{supp}_{\{0,1\}}(f)$, and $\operatorname{supp}^{> k}_{\{0,1\}}(f)$.

\subsection{Constant-depth circuits for $f\text{-}\mathsf{UCG}$s}

\begin{theorem}[Real implementation of $f$-$\mathsf{UCG}$]\label{thr:ucg_boolean_construction1}
    Given $f:\{0,1\}^n\to\mathcal{U}(\mathbb{C}^{2\times 2})$, there is a $O(1)$-depth circuit for $f\text{-}\mathsf{UCG}$ that uses 
    \begin{itemize}
        \item either $\sum_{S\in\operatorname{supp}(f)}|S| + 2|\operatorname{supp}^{>1}(f)|$ ancillae and $2|\bigcup_{S\in\operatorname{supp}^{>1}(f)}S|+2|\operatorname{supp}^{>1}(f)|+6$ Fan-Out gates with arity $\leq 1+\max\{|\operatorname{supp}^{>0}(f)|,\operatorname{deg}(f)\}$, 
        \item or $|\operatorname{supp}^{>1}(f)|$ ancillae and $5$ $\mathsf{GT}$ gates with arity $\leq 2\sum_{S\in\operatorname{supp}(f)} |S|$.
    \end{itemize}
\end{theorem}
\begin{proof}
    Consider the initial state $|x\rangle_{\mathtt{I}}|b\rangle_{\mathtt{T}}$ for $x\in\{0,1\}^n$ and $b\in\{0,1\}$. We wish to implement $|x\rangle_{\mathtt{I}}|b\rangle_{\mathtt{T}} \mapsto |x\rangle_{\mathtt{I}}f(x)|b\rangle_{\mathtt{T}}$. Let $\alpha,\beta,\gamma,\delta:\{0,1\}^n\to[-1,1]$ be the $\mathsf{Z}$-decomposition of $f$. Let $\alpha(x) = \sum_{S\subseteq[n]} \widehat{\alpha}(S)\chi_S(x)$ be the Fourier expansion of $\alpha$, and similarly for $\beta$, $\gamma$, and $\delta$. For ease of notation, write $m:= |\operatorname{supp}^{>0}(f)|$. Let also $m_i := |\{S\in\operatorname{supp}^{>1}(f):i\in S\}|$ be the number of sets of size greater than $1$ that contain the coordinate $i\in[n]$. 
    
    Consider first the Fan-Out-based circuit (see \Cref{fig:boolean_construction_fanout}):
    \begin{figure}[t]
        \centering
        \includegraphics[trim={1.45cm 16cm 0.6cm 0.8cm},clip,width=\textwidth]{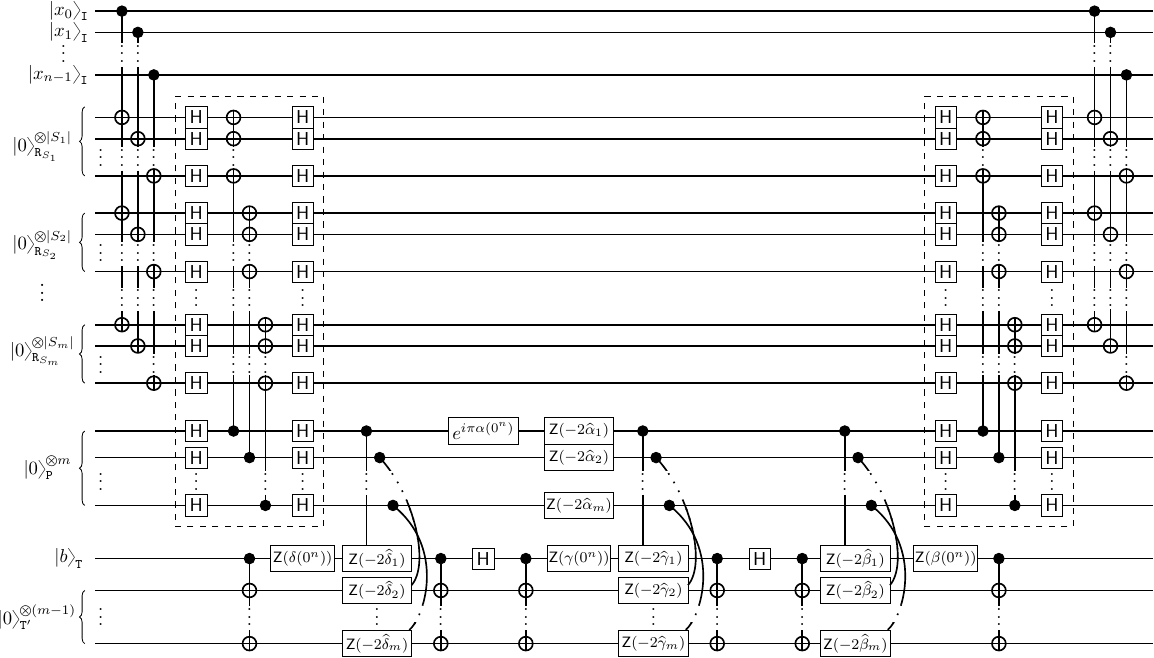}
        \caption{The circuit for an $f$-$\mathsf{UCG}^{(n)}$ using Fan-Out gates. Here $m:=|\operatorname{supp}^{>0}(f)|$. We highlight the $\mathsf{PARITY}$ operations inside dashed boxes. For simplicity, we write $\widehat{\alpha}(S_j)$ as $\widehat{\alpha}_j$ (and similarly for $\beta$, $\gamma$, $\delta$). Moreover, we depict $S_1,\dots,S_m\in\operatorname{supp}^{>0}(f)$, but in reality there is no need to compute the parities of sets with size $1$ (hence why the register $\mathtt{P}$ is shown with size $m$). Moreover, we include targets onto all registers $\{\mathtt{R}_S\}_S$ in the Fan-Out gates copying $x_0,\dots,x_{n-1}$, but in reality $x_i$ is copied only onto the registers such that $S\ni i$.}
        \label{fig:boolean_construction_fanout}
    \end{figure}
    
    \begin{enumerate}
    \setcounter{enumi}{1}
        \item [1a.] Attach an ancillary register $\bigotimes_{S\in\operatorname{supp}^{>1}(f)}|0\rangle^{\otimes|S|}_{\mathtt{R}_S}$. For each $i\in[n]$ in parallel, copy $m_i$ number of times the qubit $|x_i\rangle_{\mathtt{I}}$ using a $(1+m_i)$-arity Fan-Out gate. This leads to
        \begin{align*}
            |x\rangle_{\mathtt{I}}|b\rangle_{\mathtt{T}} \mapsto |x\rangle_{\mathtt{I}}|b\rangle_{\mathtt{T}} \bigotimes_{S\in\operatorname{supp}^{>1}(f)}|x_S\rangle_{\mathtt{R}_S}.
        \end{align*}
        \item [1b.] Attach an ancillary register $|0\rangle^{\otimes |\operatorname{supp}^{>1}(f)|}_{\mathtt{P}} = \bigotimes_{S\in\operatorname{supp}^{>1}(f)}|0\rangle_{\mathtt{P}_S}$. For each $S\in\operatorname{supp}^{>1}(f)$ in parallel, apply a $\mathsf{PARITY}^{(|S|)}_{\mathtt{R}_S \to \mathtt{P}_S}$ gate using a $(1+|S|)$-arity Fan-Out gate. We obtain the state 
        \begin{align*}
            |x\rangle_{\mathtt{I}}|b\rangle_{\mathtt{T}} \bigotimes_{S\in\operatorname{supp}^{>1}(f)}|x_S\rangle_{\mathtt{R}_S} \mapsto |x\rangle_{\mathtt{I}}|b\rangle_{\mathtt{T}} \bigotimes_{S\in\operatorname{supp}^{>1}(f)}|x_S\rangle_{\mathtt{R}_S}\big|{\bigoplus}_{i\in S}x_i\big\rangle_{\mathtt{P}_S}.
        \end{align*}
        \item [2a.] Attach an ancillary register $|0\rangle^{\otimes (m-1)}_{\mathtt{T}'}$ and apply an $m$-arity Fan-Out gate $\mathsf{FO}^{(m)}_{\mathtt{T}\to\mathtt{T}'}$ from register $\mathtt{T}$ onto register $\mathtt{T}'$. Apply a $\mathsf{Z}(\delta(0^n))_{\to\mathtt{T}}$ gate onto register $\mathtt{T}$. Notice that $\delta(0^n) = \sum_{S\subseteq[n]}\widehat{\delta}(S)$. Then, for each $S\in\operatorname{supp}^{>0}(\delta)$ in parallel, apply a $\mathsf{Z}(-2\widehat{\delta}(S))$ gate controlled on register $\mathtt{P}_S$ onto the $S$-th qubit in register $\mathtt{T}'$ (if $|S|=1$, the gate is controlled on $|x_S\rangle_{\mathtt{I}}$). Finally, apply $\mathsf{FO}^{(m)}_{\mathtt{T}\to\mathtt{T}'}$ again. This chain of operations leads to (omit registers $\mathtt{R}_S$ and $\mathtt{P}_S$ for simplicity)
        \begin{align*}
            |x\rangle_{\mathtt{I}}|b\rangle_{\mathtt{T}} \mapsto 
            |x\rangle_{\mathtt{I}} |b\rangle^{\otimes m}_{\mathtt{T},\mathtt{T}'}
            &\mapsto |x\rangle_{\mathtt{I}} \mathsf{Z}\left(\sum_{S\subseteq[n]}\widehat{\delta}(S)\left(1 - 2\bigoplus_{i\in S} x_i\right) \right) |b\rangle^{\otimes m}_{\mathtt{T},\mathtt{T}'}\\
            &= |x\rangle_{\mathtt{I}} \mathsf{Z}\left(\sum_{S\subseteq[n]}\widehat{\delta}(S)\chi_S(x) \right) |b\rangle^{\otimes m}_{\mathtt{T},\mathtt{T}'}
            \mapsto |x\rangle_{\mathtt{I}} \mathsf{Z}(\delta(x)) |b\rangle_{\mathtt{T}}.
        \end{align*}
        \item [2b.] Apply a $\mathsf{H}_{\to\mathtt{T}}$ gate onto register $\mathtt{T}$ followed by an $m$-arity Fan-Out gate $\mathsf{FO}^{(m)}_{\mathtt{T}\to\mathtt{T}'}$ from register $\mathtt{T}$ onto register $\mathtt{T}'$. Apply a $\mathsf{Z}(\gamma(0^n))_{\to\mathtt{T}}$ gate onto register $\mathtt{T}$. Then, for each $S\in\operatorname{supp}^{>0}(\gamma)$ in parallel, apply a $\mathsf{Z}(-2\widehat{\gamma}(S))$ gate controlled on register $\mathtt{P}_S$ onto the $S$-th qubit in register $\mathtt{T}'$ (if $|S|=1$, the gate is controlled on $|x_S\rangle_{\mathtt{I}}$). Finally, apply $\mathsf{FO}^{(m)}_{\mathtt{T}\to\mathtt{T}'}$ again. For simplicity, write $\mathsf{H} \mathsf{Z}(\delta(x))|b\rangle_{\mathtt{T}} = r_{b,x}|0\rangle_{\mathtt{T}} + s_{b,x}|1\rangle_{\mathtt{T}}$. This chain of operations leads to
        \begin{align*}
            |x\rangle_{\mathtt{I}}\mathsf{H}\mathsf{Z}(\delta(x))|b\rangle_{\mathtt{T}} &\mapsto 
            |x\rangle_{\mathtt{I}}\big(r_{b,x}|0\rangle^{\otimes m}_{\mathtt{T},\mathtt{T}'} + s_{b,x}|1\rangle^{\otimes m}_{\mathtt{T},\mathtt{T}'}\big)\\
            &\mapsto |x\rangle_{\mathtt{I}} \mathsf{Z}\left(\sum_{S\subseteq[n]}\widehat{\gamma}(S)\left(1 - 2\bigoplus_{i\in S} x_i\right) \right) \big(r_{b,x}|0\rangle^{\otimes m}_{\mathtt{T},\mathtt{T}'} + s_{b,x}|1\rangle^{\otimes m}_{\mathtt{T},\mathtt{T}'}\big)\\
            &\mapsto |x\rangle_{\mathtt{I}} \mathsf{Z}\left(\sum_{S\subseteq[n]}\widehat{\gamma}(S)\chi_S(x) \right)\mathsf{H}\mathsf{Z}(\delta(x)) |b\rangle_{\mathtt{T}}\\
            &= |x\rangle_{\mathtt{I}} \mathsf{Z}(\gamma(x)) \mathsf{H}\mathsf{Z}(\delta(x))|b\rangle_{\mathtt{T}}.
        \end{align*}

        \item [2c.] Apply a $\mathsf{H}_{\to\mathtt{T}}$ gate onto register $\mathtt{T}$ followed by an $m$-arity Fan-Out gate $\mathsf{FO}^{(m)}_{\mathtt{T}\to\mathtt{T}'}$ from register $\mathtt{T}$ onto register $\mathtt{T}'$. Apply a $\mathsf{Z}(\beta(0^n))_{\to\mathtt{T}}$ gate onto register $\mathtt{T}$. Then, for each $S\in\operatorname{supp}^{>0}(\beta)$ in parallel, apply a $\mathsf{Z}(-2\widehat{\beta}(S))$ gate controlled on register $\mathtt{P}_S$ onto the $S$-th qubit in register $\mathtt{T}'$ (if $|S|=1$, the gate is controlled on $|x_S\rangle_{\mathtt{I}}$). Finally, apply $\mathsf{FO}^{(m)}_{\mathtt{T}\to\mathtt{T}'}$ again. Similarly to the previous step, this chain of operations leads to
        \begin{align*}
            |x\rangle_{\mathtt{I}} \mathsf{H}\mathsf{Z}(\gamma(x)) \mathsf{H}\mathsf{Z}(\delta(x))|b\rangle_{\mathtt{T}}
            \mapsto |x\rangle_{\mathtt{I}} \mathsf{Z}(\beta(x))\mathsf{H}\mathsf{Z}(\gamma(x)) \mathsf{H}\mathsf{Z}(\delta(x))|b\rangle_{\mathtt{T}}.
        \end{align*}

        \item [2d.] Apply an overall $e^{i\pi \alpha(0^n)}$ phase. For each $S\in\operatorname{supp}^{>0}(\alpha)$ in parallel, apply a $\mathsf{Z}(-2\widehat{\alpha}(S))$ gate onto register $\mathtt{P}_S$ (if $|S|=1$, apply the gate onto $|x_S\rangle_{\mathtt{I}}$). This leads to
        \begin{align*}
            |x\rangle_{\mathtt{I}} e^{i\pi \alpha(x)} \mathsf{Z}(\beta(x))\mathsf{H}\mathsf{Z}(\gamma(x)) \mathsf{H}\mathsf{Z}(\delta(x))|b\rangle_{\mathtt{T}} = |x\rangle_{\mathtt{I}}f(x)|b\rangle_{\mathtt{T}}.
        \end{align*}

        \item [3.] Uncompute Step~$1$.
    \end{enumerate}
    We now consider the $\mathsf{GT}$-gate-based circuit (see \Cref{fig:boolean_construction_gt}), which is basically the same as the Fan-Out-based circuit, the main difference being that it is no longer necessary to copy the register $|x\rangle_{\mathtt{I}}$ several times into $\bigotimes_{S\in\operatorname{supp}^{>1}(f)}\bigotimes_{i\in S}|x_i\rangle_{\mathtt{R}_S}$ as an intermediate step in order to compute the terms $\bigoplus_{i\in S} x_i$ for all $S\in\operatorname{supp}^{>1}(f)$. A single $\mathsf{GT}$ gate can compute all the parity terms in parallel according to \Cref{claim:fanoutasGT}. In the following, write register $\mathtt{I}$ as $|x\rangle_{\mathtt{I}} = \bigotimes_{i\in[n]}|x_i\rangle_{\mathtt{I}_i}$.
    \begin{figure}[t]
        \centering
        \includegraphics[trim={1.45cm 22.8cm 1.0cm 0.75cm},clip,width=\textwidth]{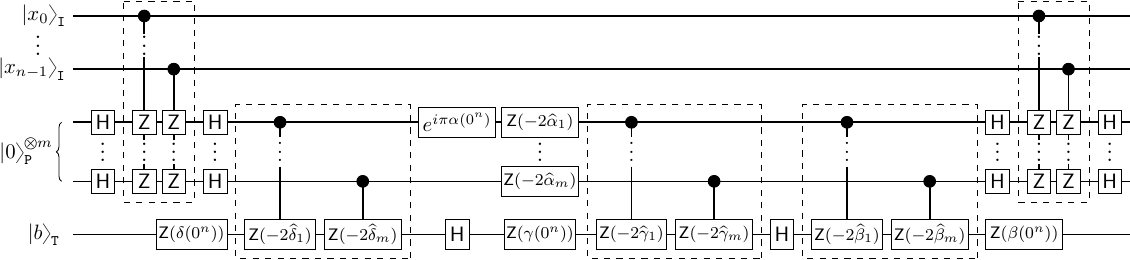}
        \caption{The circuit for an $f$-$\mathsf{UCG}^{(n)}$ using $\mathsf{GT}$ gates. Here $m:=|\operatorname{supp}^{>0}(f)|$. We highlight the $\mathsf{GT}$ gates inside dashed boxes. For simplicity, we write $\widehat{\alpha}(S_j)$ as $\widehat{\alpha}_j$ (and similarly for $\beta$, $\gamma$,~$\delta$). Moreover, we depict $S_1,\dots,S_m\in\operatorname{supp}^{>0}(f)$, but in reality there is no need to compute the parities of sets with size $1$ (hence why the register $\mathtt{P}$ is shown with size $m$). Moreover, we apply $\mathsf{Z}$ gates onto all registers $\{\mathtt{P}_S\}_S$ controlled on $x_0,\dots,x_{n-1}$ being in the $|1\rangle$ state, but in reality the $\mathsf{Z}$ gates controlled on $x_i$ are only applied onto the registers indexed by sets $S$ such that $S\ni i$.}
        \label{fig:boolean_construction_gt}
    \end{figure}
    
    \begin{enumerate}
        \item Apply Hadamard gates to an ancillary register $|0\rangle_{\mathtt{P}}^{\otimes|\operatorname{supp}^{>1}(f)|} = \bigotimes_{S\in\operatorname{supp}^{>1}(f)}|0\rangle_{\mathtt{P}_S}$. Then, using $1$ $\mathsf{GT}$ gate with arity $|\bigcup_{S\in\operatorname{supp}^{>1}(f)}S| + \sum_{S\in\operatorname{supp}^{>1}(f)}|S|$, apply, for each $i\in[n]$, a $\mathsf{Z}$ gate controlled on $|x_i\rangle_{\mathtt{I}_i}$ to the registers $\mathtt{P}_S$ indexed by the sets $S\in\operatorname{supp}^{>1}(f)$ that contain~$i$. Finally, apply another layer of Hadamard gates to the ancillary register $\mathtt{P}$. We obtain
        \begin{align*}
            |x\rangle_{\mathtt{I}}|b\rangle_{\mathtt{T}} \mapsto |x\rangle_{\mathtt{I}}|b\rangle_{\mathtt{T}}\bigotimes_{S\in\operatorname{supp}^{>1}(f)}\big|{\bigoplus}_{i\in S}x_i\big\rangle_{\mathtt{P}_S}.
        \end{align*}

        \item [2.] Apply the gate (write $\mathtt{P}_S := \mathtt{I}_S$ if $|S|=1$)
        \begin{align*}
            &\scalemath{0.97}{\left(e^{i\pi\alpha(0^n)}\prod_{S\in\operatorname{supp}^{>0}(\alpha)}\mathsf{Z}(\widehat{\alpha}(-2S))_{\to\mathtt{P}_{S}}\right)\left(\mathsf{Z}(\beta(0^n))_{\to\mathtt{T}}\prod_{S\in\operatorname{supp}^{>0}(\beta)}\mathsf{C}_{\mathtt{P}_S}\text{-}\mathsf{Z}(\widehat{\beta}(-2S))_{\to\mathtt{T}}\right)\mathsf{H}_{\to\mathtt{T}}} \\
            &\scalemath{0.97}{\cdot\left(\mathsf{Z}(\gamma(0^n))_{\to\mathtt{T}}\prod_{S\in\operatorname{supp}^{>0}(\gamma)}\mathsf{C}_{\mathtt{P}_S}\text{-}\mathsf{Z}(\widehat{\gamma}(-2S))_{\to\mathtt{T}}\right)\mathsf{H}_{\to\mathtt{T}}\left(\mathsf{Z}(\delta(0^n))_{\to\mathtt{T}}\prod_{S\in\operatorname{supp}^{>0}(\delta)}\mathsf{C}_{\mathtt{P}_{S}}\text{-}\mathsf{Z}(\widehat{\delta}(-2S))_{\to\mathtt{T}}\right)}
        \end{align*}
        using $3$ $\mathsf{GT}$ gates (one for each $\prod_{S\in\operatorname{supp}^{>0}(\cdot)}\mathsf{C}_{\mathtt{P}_S}\text{-}\mathsf{Z}(\cdot)_{\to\mathtt{T}}$).
        \item [3.] Uncompute Step~$1$.
    \end{enumerate}
    We now analyse the resources for each step:
    \begin{itemize}
        \item Step $1$: in the Fan-Out-based construction we need to copy each $x_i$, $i\in[n]$, for every\footnote{We do not need to copy $|x_i\rangle$ for $S\ni i$ if $|S|=1$ since the state $\big|\bigoplus_{j\in S}x_j\big\rangle$ already equals $|x_i\rangle$.} $S\in\operatorname{supp}^{>1}(f)$ such that $i\in S$. %Equivalently, we need to copy $x_i$ for $i\in \bigcup_{S\in\operatorname{supp}^{>1}(f)}S$. 
        Registers $\{\mathtt{R}_S\}_S$ thus require $\sum_{i=0}^{n-1} m_i = \sum_{S\in\operatorname{supp}^{>1}(f)} |S|$ ancillae. Such copying can be done with $|\bigcup_{S\in\operatorname{supp}^{>1}(f)}S|$ Fan-Outs with arity $\leq 1+\max_{i\in[n]}m_i$. Moreover, all the parity terms $\bigoplus_{i\in S} x_i$ in registers $\{\mathtt{P}_S\}_S$ for $S\in\operatorname{supp}^{>1}(f)$ require $|\operatorname{supp}^{>1}(f)|$ ancillae and can be computed using either $|\operatorname{supp}^{>1}(f)|$ Fan-Outs with arity $\leq 1+\operatorname{deg}(f)$ or $1$ $\mathsf{GT}$ gate with arity $|\bigcup_{S\in\operatorname{supp}^{>1}(f)}S| + \sum_{S\in\operatorname{supp}^{>1}(f)}|S| \leq 2\sum_{S\in\operatorname{supp}^{>1}(f)}|S|$;
        
        \item Step $2$: constructing $f(x)$ requires either $m-1$ ancillae and $6$ Fan-Out gates with arity $m$ or no ancillae and $3$ $\mathsf{GT}$ gates with arity $m$;
        
        \item Step $3$: either $|\bigcup_{S\in\operatorname{supp}^{>1}(f)}S| + |\operatorname{supp}^{>1}(f)|$ Fan-Out gates or $1$ $\mathsf{GT}$ gate.
    \end{itemize}
    We use $\sum_{S\in\operatorname{supp}^{>1}(f)}|S| + |\operatorname{supp}^{>0}(f)| + |\operatorname{supp}^{>1}(f)| = \sum_{S\in\operatorname{supp}(f)}|S| + 2|\operatorname{supp}^{>1}(f)|$ ancillae and $2|\bigcup_{S\in\operatorname{supp}^{>1}(f)}S|+2|\operatorname{supp}^{>1}(f)|+6$ Fan-Outs with arity $\leq 1+\max\{m,\operatorname{deg}(f)\}$, or $|\operatorname{supp}^{>1}(f)|$ ancillae and $5$ $\mathsf{GT}$ gates with arity $\leq \max\{m,2 \sum_{S\in\operatorname{supp}^{>1}(f)} |S|\} \leq 2\sum_{S\in\operatorname{supp}(f)} |S|$.
\end{proof}

Instead of exactly simulating an $f$-$\mathsf{UCG}^{(n)}$ as in the previous result, it is possible to approximate it by a simpler $\mathsf{UCG}^{(n)}$. For such, we can employ real polynomials that equal $\alpha,\beta,\gamma,\delta$ on all inputs up to a small additive error. We shall use the next technical result.
\begin{lemma}[{\cite[Theorem~5.12]{o2014analysis}}]
    \label{thr:polynomial_approx}
    Let $f:\{0,1\}^n\to\mathbb{R}$ be nonzero, $\epsilon>0$, and $s\geq 4n\hat{\|}f\hat{\|}_1^2/\epsilon^2$ an integer. Then there is a multilinear polynomial $p:\{0,1\}^n\to\mathbb{R}$ of degree and sparsity at most $\operatorname{deg}(f)$ and $s$, respectively, such that $\max_{x\in\{0,1\}^n}|f(x)-p(x)| < \epsilon$.
\end{lemma}

\begin{theorem}[Approximate real implementation of $f$-$\mathsf{UCG}$]\label{thr:ucg_boolean_construction2}
    Let $\epsilon>0$. Given $f:\{0,1\}^n\to\mathcal{U}(\mathbb{C}^{2\times 2})$, let $\alpha,\beta,\gamma,\delta:\{0,1\}^n\to[-1,1]$ be its $\mathsf{Z}$-decomposition. For $\nu\in\{\alpha,\beta,\gamma,\delta\}$, define $s_\nu = \min\{\operatorname{supp}^{>1}(\nu),\big\lceil 64\pi^2n\hat{\|}\nu^{>1}\hat{\|}_1^2/\epsilon^2\big\rceil\}$ and $s = s_\alpha + s_\beta + s_\gamma + s_\delta$. There is a $O(1)$-depth circuit that implements an $f'$-$\mathsf{UCG}$ with $f':\{0,1\}^n\to\mathcal{U}(\mathbb{C}^{2\times 2})$ such that $\max_{x\in\{0,1\}^n}\|f(x) - f'(x)\| \leq \epsilon$ and uses 
    \begin{itemize}
        \item either $s(\operatorname{deg}(f)+2) + |\operatorname{supp}^{=1}(f)|$ ancillae and $2s+2|\bigcup_{S\in\operatorname{supp}^{>1}(f)}S|+6$ Fan-Out gates with arity $\leq 1 + \max\{s+|\operatorname{supp}^{=1}(f)|,\operatorname{deg}(f)\}$,
        \item or $s$ ancillae and $5$ $\mathsf{GT}$ gates with arity $\leq 2s\operatorname{deg}(f) + 2|\operatorname{supp}^{= 1}(f)|$.
    \end{itemize} 
\end{theorem}
\begin{proof}
    Consider the initial state $|x\rangle_{\mathtt{I}}|b\rangle_{\mathtt{T}}$ for  $x\in\{0,1\}^n$ and $b\in\{0,1\}$. We wish to perform the operation $|x\rangle_{\mathtt{I}}|b\rangle_{\mathtt{T}} \mapsto |x\rangle_{\mathtt{I}}f'(x)|b\rangle_{\mathtt{T}}$ for some $f':\{0,1\}^n\to\mathcal{U}(\mathbb{C}^{2\times 2})$ such that $\max_{x\in\{0,1\}^n}\|f(x) - f'(x)\| \leq \epsilon$. Consider, for $\nu\in\{\alpha,\beta,\gamma,\delta\}$, a multilinear polynomial $p_\nu:\{0,1\}^n\to\mathbb{R}$ of degree and sparsity at most $\operatorname{deg}(f)$ and $s_\nu$, respectively, such that $\max_{x\in\{0,1\}^n}|p_\nu(x)-\nu^{>1}(x)| < \sqrt{2\epsilon}/\pi$ according to \Cref{thr:polynomial_approx}.
    Assume without lost of generality that $|\operatorname{supp}(p_\nu)| < |\operatorname{supp}^{>1}(\nu)|$ for $\nu\in\{\alpha,\beta,\gamma,\delta\}$. Our construction is the same as from \Cref{thr:ucg_boolean_construction1}, the only difference is that we now use $p_\alpha$ in place of $\alpha^{>1}$, so $\alpha(x)$ is replaced with $\alpha^{\leq 1}(x) + p_\alpha(x)$ (and similarly for $\beta$, $\gamma$,~$\delta$). This means that $\supp^{>1}(f)$ is replaced with $\supp(p_\alpha)\cup\supp(p_\beta)\cup\supp(p_\gamma)\cup\supp(p_\delta)$. The number of resources follows from \Cref{thr:ucg_boolean_construction1} by replacing
    \begin{align*}
        |\supp^{>1}(f)| &\to |\supp(p_\alpha)\cup\supp(p_\beta)\cup\supp(p_\gamma)\cup\supp(p_\delta)|,\\
        |\supp^{>0}(f)| &\to |\supp^{=1}(f)| + |\supp(p_\alpha)\cup\supp(p_\beta)\cup\supp(p_\gamma)\cup\supp(p_\delta)|,\\
        \Big|{\bigcup}_{S\in\supp^{>1}(f)}S\Big| &\to \Big|{\bigcup}_{S\in\supp(p_\alpha)\cup\supp(p_\beta)\cup\supp(p_\gamma)\cup\supp(p_\delta)}S\Big|,\\
        {\sum}_{S\in\supp(f)}|S| &\to |\supp^{=1}(f)| + {\sum}_{ S\in\supp(p_\alpha)\cup\supp(p_\beta)\cup\supp(p_\gamma)\cup\supp(p_\delta)}|S|,
    \end{align*}
    and bounding
    \begin{align*}
        |\supp(p_\alpha)\cup\supp(p_\beta)\cup\supp(p_\gamma)\cup\supp(p_\delta)| &\leq s,\\
        \Big|{\bigcup}_{S\in\supp(p_\alpha)\cup\supp(p_\beta)\cup\supp(p_\gamma)\cup\supp(p_\delta)}S\Big| &\leq \Big|{\bigcup}_{S\in\supp^{>1}(f)}S\Big|, \\
        {\sum}_{S\in\supp(p_\alpha)\cup\supp(p_\beta)\cup\supp(p_\gamma)\cup\supp(p_\delta)}|S| &\leq s\operatorname{deg}(f).
    \end{align*}

    To show correctness of the circuit, define $\overline{p}_\alpha := \alpha^{\leq 1} + p_\alpha$ for simplicity (and similarly for $\overline{p}_\beta$, $\overline{p}_\gamma$, $\overline{p}_\delta$). Then (omit the $x$ dependence for clarity)
    \begin{align*}
        \|&e^{i\pi \alpha}\mathsf{Z}(\beta)\mathsf{H}\mathsf{Z}(\gamma)\mathsf{H}\mathsf{Z}(\delta) - e^{i\pi \overline{p}_\alpha}\mathsf{Z}(\overline{p}_\beta)\mathsf{H}\mathsf{Z}(\overline{p}_\gamma)\mathsf{H}\mathsf{Z}(\overline{p}_\delta) \|\\
        &\leq \|(e^{i\pi \alpha}-e^{i\pi\overline{p}_\alpha})\mathsf{Z}(\beta)\mathsf{H}\mathsf{Z}(\gamma)\mathsf{H}\mathsf{Z}(\delta)\| + \|e^{i\pi \overline{p}_\alpha}\mathsf{Z}(\beta)\mathsf{H}\mathsf{Z}(\gamma)\mathsf{H}\mathsf{Z}(\delta) - e^{i\pi \overline{p}_\alpha}\mathsf{Z}(\overline{p}_\beta)\mathsf{H}\mathsf{Z}(\overline{p}_\gamma)\mathsf{H}\mathsf{Z}(\overline{p}_\delta) \|\\
        &= 2|\sin(\pi(\alpha - \overline{p}_\alpha)/2)| + \|\mathsf{Z}(\beta)\mathsf{H}\mathsf{Z}(\gamma)\mathsf{H}\mathsf{Z}(\delta) - \mathsf{Z}(\overline{p}_\beta)\mathsf{H}\mathsf{Z}(\overline{p}_\gamma)\mathsf{H}\mathsf{Z}(\overline{p}_\delta) \|\\
        &\leq 2|\sin(\pi(\alpha - \overline{p}_\alpha)/2)| + \|\big(\mathsf{Z}(\beta)-\mathsf{Z}(\overline{p}_\beta)\big)\mathsf{H}\mathsf{Z}(\gamma)\mathsf{H}\mathsf{Z}(\delta)\| + \|\mathsf{Z}(\overline{p}_\beta)\mathsf{H}\big(\mathsf{Z}(\gamma)\mathsf{H}\mathsf{Z}(\delta)-\mathsf{Z}(\overline{p}_\gamma)\mathsf{H}\mathsf{Z}(\overline{p}_\delta)\big) \|\\
        &= 2|\sin(\pi(\alpha - \overline{p}_\alpha)/2)| + 2|\sin(\pi(\beta-\overline{p}_\beta)/2)| + \|\mathsf{Z}(\gamma)\mathsf{H}\mathsf{Z}(\delta)-\mathsf{Z}(\overline{p}_\gamma)\mathsf{H}\mathsf{Z}(\overline{p}_\delta) \|\\
        &\leq 2|\sin(\pi(\alpha - \overline{p}_\alpha)/2)| + 2|\sin(\pi(\beta-\overline{p}_\beta)/2)| + \|\big(\mathsf{Z}(\gamma) - \mathsf{Z}(\overline{p}_\gamma)\big)\mathsf{H}\mathsf{Z}(\delta)\| + \|\mathsf{Z}(\overline{p}_\gamma)\mathsf{H}\big(\mathsf{Z}(\delta)-\mathsf{Z}(\overline{p}_\delta)\big) \| \\
        &= 2|\sin(\pi(\alpha - \overline{p}_\alpha)/2)| + 2|\sin(\pi(\beta-\overline{p}_\beta)/2)| + 2|\sin(\pi(\gamma - \overline{p}_\gamma)/2)| + 2|\sin(\pi(\delta-\overline{p}_\delta)/2)|\\
        &\leq 8|\sin(\epsilon/8)| \leq \epsilon. \qedhere
    \end{align*}
\end{proof}

Our final construction uses the real-polynomial $\{0,1\}$-representation based on $\mathsf{AND}$ functions which can be computed using \Cref{thr:or_constantdepth} or \Cref{thr:or_constantdepth_gt}.
\begin{theorem}[Real $\{0,1\}$-implementation of $f$-$\mathsf{UCG}$]\label{thr:ucg_boolean_construction3}
    Given $f:\{0,1\}^n\to\mathcal{U}(\mathbb{C}^{2\times 2})$, there is a $O(1)$-depth circuit for $f\text{-}\mathsf{UCG}$ that uses 
    \begin{itemize}
        \item either $\sum_{S\in\operatorname{supp}^{>0}_{\{0,1\}}(f)} \big(2|S|\log|S| + O(|S|)\big)$ ancillae and $\sum_{S\in\operatorname{supp}^{>1}_{\{0,1\}}(f)} \big(8|S| + O(\log|S|)\big)$ Fan-Out gates with arity $\leq \max\{1+|\operatorname{supp}^{>0}_{\{0,1\}}(f)|,2\operatorname{deg}(f)\}$,
        \item or $\sum_{S\in\operatorname{supp}^{>1}_{\{0,1\}}(f)} \big(3|S| + O(\log|S|)\big)$ ancillae and $9$ $\mathsf{GT}$ gates with arity $\leq 2\sum_{S\in\operatorname{supp}_{\{0,1\}}(f)} |S|$.
    \end{itemize}
\end{theorem}
\begin{proof}
    Consider the initial state $|x\rangle_{\mathtt{I}}|b\rangle_{\mathtt{T}}$ for  $x\in\{0,1\}^n$ and $b\in\{0,1\}$. We wish to implement $|x\rangle_{\mathtt{I}}|b\rangle_{\mathtt{T}} \mapsto |x\rangle_{\mathtt{I}}f(x)|b\rangle_{\mathtt{T}}$. Let $\alpha(x) = \sum_{S\subseteq[n]} \widetilde{\alpha}(S)x^S$ be the real-polynomial $\{0,1\}$-representation of $\alpha$, and similarly for $\beta,\gamma,\delta$. Write $m:= |\operatorname{supp}^{>0}_{\{0,1\}}(f)|$ and $m_i := |\{S\in\operatorname{supp}^{>1}_{\{0,1\}}(f):i\in S\}|$ for the number of sets of size greater than $1$ that contain the coordinate $i\in[n]$. 
    
    Consider first the Fan-Out-based circuit:
    \begin{enumerate}
        \item Attach an ancillary register $\bigotimes_{S\in\operatorname{supp}^{>1}_{\{0,1\}}(f)}|0\rangle^{\otimes |S|}_{\mathtt{R}_S}$. For each $i\in[n]$ in parallel, copy $m_i$ number of times the qubit $|x_i\rangle_{\mathtt{I}}$ by using a $(1+m_i)$-arity Fan-Out to obtain
        \begin{align*}
            |x\rangle_{\mathtt{I}}|b\rangle_{\mathtt{T}} \mapsto |x\rangle_{\mathtt{I}} |b\rangle_{\mathtt{T}}\bigotimes_{S\in\operatorname{supp}^{>1}_{\{0,1\}}(f)}|x_S\rangle_{\mathtt{R}_S}.
        \end{align*}
        \item Attach an ancillary register $\bigotimes_{S\in\operatorname{supp}^{>1}_{\{0,1\}}(f)}|0\rangle_{\mathtt{P}_S}$. For each $S\in\operatorname{supp}^{>1}_{\{0,1\}}(f)$ in parallel, apply an $\mathsf{AND}_{\mathtt{R}_S\to\mathtt{P}_S}^{(|S|)}$ gate to get
        \begin{align*}
            |x\rangle_{\mathtt{I}} |b\rangle_{\mathtt{T}}\bigotimes_{S\in\operatorname{supp}^{>1}_{\{0,1\}}(f)}|x_S\rangle_{\mathtt{R}_S} \mapsto |x\rangle_{\mathtt{I}}|b\rangle_{\mathtt{T}}\bigotimes_{S\in\operatorname{supp}^{>1}_{\{0,1\}}(f)}|x_S\rangle_{\mathtt{R}_S}|x^S\rangle_{\mathtt{P}_S}.
        \end{align*}
        \item [3a.] Attach an ancillary register $|0\rangle^{\otimes (m-1)}_{\mathtt{T}'}$ and apply an $m$-arity Fan-Out gate $\mathsf{FO}^{(m)}_{\mathtt{T}\to\mathtt{T}'}$ from register $\mathtt{T}$ onto register $\mathtt{T}'$.  Apply a $\mathsf{Z}(\widetilde{\delta}(\emptyset))_{\to\mathtt{T}}$ gate onto register $\mathtt{T}$. Then, for each $S\in\operatorname{supp}^{>0}_{\{0,1\}}(\delta)$ in parallel, apply a $\mathsf{Z}(\widetilde{\delta}(S))$ gate controlled on register $\mathtt{P}_S$ onto the $S$-th qubit in register $\mathtt{T}'$ (if $|S|=1$, apply the gate onto $|x_S\rangle_{\mathtt{I}}$). Finally, apply $\mathsf{FO}^{(m)}_{\mathtt{T}\to\mathtt{T}'}$ again. This chain of operations leads to (omit registers $\mathtt{P}_S$ and $\mathtt{R}_S$ for simplicity)
        \begin{align*}
            |x\rangle_{\mathtt{I}}|b\rangle_{\mathtt{T}} \mapsto 
            |x\rangle_{\mathtt{I}}|b\rangle^{\otimes m}_{\mathtt{T},\mathtt{T}'}
            \mapsto |x\rangle_{\mathtt{I}}\mathsf{Z}\left(\sum_{S\subseteq[n]}\widetilde{\delta}(S)x^S \right) |b\rangle^{\otimes m}_{\mathtt{T},\mathtt{T}'}
            \mapsto |x\rangle_{\mathtt{I}}\mathsf{Z}(\delta(x)) |b\rangle_{\mathtt{T}}.
        \end{align*}
        \item [3b.] Apply a $\mathsf{H}_{\to\mathtt{T}}$ gate onto register $\mathtt{T}$ followed by an $m$-arity Fan-Out gate $\mathsf{FO}^{(m)}_{\mathtt{T}\to\mathtt{T}'}$ from register $\mathtt{T}$ onto register $\mathtt{T}'$. Apply a $\mathsf{Z}(\widetilde{\gamma}(\emptyset))_{\to\mathtt{T}}$ gate onto register $\mathtt{T}$. Then, for each $S\in\operatorname{supp}^{>0}_{\{0,1\}}(\gamma)$ in parallel, apply a $\mathsf{Z}(\widetilde{\gamma}(S))$ gate controlled on register $\mathtt{P}_S$ onto the $S$-th qubit in register $\mathtt{T}'$ (if $|S|=1$, apply the gate onto $|x_S\rangle_{\mathtt{I}}$). Finally, apply $\mathsf{FO}^{(m)}_{\mathtt{T}\to\mathtt{T}'}$ again. We obtain
        \begin{align*}
            |x\rangle_{\mathtt{I}}\mathsf{H}\mathsf{Z}(\delta(x))|b\rangle_{\mathtt{T}} \mapsto  |x\rangle_{\mathtt{I}}\mathsf{Z}\left(\sum_{S\subseteq[n]}\widetilde{\gamma}(S)x^S \right)\mathsf{H}\mathsf{Z}(\delta(x)) |b\rangle_{\mathtt{T}}
            = |x\rangle_{\mathtt{I}}\mathsf{Z}(\gamma(x)) \mathsf{H}\mathsf{Z}(\delta(x))|b\rangle_{\mathtt{T}}.
        \end{align*}

        \item [3c.] Apply a $\mathsf{H}_{\to\mathtt{T}}$ gate onto register $\mathtt{T}$ followed by an $m$-arity Fan-Out gate $\mathsf{FO}^{(m)}_{\mathtt{T}\to\mathtt{T}'}$ from register $\mathtt{T}$ onto register $\mathtt{T}'$.  Apply a $\mathsf{Z}(\widetilde{\beta}(\emptyset))_{\to\mathtt{T}}$ gate onto register $\mathtt{T}$. Then, for each $S\in\operatorname{supp}^{>0}_{\{0,1\}}(\beta)$ in parallel, apply a $\mathsf{Z}(\widetilde{\beta}(S))$ gate controlled on register $\mathtt{P}_S$ onto the $S$-th qubit in register $\mathtt{T}'$ (if $|S|=1$, apply the gate onto $|x_S\rangle_{\mathtt{I}}$). Finally, apply $\mathsf{FO}^{(m)}_{\mathtt{T}\to\mathtt{T}'}$ again. Similarly to the previous steps, this chain of operations leads to
        \begin{align*}
            |x\rangle_{\mathtt{I}}\mathsf{H}\mathsf{Z}(\gamma(x)) \mathsf{H}\mathsf{Z}(\delta(x))|b\rangle_{\mathtt{T}}
            \mapsto |x\rangle_{\mathtt{I}}\mathsf{Z}(\beta(x))\mathsf{H}\mathsf{Z}(\gamma(x)) \mathsf{H}\mathsf{Z}(\delta(x))|b\rangle_{\mathtt{T}}.
        \end{align*}

        \item [3d.] Apply an overall phase $e^{i\pi \widetilde{\alpha}(\emptyset)}$. Then, for each $S\in\operatorname{supp}^{>0}_{\{0,1\}}(\alpha)$ in parallel, apply a $\mathsf{Z}(\widetilde{\alpha}(S))$ gate onto register $\mathtt{P}_S$ (if $|S|=1$, apply the gate onto $|x_S\rangle_{\mathtt{I}}$). This yields
        \begin{align*}
            |x\rangle_{\mathtt{I}} e^{i\pi \alpha(x)} \mathsf{Z}(\beta(x))\mathsf{H}\mathsf{Z}(\gamma(x)) \mathsf{H}\mathsf{Z}(\delta(x))|b\rangle_{\mathtt{T}} = |x\rangle_{\mathtt{I}}f(x)|b\rangle_{\mathtt{T}}.
        \end{align*}

        \item [4.] Uncompute Steps~$1$ and~$2$.
    \end{enumerate}
    We now consider the $\mathsf{GT}$-gate-based circuit. Steps~$3$a-d are replaced with the following Step~$3$. In the following, write register $\mathtt{I}$ as $|x\rangle_{\mathtt{I}} = \bigotimes_{i\in[n]}|x_i\rangle_{\mathtt{I}_i}$.
    \begin{enumerate}
        \item [3.] Apply the gate (write $\mathtt{P}_S := \mathtt{I}_S$ if $|S|=1$ and $\mathtt{P}_{\emptyset} := \emptyset$)
        \begin{align*}
        \begin{multlined}[b][0.94\textwidth]
            \left(e^{i\pi \widetilde{\alpha}(\emptyset)}\prod_{S\in\operatorname{supp}_{\{0,1\}}^{>0}(\alpha)}\mathsf{Z}(\widetilde{\alpha}(S))_{\to\mathtt{P}_{S}}\right)\left(\prod_{S\in\operatorname{supp}_{\{0,1\}}(\beta)}\mathsf{C}_{\mathtt{P}_S}\text{-}\mathsf{Z}(\widetilde{\beta}(S))_{\to\mathtt{T}}\right)\mathsf{H}_{\to\mathtt{T}} \\
            \cdot\left(\prod_{S\in\operatorname{supp}_{\{0,1\}}(\gamma)}\mathsf{C}_{\mathtt{P}_S}\text{-}\mathsf{Z}(\widetilde{\gamma}(S))_{\to\mathtt{T}}\right)\mathsf{H}_{\to\mathtt{T}}\left(\prod_{S\in\operatorname{supp}_{\{0,1\}}(\delta)}\mathsf{C}_{\mathtt{P}_{S}}\text{-}\mathsf{Z}(\widetilde{\delta}(S))_{\to\mathtt{T}}\right)
        \end{multlined}
        \end{align*}
        using $3$ $\mathsf{GT}$ gates (one for each $\prod_{S\in\operatorname{supp}_{\{0,1\}}(\cdot)}\mathsf{C}_{\mathtt{P}_S}\text{-}\mathsf{Z}(\cdot)_{\to\mathtt{T}}$).
    \end{enumerate}
    
    We now analyse the resources for each step:
    \begin{itemize}
        \item Step~$1$: we need to copy each $x_i$, $i\in[n]$, for every $S\in\operatorname{supp}^{>1}_{\{0,1\}}(f)$ such that $i\in S$.  Thus registers $\{\mathtt{R}_S\}_S$ require $\sum_{i\in[n]} m_i = \sum_{S\in\operatorname{supp}^{>1}_{\{0,1\}}(f)} |S|$ ancillae. Such copying can be done with either $|\bigcup_{S\in\operatorname{supp}^{>1}_{\{0,1\}}(f)}S| \leq \sum_{S\in\supp^{>1}_{\{0,1\}}(f)}|S|$ Fan-Outs, each with arity at most $1+\max_{i\in[n]}m_i$, or $1$ $\mathsf{GT}$ gate with arity $|\bigcup_{S\in\operatorname{supp}^{>1}_{\{0,1\}}(f)}S|+\sum_{S\in\operatorname{supp}^{>1}_{\{0,1\}}(f)} |S| \leq 2\sum_{S\in\operatorname{supp}^{>1}_{\{0,1\}}(f)} |S|$;
        
        \item Step $2$: the $|\operatorname{supp}^{>1}_{\{0,1\}}(f)|$ $\mathsf{AND}^{(|S|)}_{\mathtt{R}_S\to\mathtt{P}_S}$ gates can be performed with either
        \begin{align*}
            \sum_{S\in\operatorname{supp}^{>1}_{\{0,1\}}(f)} \big(2|S|\log|S| + O(|S|)\big) ~\text{ancillae} \quad \text{and} \quad  \sum_{S\in\operatorname{supp}^{>1}_{\{0,1\}}(f)} \big(3|S| + O(\log|S|)\big) ~\text{Fan-Outs}
        \end{align*}
        with arity at most $2\operatorname{deg}(f)$ (\Cref{thr:or_constantdepth}), or 
        \begin{align*}
            \sum_{S\in\operatorname{supp}^{>1}_{\{0,1\}}(f)} \big(2|S| + O(\log|S|)\big) ~\text{ancillae} \quad \text{and} \quad 2~ \mathsf{GT}~\text{gates}
        \end{align*}
        with arity at most $2\operatorname{deg}(f) + O(\log\operatorname{deg}(f))$ (\Cref{thr:or_constantdepth_gt}) if we postpone their inner uncomputation part until Step~$5$;

        \item Step $3$: constructing $f(x)$ requires either $m-1$ ancillae and $6$ Fan-Out gates with arity $m$ or no ancillae and 3 $\mathsf{GT}$ gates with arity $m$;

        \item Step $4$: uncomputing Steps~$1$ and~$2$ requires either $\sum_{S\in\operatorname{supp}^{>1}_{\{0,1\}}(f)} \big(4|S| + O(\log|S|)\big)$ Fan-Out gates or $3$ $\mathsf{GT}$ gates.
    \end{itemize}
    We require either $\sum_{S\in\operatorname{supp}^{>0}_{\{0,1\}}(f)} \big(2|S|\log|S| + O(|S|)\big)$ ancillae and $\sum_{S\in\operatorname{supp}^{>1}_{\{0,1\}}(f)} \big(8|S| + O(\log|S|)\big)$ Fan-Out gates with arity $\leq\max\{1+|\operatorname{supp}^{>0}_{\{0,1\}}(f)|,2\operatorname{deg}(f)\}$ or $\sum_{S\in\operatorname{supp}^{>1}_{\{0,1\}}(f)} \big(3|S| + O(\log|S|)\big)$ ancillae and $9$ $\mathsf{GT}$ gates with arity $\leq 2\sum_{S\in\operatorname{supp}_{\{0,1\}}(f)} |S|$.
\end{proof}

\subsection{Constant-depth circuits for $f$-$\mathsf{FIN}$s}

Similarly to \Cref{sec:fin_onehot}, we now show how the circuits from the previous section can be simplified and used to implement $f$-$\mathsf{FIN}$s. 
\begin{theorem}[Real implementation of $f$-$\mathsf{FIN}$]\label{thr:fin_boolean_construction1}
    Given $f:\{0,1\}^n\to\{0,1\}$, there is a $O(1)$-depth circuit for $f\text{-}\mathsf{FIN}$ that uses 
    \begin{itemize}
        \item either $\sum_{S\in\operatorname{supp}(f)}|S| + 2|\operatorname{supp}^{>1}(f)|$ ancillae and $2|\bigcup_{S\in\operatorname{supp}^{>1}(f)}S|+2|\operatorname{supp}^{>1}(f)|+2$ Fan-Out gates with arity $\leq 1+\max\{|\operatorname{supp}^{>0}(f)|,\operatorname{deg}(f)\}$, 
        \item or $2|\operatorname{supp}^{>0}(f)|$ ancillae and $2$ $\mathsf{GT}$ gates with arity $\leq 3\sum_{S\in\operatorname{supp}(f)} |S|$.
    \end{itemize}
\end{theorem}
\begin{proof}
    Since an $f$-$\mathsf{FIN}$ is simply an $f'$-$\mathsf{UCG}$ whose $f'$ $\mathsf{Z}$-decomposition is $\mathsf{H}\mathsf{Z}(f(x))\mathsf{H}$, Step~$2$ in \Cref{thr:ucg_boolean_construction1} only requires $2$ Fan-Out gates or $1$ $\mathsf{GT}$ gate. This gives the resource count for the Fan-Out-based construction and the number of $\mathsf{GT}$ gates is reduced to $3$. 
    It is possible to further reduce the number of $\mathsf{GT}$ gates to $2$ by using the $(|\supp^{>0}(f)|-1)$-qubit register $\mathtt{T}'$ and, similarly to the Fan-Out-based construction, to use the parallelisation method depicted in \Cref{fig:two_methods} (see proof of \Cref{thr:fin_onehot}). This increases the number of ancillae to $|\supp^{>0}(f)| + |\supp^{>1}(f)| \leq 2|\supp^{>0}(f)|$. The new $\mathsf{GT}$ gates' arity is $|\bigcup_{S\in\operatorname{supp}^{>1}(f)}S| + \sum_{S\in\operatorname{supp}^{>1}(f)} |S| + |\supp^{>0}(f)| \leq 3\sum_{S\in\operatorname{supp}(f)} |S|$.
\end{proof}

To obtain the next result on an approximate implementation of $f$-$\mathsf{FIN}$s, the modifications to be made to \Cref{thr:ucg_boolean_construction2} are the same that were conducted in the previous theorem. Recall that an $f$-$\mathsf{FIN}$ is an $f'$-$\mathsf{UCG}$ such that $f'(x) = \mathsf{X}^{f(x)}$. Therefore, an approximate circuit for an $f$-$\mathsf{FIN}$ implements an $f'$-$\mathsf{UCG}$ with $f'(x)$ close to $\mathbb{I}_1$ or $\mathsf{X}$.
\begin{theorem}[Approximate real implementation of $f$-$\mathsf{FIN}$]\label{thr:fin_boolean_construction2}
    Let $\epsilon>0$, $f:\{0,1\}^n\to\{0,1\}$, and $s = \min\{\operatorname{supp}^{>1}(f),\big\lceil 4\pi^2n\hat{\|}f^{>1}\hat{\|}_1^2/\epsilon^2\big\rceil\}$. There is a $O(1)$-depth circuit that implements an $f'$-$\mathsf{UCG}$ with $f':\{0,1\}^n\to\mathcal{U}(\mathbb{C}^{2\times 2})$ such that $\max_{x\in\{0,1\}^n}\|f'(x) - \mathsf{X}^{f(x)}\| \leq \epsilon$ and uses 
    \begin{itemize}
        \item either $s(\operatorname{deg}(f)+2) + |\operatorname{supp}^{= 1}(f)|$ ancillae and $2s+2|\bigcup_{S\in\operatorname{supp}^{>1}(f)}S|+2$ Fan-Out gates with arity $\leq s+|\operatorname{supp}^{= 1}(f)|$,
        \item or $2s + |\operatorname{supp}^{= 1}(f)|$ ancillae and $2$ $\mathsf{GT}$ gates with arity $\leq s(2\operatorname{deg}(f) + 1) + |\operatorname{supp}^{= 1}(f)|$.
    \end{itemize} 
\end{theorem}

Finally, \Cref{thr:ucg_boolean_construction3} can be considerably simplified in the case of $f$-$\mathsf{FIN}$s since we can use the $\mathbb{F}_2$-polynomial representation of $f(x)$ instead of its real $\{0,1\}$-representation.
\begin{theorem}[$\mathbb{F}_2$-implementation of $f$-$\mathsf{FIN}$]\label{thr:fin_boolean_construction3}
    Given $f:\{0,1\}^n\to\{0,1\}$, there is a $O(1)$-depth circuit for $f\text{-}\mathsf{FIN}$ that uses 
    \begin{itemize}
        \item either $\sum_{S\in\operatorname{supp}^{>1}_{\mathbb{F}_2}(f)} \big(2|S|\log|S| + O(|S|)\big)$ ancillae and $\sum_{S\in\operatorname{supp}^{>1}_{\mathbb{F}_2}(f)} \big(8|S| + O(\log|S|)\big)$ Fan-Out gates with arity $\leq \max\{1+|{\operatorname{supp}}^{>0}_{\mathbb{F}_2}(f)|,2\operatorname{deg}_{\mathbb{F}_2}(f)\}$,
        \item or $\sum_{S\in\operatorname{supp}^{>0}_{\mathbb{F}_2}(f)} \big(4|S| + O(\log|S|)\big)$ ancillae and $6$ $\mathsf{GT}$ gates with arity $\leq 3\sum_{S\in\operatorname{supp}_{\mathbb{F}_2}(f)} |S|$.
    \end{itemize}
\end{theorem}
\begin{proof}
    After constructing the state
    \begin{align*}
        |x\rangle_{\mathtt{I}}|b\rangle_{\mathtt{T}}\bigotimes_{S\in\operatorname{supp}^{>1}_{\mathbb{F}_2}(f)}|x_S\rangle_{\mathtt{R}_S}|x^S\rangle_{\mathtt{P}_S}
    \end{align*}
    as in \Cref{thr:ucg_boolean_construction3}, apply a $\mathsf{X}_{\to\mathtt{T}}^{\widetilde{f}_{\mathbb{F}_2}(\emptyset)}$ gate onto register $\mathtt{T}$, a $\mathsf{PARITY}_{\{\mathtt{I}_j\}_{j\in\operatorname{supp}_{\mathbb{F}_2}^{=1}(f)}\to\mathtt{T}}$ gate from registers $\{\mathtt{I}_j\}_{j\in\operatorname{supp}_{\mathbb{F}_2}^{=1}(f)}$ onto register $\mathtt{T}$ ($\mathtt{I}_j$ contains $x_j$), and finally a $\mathsf{PARITY}_{\{\mathtt{P}_S\}_{S\in\operatorname{supp}_{\mathbb{F}_2}^{>1}(f)}\to\mathtt{T}}$ gate from registers $\{\mathtt{P}_S\}_{S\in\operatorname{supp}_{\mathbb{F}_2}^{>1}(f)}$ onto register $\mathtt{T}$ (both $\mathsf{PARITY}$ gates can be performed together). We get
    \begin{align*}
        |x\rangle_{\mathtt{I}}|b\rangle_{\mathtt{T}}\bigotimes_{S\in\operatorname{supp}^{>1}_{\mathbb{F}_2}(f)}|x_S\rangle_{\mathtt{R}_S}|x^S\rangle_{\mathtt{P}_S} \mapsto ~&|x\rangle_{\mathtt{I}}\big|b\oplus \widetilde{f}_{\mathbb{F}_2}(\emptyset)\oplus {\bigoplus}_{S\in\operatorname{supp}^{>1}_{\mathbb{F}_2}(f)}x^S\big\rangle_{\mathtt{T}}\bigotimes_{S\in\operatorname{supp}^{>1}_{\mathbb{F}_2}(f)}|x_S\rangle_{\mathtt{R}_S}|x^S\rangle_{\mathtt{P}_S} \\
        = ~&|x\rangle_{\mathtt{I}}\big|b\oplus {\bigoplus}_{S\subseteq[n]}\widetilde{f}_{\mathbb{F}_2}(S)x^S\big\rangle_{\mathtt{T}}\bigotimes_{S\in\operatorname{supp}^{>1}_{\mathbb{F}_2}(f)}|x_S\rangle_{\mathtt{R}_S}|x^S\rangle_{\mathtt{P}_S} \\
        =~&|x\rangle_{\mathtt{I}}|b\oplus f(x)\rangle_{\mathtt{T}}\bigotimes_{S\in\operatorname{supp}^{>1}_{\mathbb{F}_2}(f)}|x_S\rangle_{\mathtt{R}_S}|x^S\rangle_{\mathtt{P}_S}.
    \end{align*}
    Uncomputing registers $\{\mathtt{R}_S\}_S$ and $\{\mathtt{P}_S\}_S$ gives the desired state. As in \Cref{thr:ucg_boolean_construction3}, the total cost of computing and uncomputing $\{\mathtt{R}_S\}_S$ and $\{\mathtt{P}_S\}_S$ is either 
    \begin{align*}
        \sum_{S\in\operatorname{supp}^{>1}_{\mathbb{F}_2}(f)} \big(2|S|\log|S| + O(|S|)\big) ~\text{ancillae} \quad \text{and} \quad \sum_{S\in\operatorname{supp}^{>1}_{\mathbb{F}_2}(f)} \big(8|S| + O(\log|S|)\big) ~\text{Fan-Outs}
    \end{align*}
    with arity $\leq \max\{1+|{\operatorname{supp}}^{>0}_{\mathbb{F}_2}(f)|,2\operatorname{deg}_{\mathbb{F}_2}(f)\}$ or 
    \begin{align*}
        \sum_{S\in\operatorname{supp}^{>1}_{\mathbb{F}_2}(f)} \big(3|S| + O(\log|S|)\big) ~\text{ancillae} \quad\text{and}\quad 6 ~\mathsf{GT}~ \text{gates}
    \end{align*}
    with arity $\leq 2\sum_{S\in\operatorname{supp}_{\mathbb{F}_2}(f)} |S|$. The $\mathsf{PARITY}$ gates cost another $(1+|\supp^{>0}_{\mathbb{F}_2}(f)|)$-arity Fan-Out or $\mathsf{GT}$ gate. It is possible to reduce the number of $\mathsf{GT}$ gates from $7$ to $6$ by using $|\operatorname{supp}^{>0}_{\mathbb{F}_2}(f)|$ extra ancillae similarly to \Cref{thr:fin_boolean_construction1}.
\end{proof}

\subsection{Constant-depth circuits for quantum memory devices via Boolean analysis}
\label{sec:boolean_qram_construction}

In this section, we apply our Boolean-based circuit constructions to the case of $\mathsf{QRAM}$. We first compute the Fourier coefficients of $f(x,i) = x_i$.
\begin{lemma}\label{lem:fourier_coefficients}
    Let $n\in\mathbb{N}$ be a power of $2$ and let $f:\{0,1\}^n\times\{0,1\}^{\log{n}}\to\{0,1\}$ be the Boolean function $f(x,i) = x_i$. The Fourier coefficients of $f$ are
    \begin{align*}
        \widehat{f}(S,T) = \begin{cases}
            \frac{1}{2} &\text{if}~ (S,T) = (\emptyset,\emptyset),\\
            \frac{-\chi_T(k)}{2n} &\forall (S,T)\subseteq[n]\times[\log{n}], S=\{k\},\\
            0 &\text{otherwise},
        \end{cases}
    \end{align*}
    where $\chi_T(k) = (-1)^{\sum_{i\in T} k_i}$ and $k = k_0\dots k_{\log{n}-1}\in\{0,1\}^{\log{n}}$.
\end{lemma}
\begin{proof}
    By a straightforward calculation,
    \begin{align*}
        \widehat{f}(\emptyset,T) = \frac{1}{2^{n}n}\sum_{i\in\{0,1\}^{\log{n}}}\chi_T(i)\sum_{x\in\{0,1\}^n}x_i = \frac{1}{2n}\sum_{i\in\{0,1\}^{\log{n}}} \chi_T(i) = \begin{cases}
            \frac{1}{2} &\text{if}~T = \emptyset,\\
            0 &\text{otherwise}.
        \end{cases}
    \end{align*}
    Moreover, 
    \begin{align*}
        \widehat{f}(S,T) = \frac{1}{2^n n}\sum_{i\in\{0,1\}^{\log{n}}}\chi_T(i)\sum_{x\in\{0,1\}^n}x_i(-1)^{\sum_{j\in S} x_j} = \begin{cases}
            \frac{-\chi_T(k)}{2n} &\text{if}~S = \{k\},\\
            0 &\text{if}~ |S| \geq 2,
        \end{cases}
    \end{align*}
    for every $T\subseteq[\log{n}]$, since $\sum_{x\in\{0,1\}^n}x_i(-1)^{x_j}$ equals $0$ if $i\neq j$ and $-2^{n-1}$ if $i=j$.
\end{proof}
\begin{theorem}[Real implementation of $\QRAM$]\label{lem:qramfourier}
    Let $n\in\mathbb{N}$ be a power of $2$. A $\QRAM$ of memory size $n$ can be implemented in $O(1)$-depth using
    \begin{itemize}
        \item either $\frac{1}{2}n^2\log{n} + O(n^2)$ ancillae and $2n^2 + O(n\log{n})$ Fan-Out gates with arity $\leq 1+n^2$,
        \item or $2n^2$ ancillae and $2$ $\mathsf{GT}$ gates with arity $\leq \frac{1}{2}n^2\log{n} + O(n^2)$.
    \end{itemize}
\end{theorem}
\begin{proof}
    By \Cref{lem:fourier_coefficients}, $\operatorname{supp}(f) = \{(S,T)\subseteq[n]\times[\log{n}]:|S| = 1\}\cup\{(\emptyset,\emptyset)\}$. Thus $|\operatorname{supp}^{>0}(f)| = n^2$, $|\operatorname{supp}^{>1}(f)| = n^2 - n$, $|\bigcup_{(S,T)\in\operatorname{supp}^{>1}(f)}(S,T)| = n+\log{n}$, $\operatorname{deg}(f) = 1+\log{n}$, and
    \begin{align*}
        \sum_{(S,T)\in\operatorname{supp}(f)} |S| + |T| = n\sum_{k=0}^{\log{n}}\binom{\log{n}}{k}(1+k) = n^2 + \frac{n^2\log{n}}{2}.
    \end{align*}
    By \Cref{thr:fin_boolean_construction1}, there is a $O(1)$-depth circuit for $\mathsf{QRAM}$ that uses either $\frac{1}{2}n^2\log{n} + O(n^2)$ ancillae and $2n^2 + O(n\log{n})$ Fan-Out gates with arity at most $1+n^2$, or $2n^2$ ancillae and $2$ $\mathsf{GT}$ gates with arity at most $\frac{1}{2}n^2\log{n} + O(n^2)$.
\end{proof}
\begin{remark}
    Since the $\mathbb{F}_2$-support of $f:\{0,1\}^{n}\times\{0,1\}^{\log{n}} \to \{0,1\}$, $f(x,i) = x_i$, is quite dense, {\rm \Cref{thr:fin_boolean_construction3}} is not suited for constructing an efficient $\QRAM$. Moreover,
    \begin{align*}
        \hat{\|}f^{>1}\hat{\|}_1 = \sum_{(S,T)\in\operatorname{supp}^{>1}(f)}|\widehat{f}(S,T)| = \frac{1}{2n}n(2^{\log{n}} - 1) = \frac{n-1}{2}.
    \end{align*}
    Therefore, $s = \big\lceil 4\pi^2n\hat{\|}f^{>1}\hat{\|}_1^2/\epsilon^2\big\rceil = \lceil \pi^2 n^3/\epsilon^2\rceil - O(n^2)$ and the approximate real constructions in {\rm \Cref{thr:fin_boolean_construction2}} require many more ancillae compared to {\rm \Cref{lem:qramfourier}}.
\end{remark}

Similarly to the one-hot encoding construction, it is possible to reduce the number of ancillae by reducing the problem into small blocks and solving each with a small $\mathsf{QRAM}$ circuit.
\begin{theorem}\label{thr:qram_recursive_procedure_boolean}
    For every $n,d \in \mathbb{N}$, a $\mathsf{QRAM}$ of memory size $n$ can be implemented in $O(d)$-depth using
    \begin{itemize}
        \item either $O\big(n^{1/(1-2^{-d})}\log{n}\big)$ ancillae and $O\big(n^{1/(1-2^{-d})}\big)$ Fan-Out gates,
        \item or $O\big(n^{1/(1-2^{-d})}\big)$ ancillae and $8d-6$ $\mathsf{GT}$ gates.
    \end{itemize}
\end{theorem}
\begin{proof}
    The proof is very similar to \Cref{thr:qram_recursive_procedure}, only the size of the blocks is different. For $d=1$, the result follows from \Cref{lem:qramfourier}. For the induction step, we divide the input $x\in\{0,1\}^n$ into $m := n^{(2^d-2)/(2^d-1)}$ blocks of $b := n^{1/(2^d-1)}$ qubits each. The $d$-th $\mathsf{QRAM}$ circuit thus uses either $O(mb^2\log{b}) = O\big(n^{1/(1-2^{-d})}\log{n}\big)$ ancillae and $O(mb^2) = O\big(n^{1/(1-2^{-d})}\big)$ Fan-Outs, or $O\big(n^{1/(1-2^{-d})}\big)$ ancillae and $2$ $\mathsf{GT}$ gates (the $\mathsf{GT}$ gates from different blocks can be done in parallel). We are left with $m = n^{(2^d-2)/(2^d-1)}$ output qubits. Using the induction hypothesis, the remaining $d-1$ $\mathsf{QRAM}$-levels uses either $O\big(m^{1/(1-2^{-d+1})}\log{m}\big) = O\big(n^{1/(1-2^{-d})}\log{n}\big)$ ancillae and $O\big(m^{1/(1-2^{-d+1})}\big) = O\big(n^{1/(1-2^{-d})}\big)$ Fan-Outs, or $O\big(n^{1/(1-2^{-d})}\big)$ ancillae and $8(d-1) - 6$ $\mathsf{GT}$ gates. The output qubit is $|b\oplus x_{bq+r}\rangle_{\mathtt{T}} = |b\oplus x_{i}\rangle_{\mathtt{T}}$ as required. 

    Regarding the resources, we must take into consideration the computation of the quotient $q := \lfloor i/b\rfloor$ and remainder $r\equiv i~(\text{mod}~b)$ (and copying/uncopying the remainder $m-1$ times). As mentioned in the proof of \Cref{thr:qram_recursive_procedure}, $q$ can be computed with a depth-$4$ polynomial-size threshold circuit and $r$ can be computed with a depth-$2$ polynomial-size threshold circuit~\cite{siu1993depth}. In the Fan-Out-based construction, $\poly\log{n}$ ancillae and Fan-Out gates are sufficient since $i$ is a $\log{n}$-bit string. Regarding the $\mathsf{GT}$-gate-based construction, we employ \Cref{thr:fin_boolean_construction1} to perform a $\mathsf{THRESHOLD}$ function using $O(|\operatorname{supp}^{>0}(\mathsf{THRESHOLD}^{(\log{n})})|) = O(n)$ ancillae and $2$ $\mathsf{GT}$ gates. Each layer of the threshold circuit is made of $\poly\log{n}$ $\mathsf{THRESHOLD}$ functions. It is then possible to either 
    \begin{enumerate}
        \item compute all the $\mathsf{PARITY}$ functions $x_S$, $S\in\operatorname{supp}^{>0}(\mathsf{THRESHOLD}^{(\log{n})})$, a number of $\poly\log{n}$ times in parallel, so that each $\mathsf{THRESHOLD}$ function has its own set $\{x_S\}_S$, and this requires $O(n\poly\log{n})$ ancillae and $1$ $\mathsf{GT}$ gate (plus $1$ $\mathsf{GT}$ gate for uncomputation). Therefore, computing $q$ uses $O(n\poly\log{n})$ ancillae and $8$ $\mathsf{GT}$ gates ($2$ per layer), and computing $r$ uses $O(n\poly\log{n})$ ancillae and $4$ $\mathsf{GT}$ gates;
        \item or compute the $\mathsf{PARITY}$ functions $\{x_S\}_S$ just once and share them with all $\mathsf{THRESHOLD}$ functions of a layer. This means that, in order to apply the gates $\mathsf{Z}(\widehat{f}(S))$ controlled on $x_S$ (cf.\ \Cref{fig:boolean_construction_gt}), we require another $\mathsf{GT}$ gate since $x_S$ will serve as control-qubit for all $\mathsf{THRESHOLD}$ functions (plus $1$ $\mathsf{GT}$ for uncomputation). This requires $O(n)$ ancillae and $4$ $\mathsf{GT}$ gates. Computing $q$ uses $O(n)$ ancillae and $16$ $\mathsf{GT}$ gates ($4$ per layer), and computing $r$ uses $O(n)$ ancillae and $8$ $\mathsf{GT}$ gates.
    \end{enumerate}
    For the one-hot-based circuit (\Cref{thr:qram_recursive_procedure}), we employ the second construction since it only uses $O(n)$ ancillae. Here, for the Boolean-based circuit, we employ the first construction as we already used $O(n^{1/(1-2^{-d})})$ ancillae. We note that computing $q$ ($8$ $\mathsf{GT}$ gates) can be done in parallel to computing plus copying/uncopying $r$ and performing the $d$-th $\mathsf{QRAM}$ circuit ($4+2+2 = 8$ $\mathsf{GT}$ gates), so the $d$-th level-$\mathsf{QRAM}$ uses $8$ $\mathsf{GT}$ gates. In total, we use $8d - 6$ $\mathsf{GT}$ gates. 
\end{proof}

% \begin{remark}
%     It is possible to reduce the number of $\mathsf{GT}$ gates by trading in extra Fan-Outs. The $q$ and $r$ computation and copying uses $\poly\log{n}$ Fan-Outs, thus $\mathsf{QRAM}$ can use $2d$ $\mathsf{GT}$ gates plus $\poly\log{n}$ Fan-Outs.
% \end{remark}

\section{Acknowledgement}
We thank Rainer Dumke, Yvonne Gao, Wenhui Li, and Daniel Weiss for useful discussions on physical implementations of $\QRAM$, Arthur Rattew for interesting conversations on the feasibility of $\QRAM$s, Patrick Rebentrost for initial discussions, Sathyawageeswar Subramanian for Ref.~\cite{kumar2023tight}, and Shengyu Zhang for general discussions and for clarifying some points in~\cite{STY-asymptotically,yuan2023optimal}. This research is supported by the National Research Foundation, Singapore and A*STAR under its CQT Bridging Grant and its Quantum Engineering Programme under grant NRF2021-QEP2-02-P05. This work was done in part while JFD, AL, and MS were visiting the Simons Institute for the Theory of Computing.

\DeclareRobustCommand{\DE}[2]{#2}
\bibliographystyle{alphaurl}
\bibliography{biblio.bib}
%\printbibliography

\end{document}